\documentclass{article}

\usepackage{fullpage}
\usepackage[bottom]{footmisc}
\usepackage{amsmath}
\usepackage[utf8]{inputenc}
\usepackage[numbers]{natbib}
\usepackage{amsmath,amssymb}
\usepackage{bbm}

\DeclareMathOperator*{\argmax}{arg\,max}
\DeclareMathOperator*{\argmin}{arg\,min}
\usepackage[ruled]{algorithm2e}
 
\usepackage{hyperref}
\usepackage{cleveref}

\usepackage{amsthm}
\newtheorem{theorem}{Theorem}[section]
\newtheorem{lemma}[theorem]{Lemma}
\newtheorem{claim}[theorem]{Claim}

\newtheorem{definition}[theorem]{Definition}
\newtheorem{corollary}[theorem]{Corollary}

\newtheorem{observation}[theorem]{Observation}

\newtheorem{conjecture}[theorem]{Conjecture}

\usepackage{fullpage}

\usepackage[pdftex]{graphicx}

\usepackage{comment}
\usepackage{mathtools}

\usepackage[shortlabels]{enumitem}

\usepackage[shortlabels]{enumitem}

\newcommand{\ignore}[1]{}

\newcommand{\hide}[1]{}

\usepackage[shortlabels]{enumitem}

\usepackage[shortlabels]{enumitem}

\newcommand{\A}{\mathcal{A}}
\newcommand{\B}{\mathcal{B}}

\newcommand{\M}{\mathcal{M}}

\usepackage{caption}
\usepackage{xcolor}

\def\Rset{\mathbb{R}}

\DeclareMathOperator*{\conv}{conv}

\DeclareMathOperator{\Reg}{\mathsf{Reg}}
\DeclareMathOperator{\diam}{diam}

\newcommand{\cA}{\mathcal{A}}

\newcommand{\cM}{\mathcal{M}}

\newcommand{\cO}{\mathcal{O}}

\newcommand{\cT}{\mathcal{T}}

\newcommand{\cX}{\mathcal{X}}

\newcommand{\bb}{{\mathbf b}}

\newcommand{\eps}{{\varepsilon}}

\DeclareMathOperator{\Prof}{\mathsf{Prof}}
\newcommand{\MB}{\M_{\mathrm{MB}}}
\DeclareMathOperator{\BR}{\mathsf{BR}}

\newcommand{\csp}{\varphi}

\newcommand{\Up}{U^{+}}
\newcommand{\Um}{U^{-}}

\newcommand{\Uzs}{U_{ZS}}
\newcommand{\ox}{\overline{x}}

\newcommand{\boundary}{T^{1-\tau}}
\newcommand{\regret}{\alpha}

\newcommand{\intersectionmb}{\mathcal{M}_{MI}}

\newcommand{\esh}[1]{}

\newcommand{\nat}[1]{}
\newcommand\jon[1]{}

\newcommand{\ftrl}{\textsf{FTRL}}
\DeclareMathOperator{\regbound}{\rho}
\newcommand{\wtA}{\widetilde{A}}
\newcommand{\opt}{\text{opt}}
\DeclareMathOperator{\sshift}{\zeta}
\usepackage{nicematrix}

\title{Pareto-Optimal Algorithms for Learning in Games}
\author{
Eshwar Ram Arunachaleswaran\thanks{University of Pennsylvania. 
Email: \texttt{eshwarram.arunachaleswaran@gmail.com} }
\and 
Natalie Collina\thanks{University of Pennsylvania.
Email: \texttt{ncollina@seas.upenn.edu} }
\and
Jon Schneider\thanks{Google Research. Email: \texttt{jschnei@google.com}}
}

\begin{document}

\maketitle
\abstract{We study the problem of characterizing optimal learning algorithms for playing repeated games against an adversary with unknown payoffs. In this problem, the first player (called the learner) commits to a learning algorithm against a second player (called the optimizer), and the optimizer best-responds by choosing the optimal dynamic strategy for their (unknown but well-defined) payoff. Classic learning algorithms (such as no-regret algorithms) provide some counterfactual guarantees for the learner, but might perform much more poorly than other learning algorithms against particular optimizer payoffs. 

In this paper, we introduce the notion of \emph{asymptotically Pareto-optimal learning algorithms}. Intuitively, if a learning algorithm is Pareto-optimal, then there is no other algorithm which performs asymptotically at least as well against all optimizers and performs strictly better (by at least $\Omega(T)$) against some optimizer. We show that well-known no-regret algorithms such as Multiplicative Weights and Follow The Regularized Leader are Pareto-dominated. However, while no-regret is not enough to ensure Pareto-optimality, we show that a strictly stronger property, no-swap-regret, is a sufficient condition for Pareto-optimality. 

Proving these results requires us to address various technical challenges specific to repeated play, including the fact that there is no simple characterization of how optimizers who are rational in the long-term best-respond against a learning algorithm over multiple rounds of play. To address this, we introduce the idea of the \emph{asymptotic menu} of a learning algorithm: the convex closure of all correlated distributions over strategy profiles that are asymptotically implementable by an adversary. Interestingly, we show that all no-swap-regret algorithms share the same asymptotic menu, implying that all no-swap-regret algorithms are ``strategically equivalent''.

}

\pagebreak
\tableofcontents

\section{Introduction}
\label{sec:introduction}
    
Consider an agent faced with the problem of playing a repeated game against another strategic agent. In the absence of complete information about the other agent’s goals and behavior, it is reasonable for the agent to employ a learning algorithm to decide how to play. This raises the (purposefully vague) question: What is the ``best’’ learning algorithm for learning in games?

One popular yardstick for measuring the quality of learning algorithms is \emph{regret}. The regret of a learning algorithm is the worst-case gap between the algorithm’s achieved utility and the counterfactual utility it would have received if it instead had played the optimal fixed strategy in hindsight. There exist learning algorithms which achieve sublinear $o(T)$ regret when played across $T$ rounds (\emph{no-regret algorithms}), and researchers now have a very good understanding of the strongest regret guarantees possible in a variety of different settings. It is tempting to conclude that one of these regret-minimizing algorithms is the optimal choice of learning algorithm for our agent.

However, many standard no-regret algorithms -- including popular algorithms such as Multiplicative Weights and Follow-The-Regularized-Leader -- have the unfortunate property that they are vulnerable to strategic manipulation \cite{deng2019strategizing}. What this means is that if one agent (a learner) is using a such an algorithm to play a repeated game against a second agent (an optimizer), there are games where the optimizer can exploit this by playing a time-varying dynamic strategy (e.g. playing some strategy for the first $T/2$ rounds, then switching to a different strategy in the last half of the rounds). By doing so, in some games the optimizer can obtain significantly more ($\Omega(T)$) utility than they could by playing a fixed static strategy, often at cost to the learner. This is perhaps most striking in the case of auctions, where~\cite{braverman2018selling} show that if a bidder uses such an algorithm to decide their bids, the auctioneer can design a dynamic auction that extracts the full welfare of the bidder as revenue (leaving the bidder with zero net utility). On the flip side, \cite{deng2019strategizing} show that if the learner employs a learning algorithm with a stronger counterfactual guarantee -- that of \emph{no-swap-regret} -- this protects the learner from strategic manipulation, and prevents the optimizer from doing anything significantly better than playing a static strategy for all $T$ rounds. Perhaps, then, a no-swap-regret algorithm is the ``best’’ learning algorithm for game-theoretic settings.

But even this is not the complete picture: even though strategic manipulation from the other agent \emph{may} harm the learner, there are other games where both the learner and optimizer can benefit from the learner playing a manipulable algorithm. Indeed,  \cite{guruganesh2024contracting} prove that there are contract-theoretic settings where both the learner and optimizer benefit from from the learner running a manipulable no-regret algorithm (with the optimizer best-responding to it). In light of these seemingly contradictory results, is there anything meaningful one can say about what learning algorithm a strategic agent should use? 

\subsection{Our results and techniques}

In this paper, we acknowledge that there may not be a consistent total ordering among learning algorithms, and instead study this question through the lens of \emph{Pareto-optimality}. Specifically, we consider the following setting. As before, one agent (the learner) is repeatedly playing a general-sum normal-form game $G$ against a second agent (the optimizer). The learner knows their own utility $u_L$ for outcomes of this game but is uncertain of the optimizer’s utility $u_O$, and so commits to playing according to a learning algorithm $\cA$ (a procedure which decides the learner’s action at round $t$ as a function of the observed history of play of both parties). The optimizer observes this and plays a (potentially dynamic) best-response to $\cA$ that maximizes their own utility $u_O$. We remark that in addition to capturing the strategic settings mentioned above, this asymmetry also models settings where one of the participants in a repeated game (such as a market designer or large corporation) must publish their algorithms up front and has to play against a large collection of unknown optimizers.

We say one learning algorithm $\cA$ \emph{(asymptotically) Pareto-dominates} a second learning algorithm $\cA'$ for the learner in this game if: i. for any utility function the optimizer may have, the learner receives at least as much utility (up to sublinear $o(T)$ factors) under committing to $\cA$ as they do under committing to $\cA'$, and ii. there exists at least one utility function for the optimizer where the learner receives significantly more utility (at least $\Omega(T)$ more) by committing to $\cA$ instead of committing to $\cA'$. A learning algorithm $\cA$ which is not Pareto-dominated by any learning algorithm is \emph{Pareto-optimal}. 

We prove the following results about Pareto-domination of learning algorithms for games.

\begin{itemize}
\item First, our notion of Pareto-optimality is non-vacuous: there exist many learning algorithms (including many no-regret learning algorithms) which are Pareto-dominated. In fact, we can show that there exist large classes of games where any instantiation of the Follow-The-Regularized-Leader (FTRL) with a strongly convex regularizer is Pareto-dominated. This set of learning algorithms contains many of the most popular no-regret learning algorithms as special cases (e.g. Hedge, Multiplicative Weights, and Follow-The-Perturbed-Leader).

\item In contrast to this, any no-swap-regret algorithm is Pareto-optimal. This strengthens the case for the strategic power of no-swap-regret learning algorithms in repeated interactions. 

\item That said, the Pareto-domination hierarchy of algorithms is indeed not a total order: there exist infinitely many Pareto-optimal learning algorithms that are qualitatively different in significant ways. And of the learning algorithms that are Pareto-dominated, they are not all Pareto-dominated by the same Pareto-optimal algorithm (indeed, in many cases FTRL is not dominated by a no-swap-regret learning algorithm, but by a different Pareto-optimal learning algorithm). 
\end{itemize}

In addition to this, we also provide a partial characterization of all no-regret Pareto-optimal algorithms, which we employ to prove the above results. In order to understand this characterization, we need to introduce the notion of the \emph{asymptotic menu} of a learning algorithm. 

To motivate this concept, consider a transcript of the repeated game $G$. If the learner has $m$ actions to choose from each round and the optimizer has $n$ actions, then after playing for $T$ rounds, we can describe the average outcome of play via a \emph{correlated strategy profile (CSP)}: a correlated distribution over the $mn$ pairs of learner/optimizer actions. The important observation is that this correlated strategy profile (an $mn$-dimensional object) is all that is necessary to understand the average utilities of both players, regardless of their specific payoff functions -- it is in some sense a ``sufficient statistic’’ for all utility-theoretic properties of the transcript. 

Inspired by this, we define the \emph{asymptotic menu} $\cM(\cA)$ of a learning algorithm $\cA$ to be the convex closure of the set of all CSPs that are asymptotically implementable by an optimizer against a learner who is running algorithm $\cA$. \nat{would an example be helpful here? or maybe a link to an example in the appendix?} That is, a specific correlated strategy profile $\csp$ belongs to $\cM(\cA)$ if the optimizer can construct arbitrarily long transcripts by playing against $\cA$ whose associated CSPs are arbitrarily close to $\csp$. We call this a ``menu'' since we can think of this as a set of choices the learner offers the optimizer by committing to algorithm $\cA$ (essentially saying, ``pick whichever CSP in this set you prefer the most''). 

Working with asymptotic menus allows us to translate statements about learning algorithms (complex, ill-behaved objects) to statements about convex subsets of the $mn$-simplex (much nicer mathematical objects). In particular, our notion of Pareto-dominance translates directly from algorithms to menus, as do concepts like ``no-regret'' and ``no-swap-regret''. This allows us to prove the following results about asymptotic menus:

\begin{itemize}
\item First, by applying Blackwell's Approachability Theorem, we give a simple and complete characterization of which convex subsets of the $mn$-simplex $\Delta_{mn}$ are valid asymptotic menus: any set $\M$ with the property that for any optimizer action $y \in \Delta_{n}$, there exists a learner action $x \in \Delta_{m}$ such that the product distribution $x\otimes y$ belongs to $\M$ (Theorem \ref{thm:characterization}).
\item We then use this characterization to show that there is a \emph{unique} no-swap-regret menu, which we call $\M_{NSR}$ (Theorem \ref{thm:unique_nsr_menu}), can be described explicitly as a polytope, and which is contained as a subset of \emph{any} no-regret menu (Lemma \ref{lem:nsr_within_all_nr}). In particular, this implies that all no-swap-regret algorithms share the same asymptotic menu, and hence are \emph{strategically equivalent} from the perspective of an optimizer strategizing against them. This is notably not the case for no-regret algorithms, which  have many different asymptotic menus.
\item In our main result, we give a characterization of all Pareto-optimal no-regret menus for menus that are polytopes (the intersection of finitely many half-spaces). We show that such a menu $\M$ is Pareto-optimal iff the set of points in $\M$ which minimize  the learner's utility are the same as that for the no-swap-regret menu $\M_{NSR}$ (Theorem \ref{thm:pareto-optimal-characterization}). It is here where our geometric view of menus is particularly useful: it allows us to reduce this general question to a non-trivial property of two-dimensional convex curves (Lemma \ref{lem:geometry}).
\item As an immediate consequence of this, we show the no-swap-regret menu (and hence any no-swap-regret algorithm) is Pareto-optimal (Corollary \ref{cor:NSR_pareto_optimal}), and that there exist infinitely many distinct Pareto-optimal menus (each of which can be formed by starting with the no-swap-regret menu $\M_{NSR}$ and expanding it to include additional no-regret CSPs).
\item Finally, we demonstrate instances where the asymptotic menu of FTRL is Pareto-dominated. This would follow nearly immediately from the characterization above (in fact, for the even larger class of \emph{mean-based} no-regret algorithms), but for the restriction that the above characterization only applies to polytopal menus. To handle this, we also find a class of examples where we can prove that the asymptotic menu of FTRL is a polytope. Doing this involves developing new tools for optimizing against mean-based algorithms, and may be of independent interest.
\end{itemize}

\subsection{Takeaways and future directions}

What do these results imply about our original question? Can we say anything new about which learning algorithms a learner should use to play repeated games? From a very pessimistic point of view, the wealth of Pareto-optimal algorithms means that we cannot confidently say that any specific algorithm is the ``best’’ algorithm for learning in games. But more optimistically, our results clearly highlight no-swap-regret learning algorithms as a particularly fundamental class of learning algorithms in strategic settings (in a way generic no-regret algorithms are not), with the no-swap-regret menu being the \emph{minimal} Pareto-optimal menu among all no-regret Pareto-optimal menus. 

We would also argue that our results do have concrete implications for how one should think about designing new learning algorithms in game-theoretic settings. In particular, they suggest that instead of directly designing a learning algorithm via minimizing specific regret notions (which can lead to learning algorithms which are Pareto-dominated, in the case of common no-regret algorithms), it may be more fruitful to first design the specific asymptotic menu we wish the algorithm to converge to (and only then worry about the rate at which algorithms approach this menu). Our characterization of Pareto-optimal no-regret menus provides a framework under which we can do this: start with the no-swap-regret menu, and expand it to contain other CSPs that we believe may be helpful for the learner. For example, consider a learner who believes that the optimizer has a specific utility function $u_O$, but still wants to run a no-regret learning algorithm to hedge against the possibility that they do not. This learner can use our characterization to first find the best such asymptotic menu, and then construct an efficient learning algorithm that approaches it (via e.g. the Blackwell approachability technique of Theorem \ref{thm:characterization}).

There are a number of interesting future directions to explore. Most obvious is the question of extending our characterization of Pareto-optimality from polytopal no-regret menus to all asymptotic menus. While we conjecture the polytopal constraint is unnecessary, there do exist non-trivial high-regret Pareto-optimal menus (Theorem \ref{thm:high-regret-po}), and understanding the full class of such menus is an interesting open problem.

Secondly, throughout this discussion we have taken the perspective of a learner who is aware of their own payoff $u_L$ and only has uncertainty about the optimizer they face. Yet one feature of most common learning algorithms is that they do not even require this knowledge about $u_L$, and are designed to work in a setting where they learn their own utilities over time. Some of our results (such as the Pareto-optimality of no-swap-regret algorithms) carry over near identically to such utility-agnostic settings (see Appendix \ref{sec:game_agnostic}), but we still lack a clear understanding of Pareto-domination there.

Finally, we focus entirely on normal-form two-player games. But many practical applications of learning algorithms take place in more general strategic settings, such as Bayesian games or extensive-form games. What is the correct analogue of asymptotic menus and Pareto-optimality for these settings? 

\subsection{Related work}

There is a long history of work in both economics and computer science of understanding the interplay between game theory and learning. We refer the reader to any of \cite{cesa2006prediction, young2004strategic, fudenberg1998theory} for an introduction to the area. Much of the recent work in this area is focused on understanding the correlated equilibria that arise when several learning algorithms play against each other, and designing algorithms which approach this set of equilibria more quickly or more stably (e.g., \cite{syrgkanis2015fast,Anagnostides21:NearOptimal,Anagnostides22:Last-Iterate,Farina22:NearOptimal,piliouras2022beyond, zhang2023computing}). It may be helpful to compare the learning-theoretic characterization of the set of correlated equilibria (which contains all CSPs that can be implemented by having several no-swap-regret algorithms play against each other) to our definition of asymptotic menu -- in some ways, one can think of an asymptotic menu as a one-sided, algorithm-specific variant of this idea. 

Our paper is most closely connected to a growing area of work on understanding the strategic manipulability of learning algorithms in games. \cite{braverman2018selling} was one of the first works to investigate these questions, specifically for the setting of non-truthful auctions with a single buyer. Since then, similar phenomena have been studied in a variety of economic settings, including other auction settings \cite{deng2019prior, cai2023selling, kolumbus2022and, kolumbus2022auctions}, contract design \cite{guruganesh2024contracting}, Bayesian persuasion \cite{chen2023persuading}, general games \cite{deng2019strategizing, brown2023is}, and Bayesian games \cite{MMSSbayesian}. \cite{deng2019strategizing} and \cite{MMSSbayesian} show that no-swap-regret is a necessary and sufficient condition to prevent the optimizer from benefiting by manipulating the learner. \cite{brown2023is} introduce an asymmetric generalization of correlated and coarse-correlated equilibria which they use to understand when learners are incentivized to commit to playing certain classes of learning algorithms. Our no-regret and no-swap-regret menus can be interpreted as the sets of $(\emptyset, \mathcal{E})$-equilibria and $(\emptyset, \mathcal{I})$-equilibria in their model (their definition of equilibria stops short of being able to express the asymptotic menu of a specific learning algorithm, however). In constructing an example where the asymptotic menu of FTRL is a polytope, we borrow an example from \cite{guruganesh2024contracting}, who present families of principal-agent problems which are particularly nice to analyze from the perspective of manipulating mean-based agents.

Our results highlight no-swap-regret algorithms as particularly relevant algorithms for learning in games. The first no-swap-regret algorithms were provided by \cite{foster1997calibrated}, who also showed their dynamics converge to correlated equilibria. Since then, several authors have designed learning algorithms for minimizing swap regret in games \cite{hart2000simple, blum2007external, peng2023fast, dagan2023external}. Our work shows that all these algorithms are in some ``strategically equivalent’’ up to sublinear factors; this is perhaps surprising given that many of these algorithms are qualitatively quite different (especially the very recent swap regret algorithms of \cite{peng2023fast} and \cite{dagan2023external}).

Finally, although we phrase our results from the perspective of learning in games, it is equally valid to think of this work as studying a Stackelberg variant of a repeated, finite-horizon game, where one player must commit to a repeated strategy without being fully aware of the other player’s utility function. In the full-information setting (where the learner is aware of the optimizer's payoff), the computational aspects of this problem are well-understood \cite{conitzer2006computing, Peng_Shen_Tang_Zuo_2019, collina2023efficient}. In the unknown-payoff setting, preexisting work has focused on learning the optimal single-round Stackelberg distribution by playing repeatedly against a myopic~\cite{balcan15, Janusz2012, lauffer2022} or discounting~\cite{haghtalab2022learning} follower. As far as we are aware, we are the first to study this problem in the unknown-payoff setting with a fully non-myopic follower.

\section{Model and preliminaries}

We consider a setting where two players, an \emph{optimizer} $O$ and a \emph{learner} $L$, repeatedly play a two-player bimatrix game $G$ for $T$ rounds. The game $G$ has $m$ actions for the optimizer and $n$ actions for the learner, and is specified by two bounded payoff functions $u_{O}: [m] \times [n] \rightarrow [-1, 1]$ (denoting the payoff for the optimizer) and $u_{L}: [m] \times [n] \rightarrow [-1, 1]$ (denoting the payoff for the learner). During each round $t$, the optimizer picks a mixed strategy $x_t \in \Delta_{m}$ while the learner simultaneously picks a mixed strategy $y_{t} \in \Delta_{n}$; the learner then receives reward $u_{L}(x_t, y_t)$ and the optimizer receives reward $u_{O}(x_t, y_t)$ (where here we have linearly extended $u_{L}$ and $u_{O}$ to take domain $\Delta_{m} \times \Delta_{n}$). Both the learner and optimizer observe the full mixed strategy of the other player (the ``full-information'' setting).

True to their name, the learner will employ a \emph{learning algorithm} $\cA$ to decide how to play. For our purposes, a learning algorithm is a family of horizon-dependent algorithms $\{A^T\}_{T \in \mathbb{N}}$. Each $A^{T}$ describes the algorithm the learner follows for a fixed time horizon $T$. Each horizon-dependent algorithm is a mapping from the history of play to the next round's action, denoted by a  collection of $T$ functions $A^T_1, A^T_2 \cdots A^T_T$, each of which deterministically map the transcript of play (up to the corresponding round) to a mixed strategy to be used in the next round, i.e., $A^T_t(x_1, x_2,\cdots, x_{t-1}) = y_t$. 



We assume that the learner is able to see $u_{L}$ before committing to their algorithm $\cA$, but not $u_O$. The optimizer, who knows $u_L$ and $u_{O}$, will approximately best-respond by selecting a sequence of actions that approximately (up to sublinear $o(T)$ factors) maximizes their payoff. They break ties in the learner's favor. Formally, for each $T$ let

\begin{equation*}
    V_{O}(\cA, u_O, T) = \sup_{(x_1, \dots, x_T) \in \Delta_{m}^{T}} \frac{1}{T}\sum_{t=1}^{T} u_{O}(x_t, y_t)
\end{equation*}

\noindent
represent the maximum per-round utility of the optimizer with payoff $u_O$ playing against $\cA$ in a $T$ round game (here and throughout, each $y_t$ is determined by running $A^{T}_{t}$ on the prefix $x_1$ through $x_{t-1}$). For any $\eps > 0$, let 

\begin{equation*}
\cX(\cA, u_O, T, \eps) = \left\{(x_1, x_2, \dots, x_T) \in \Delta_{m}^T \bigm| \frac{1}{T}\sum_{t=1}^{T} u_{O}(x_t, y_t) \geq V_{O}(\cA, u_O, T) - \eps \right\}
\end{equation*}

\noindent
be the set of $\eps$-approximate best-responses for the optimizer to the algorithm $\cA$. Finally, let

\begin{equation*}
V_{L}(\cA, u_O, T, \eps) = \sup_{(x_1, \dots, x_T) \in \cX(\cA, u_O, T, \eps)} \frac{1}{T}\sum_{t=1}^{T} u_{L}(x_t, y_t)
\end{equation*}

\noindent
represent the maximum per-round utility of the learner under any of these approximate best-responses. 

We are concerned with the asymptotic per-round payoff of the learner as $T \rightarrow \infty$ and $\eps \rightarrow 0$. Specifically, let

\begin{equation}\label{eq:learner_value}
V_{L}(\cA, u_O) = \lim_{\eps \rightarrow 0} \liminf_{T \rightarrow \infty} V_{L}(\cA, u_O, T, \eps).
\end{equation}

\noindent
Note that the outer limit in \eqref{eq:learner_value} is well-defined since for each $T$, $V_{L}(A, u_O, T, \eps)$ is decreasing in $\eps$ (being a supremum over a smaller set). 

The learner would like to select a learning algorithm $\cA$ that is ``good'' regardless of what the optimizer payoffs $u_O$ are. In particular, the learner would like to choose a learning algorithm that is \emph{asymptotically Pareto-optimal} in the following sense.

\begin{definition}[Asymptotic Pareto-dominance for learning algorithms]
Given a fixed $u_L$, A learning algorithm $\cA'$ \emph{asymptotically Pareto-dominates} a learning algorithm $\cA$ if for all optimizer payoffs $u_O$, $V_{L}(\cA', u_O) \geq V_{L}(\cA, u_O)$, and for a positive measure set of optimizer payoffs $u_O$, $V_{L}(\cA', u_O) > V_{L}(\cA, u_O)$. A learning algorithm $\cA$ is \emph{asymptotically Pareto-optimal} if it is not asymptotically Pareto-dominated by any learning algorithm.
\end{definition}







\paragraph{Classes of learning algorithms.} We will be interested in three specific classes of learning algorithms: \emph{no-regret algorithms}, \emph{no-swap-regret algorithms}, and \emph{mean-based algorithms} (along with their subclass of \emph{FTRL algorithms}).

A learning algorithm $\cA$ is a \emph{no-regret algorithm} if it is the case that, regardless of the sequence of actions $(x_1, x_2, \dots, x_T)$ taken by the optimizer, the learner's utility satisfies:

\begin{equation*}
\sum_{t=1}^{T} u_{L}(x_t, y_t) \geq \left(\max_{y^* \in [n]} \sum_{t=1}^{T} u_{L}(x_t, y^*)\right) - o(T).
\end{equation*}

A learning algorithm $\cA$ is a \emph{no-swap-regret algorithm} if it is the case that, regardless of the sequence of actions $(x_1, x_2, \dots, x_T)$ taken by the optimizer, the learner's utility satisfies:

\begin{equation*}
\sum_{t=1}^{T} u_{L}(x_t, y_t) \geq \max_{\pi: [n] \rightarrow [n]} \sum_{t=1}^{T} u_{L}(x_t, \pi(y_t)) - o(T).
\end{equation*}

\noindent
Here the maximum is over all swap functions $\pi: [n] \rightarrow [n]$ (extended linearly to act on elements $y_t$ of $\Delta_n$). It is a fundamental result in the theory of online learning that both no-swap-regret algorithms and no-regret algorithms exist (see \cite{cesa2006prediction}).

Some no-regret algorithms have the property that each round, they approximately best-respond to the historical sequence of losses. Following \cite{braverman2018selling} and \cite{deng2019strategizing}, we call such algorithms \emph{mean-based algorithms}. Formally, we define mean-based algorithms as follows.

\begin{definition}\label{def:mean-based}
A learning algorithm $\cA$ is \emph{$\gamma(t)$-mean-based} if whenever $j, j' \in [m]$ satisfy

$$\frac{1}{t}\sum_{s=1}^{t} u_L(x_t, j') - 
\frac{1}{t}\sum_{s=1}^{t} u_L(x_s, j) \geq \gamma(t),$$

\noindent
then $y_{t, j} \leq \gamma(t)$ (i.e., if $j$ is at least $\gamma(t)$ worse than some other action $j'$ against the historical average action of the opponent, then the total probability weight on $j$ must be at most $\gamma(t)$). A learning algorithm is \emph{mean-based} if it is $\gamma(t)$-mean-based for some $\gamma(t) = o(1)$.
\end{definition}

Many standard no-regret learning algorithms are mean-based, including Multiplicative Weights, Hedge, Online Gradient Descent, and others (see \cite{braverman2018selling}). In fact, all of the aforementioned algorithms can be viewed as specific instantiations of the mean-based algorithm Follow-The-Regularized-Leader. It is this subclass of mean-based algorithms that we will eventually show is Pareto-dominated in Section \ref{sec:mb_algos_and_menus}; we define it below. 


\begin{definition}
\sloppy{$\ftrl_T(\eta,R)$ is the horizon-dependent algorithm for a given learning rate $\eta > 0$ and bounded, strongly convex regularizer $R: \Delta^n \rightarrow \mathbb{R}$ which picks action $y_t \in \Delta^n$ via $y_t = \arg \max_{y \in \Delta^n} \left(\sum_{s=1}^{t-1}u_{L}(x_s, y) - \frac{R(y)}{\eta} \right)$. A learning algorithm $\cA$ belongs to the family of learning algorithms $\ftrl$ if for all $T > 0$, the finite-horizon $A_T$ is of the form $\ftrl_T(\eta_T, R)$ for some sequence of learning rates $\eta_T$ with $\eta_T = 1/o(T)$ and fixed regularizer $R$.}
\end{definition}


As mentioned, the family $\ftrl$ contains many well-known algorithms. For instance, we can recover Multiplicative Weights with the negative entropy regularizer $R_T(y) = \sum_{i \in [n]} y_i \log y_i$, and Online Gradient Descent via the quadratic regularizer $R_T(y) = \frac{1}{2} ||y||_2^2$ (see \cite{hazan201210} for details).


\paragraph{Other game-theoretic preliminaries and assumptions.}

Fix a specific game $G$. For any mixed strategy $x$ of the optimizer, let $\BR_{L}(x) = \argmax_{y \in [n]} u_{L}(x, y)$ represent the set of best-responses to $x$ for the learner. Similarly, define $\BR_{O}(y) = \argmax_{x \in [m]} u_{O}(x, y)$. 

A \emph{correlated strategy profile (CSP)} $\csp$ is an element of $\Delta_{mn}$ and represents a correlated distribution over pairs $(i, j) \in [m] \times [n]$ of optimizer/learner actions. For each $i \in [m]$ and $j \in [n]$, $\csp_{ij}$ represents the probability that the optimizer plays $i$ and the learner plays $j$ under $\csp$. For mixed strategies $x \in \Delta_{m}$ and $y \in \Delta_{n}$, we will use tensor product notation $x \otimes y$ to denote the CSP corresponding to the product distribution of $x$ and $y$. We also extend the definitions of $u_L$ and $u_O$ to CSPs (via $u_L(\csp) = \sum_{i, j} \csp_{ij}u_L(i, j)$, and likewise for $u_O(\csp)$).

Throughout the rest of the paper, we will impose two constraints on the set of games $G$ we consider (really, on the learner payoffs $u_L$ we consider). These constraints serve the purpose of streamlining the technical exposition of our results, and both constraints only remove a measure-zero set of games from consideration. The first constraint is that we assume that over all possible strategy profiles, there is one which is uniquely optimal for the learner; i.e., a pair of moves $i^* \in [m]$ and $j^* \in [n]$ such that $u_L(i^*, j^*) > u_L(i, j)$ for any $(i, j) \neq (i^*, j^*)$.  Note that slightly perturbing the entries of any payoff $u_L$ causes this to be true with probability $1$. We let $\csp^{+} = (i^*) \otimes (j^*)$ denote the corresponding optimal CSP. 

Secondly, we assume that the learner has no weakly dominated actions. To define this, we say an action $y \in [n]$ for the learner is \emph{strictly dominated} if it is impossible for the optimizer to incentivize $y$; i.e., there doesn't exist any $x \in \Delta_{m}$ for which $y \in \BR_{L}(x)$. We say an action $y \in [n]$ for the learner is \emph{weakly dominated} if it is \emph{not} strictly dominated but it is impossible for the optimizer to \emph{uniquely} incentivize $y$; i.e., there doesn't exist any $x \in \Delta_{m}$ for which $\BR_{L}(x) = \{y\}$. Note that this is solely a constraint on $u_{L}$ (not on $u_O$) and that we still allow for the possibility of the learner having strictly dominated actions. Moreover, only a measure-zero subset of possible $u_L$ contain weakly dominated actions, since slightly perturbing the utilities of a weakly dominated action causes it to become either strictly dominated or non-dominated. This constraint allows us to remove some potential degeneracies (such as the learner having multiple copies of the same action) which in turn simplifies the statement of some results (e.g., Theorem \ref{thm:unique_nsr_menu}).

\paragraph{Game-agnostic learning algorithms.} Here we have defined learning algorithms as being associated with a fixed $u_L$ and being able to observe the optimizer's sequence of actions (if not their actual payoffs). However many natural learning algorithms (including Multiplicative Weights and FTRL) only require the counterfactual payoffs of each action from each round. In Appendix~\ref{sec:game_agnostic} we explore Pareto-optimality and Pareto-domination over this class of algorithms.

\section{From learning algorithms to menus}
\label{sec:menu_results}

\subsection{The asymptotic menu of a learning algorithm}

Our eventual goal is to understand which learning algorithms are Pareto-optimal for the learner. However, learning algorithms are fairly complex objects; instead, we will show that for our purposes we can associate each 
learning algorithm with a much simpler object we call an \emph{asymptotic menu}, which can be represented as a convex subset of $\Delta_{mn}$. Intuitively, the asymptotic menu of a learning algorithm describes the set of correlated strategy profiles an optimizer can asymptotically incentivize in the limit as $T$ approaches infinity. 

More formally, for a fixed horizon-dependent algorithm $A^{T}$, define the \emph{menu} $\M(A^T) \subseteq \Delta_{mn}$ of $A^{T}$ to be the convex hull of all CSPs of the form $\frac{1}{T}\sum_{t=1}^T x_t \otimes y_t$, where $(x_1, x_2, \dots, x_T)$ is any sequence of optimizer actions and $(y_1, y_2, \dots, y_T)$ is the response of the learner to this sequence under $A^T$ (i.e., $y_t = A^{T}_t(x_1, x_2, \dots, x_{t-1})$).

If a learning algorithm $\cA$ has the property that the sequence $\M(A^1), \M(A^2), \dots$ converges under the Hausdorff metric\footnote{The \emph{Hausdorff distance} between two bounded subsets $X$ and $Y$ of Euclidean space is given by $d_H(X, Y) = \max(\sup_{x \in X} d(x, Y), \sup_{y \in Y} d(y, X))$, where $d(a, B)$ is the minimum Euclidean distance between point $a$ and the set $B$.}, we say that the algorithm $\cA$ is \emph{consistent} and call this limit value the \emph{asymptotic menu} $\M(\cA)$ of $\cA$. More generally, we will say that a subset $\M \subseteq \Delta_{mn}$ is an asymptotic menu if it is the asymptotic menu of some consistent algorithm. It is possible to construct learning algorithms that are not consistent (for example, imagine an algorithm that runs multiplicative weights when $T$ is even, and always plays action $1$ when $T$ is odd); however even in this case we can find subsequences of time horizons where this converges and define a reasonable notion of asymptotic menu for such algorithms. We defer discussion of this to Appendix \ref{sec:inconsistent}, and otherwise will only concern ourselves with consistent algorithms. See also Appendix \ref{sec:examples} for some explicit examples of asymptotic menus.

The above definition of asymptotic menu allows us to recast the Stackelberg game played by the learner and optimizer in more geometric terms. Given some $u_L$, the learner begins by picking a valid asymptotic menu $\cM$. The optimizer then picks a point $\csp$ on $\cM$ that maximizes $u_{O}(\csp)$ (breaking ties in favor of the learner). The optimizer and the learner then receive utility $u_O(\csp)$ and $u_L(\csp)$ respectively.

For any asymptotic menu $\M$, define $V_{L}(\M, u_O)$ to be the utility the learner ends up with under this process. Specifically, define $V_{L}(\M, u_O) = \max \{ u_{L}(\csp) \mid \csp \in \argmax_{\csp \in \M} u_{O}(\csp)\}$. We can verify that this definition is compatible with our previous definition of $V_{L}$ as a function of the learning algorithm $\cA$ (see Appendix \ref{app:omitted_proofs} for proof).

\begin{lemma}\label{lem:menu_alg_equiv}
For any learning algorithm $\cA$, $V_{L}(\M(\cA), u_O) = V_{L}(\cA, u_O)$.
\end{lemma}



As a consequence of Lemma \ref{lem:menu_alg_equiv}, instead of working with asymptotic Pareto-dominance of learning algorithms, we can entirely work with Pareto-dominance of \emph{asymptotic menus}, defined as follows.

\begin{definition}[Pareto-dominance for asymptotic menus]
Fix a payoff $u_L$ for the learner. An asymptotic menu $\cM'$ \emph{Pareto-dominates} an asymptotic menu $\cM$ if for all optimizer payoffs $u_O$, $V_{L}(\cM, u_O) \geq V_{L}(\cM', u_O)$, and for at least one\footnote{Alternatively, we can ask that one menu strictly beats the other on a positive measure set of payoffs. This may seem more robust, but turns out to be equivalent to the single-point definition. We prove this in Appendix \ref{app:pareto-dominance-eq} (in particular, see Theorem \ref{thm:pareto-equiv}). Note that by Lemma \ref{lem:menu_alg_equiv}, this implies a similar equivalence for Pareto-domination of algorithms.} payoff $u_O$, $V_{L}(\cM, u_O) > V_{L}(\cM', u_O)$. An asymptotic menu $\cM$ is \emph{Pareto-optimal} if it is not Pareto-dominated by any asymptotic menu. 
\end{definition}

\subsection{Characterizing possible asymptotic menus}

Before we address the harder question of which asymptotic menus are Pareto-optimal, it is natural to wonder which asymptotic menus are even possible: that is, which convex subsets of $\Delta_{mn}$ are even attainable as asymptotic menus of some learning algorithm. In this section we provide a complete characterization of all possible asymptotic menus, which we describe below.

\begin{theorem}\label{thm:characterization}
A closed, convex subset $\M \subseteq \Delta_{mn}$ is an asymptotic menu iff for every $x \in \Delta_{m}$, there exists a $y \in \Delta_{n}$ such that $x \otimes y \in \M$. 
\end{theorem}

The necessity condition of Theorem \ref{thm:characterization} follows quite straightforwardly from the observation that if the optimizer only ever plays a fixed mixed strategy $x \in \Delta_{m}$, the resulting average CSP will be of the form $x \otimes y$ for some $y \in \Delta_{n}$. The trickier part is proving sufficiency. For this, we will need to rely on the following two lemmas.

The first lemma applies Blackwell approachability to show that any $\M$ of the form specified in Theorem \ref{thm:characterization} must \emph{contain} a valid asymptotic menu.

\begin{lemma}\label{lem:char_blackwell}
Assume the closed convex set $\M \subseteq \Delta_{mn}$ has the property that for every $x \in \Delta_{m}$, there exists a $y \in \Delta_{n}$ such that $x \otimes y \in \M$. Then there exists an asymptotic menu $\M' \subseteq \M$.
\end{lemma}
\begin{proof}

We will show the existence of an algorithm $\cA$ for which $\cM(\cA) \subseteq \M$. To do so, we will apply the Blackwell Approachability Theorem (\cite{blackwell1956analog}).

Consider the repeated vector-valued game in which the learner chooses a distribution $y_t \in \Delta_n$ over their $n$ actions, the optimizer chooses a distribution $x_t \in \Delta_m$ over their $m$ actions, and the learner receives the vector-valued, bilinear payoff $u(x_t, y_t) = x_t \otimes y_t$ (i.e., the CSP corresponding to this round). The Blackwell Approachability Theorem states that if the set $\cM$ is response-satisfiable w.r.t. $u$ -- that is, for all $x \in \Delta_m$, there exists a $y \in \Delta_n$ such that $u(x, y) \in \cM$ -- then there exists a learning algorithm $\cA$ such that

$$\lim_{T \rightarrow \infty} d\left(\frac{1}{T}\sum_{t=1}^{T}u(x_t, y_t), \cM\right) = 0,$$

\noindent
for any sequence of optimizer actions $\{x_t\}$ (here $d(p, S)$ represents the minimal Euclidean distance from point $p$ to the set $S$). In words, the history-averaged CSP of play must approach to the set $\cM$ as the time horizon grows. Since any $\csp \in \cM(\cA)$ can be written as the limit of such history-averaged payoffs (as $T \rightarrow \infty$), this would imply $\cM(\cA) \subseteq \cM$.

Therefore all that remains is to prove that $\M$ is response-satisfiable. But this is exactly the property we assumed $\M$ to have, and therefore our proof is complete.  
\end{proof}

The second lemma shows that asymptotic menus are \emph{upwards closed}: if $\cM$ is an asymptotic menu, then so is any convex set containing it.

\begin{lemma}\label{lem:char_upwards_closed}

If $\cM$ is an asymptotic menu, then any closed convex set $\cM'$ satisfying $\cM \subseteq \cM' \subseteq \Delta_{mn}$ is an asymptotic menu. 
\end{lemma}
\begin{proof}[Proof Sketch]
We defer the details of the proof to Appendix \ref{app:omitted_proofs} and provide a high-level sketch here. Since $\cM$ is an asymptotic menu, we know there exists a learning algorithm $\cA$ with $\cM(\cA) = \cM$. We show how to take $\cA$ and transform it to a learning algorithm $\cA'$ with $\cM(\cA') = \cM'$. The algorithm $\cA'$ works as follows:

\begin{enumerate}
\item At the beginning, the optimizer selects a point $\csp \in \cM'$ they want to converge to. They also agree on a ``schedule'' of moves $(x_t, y_t)$ for both players to play whose history-average converges to the point $\csp$ without ever leaving $\cM'$. (The optimizer can communicate this to the learner solely through the actions they take in some sublinear prefix of the game -- see the full proof for details). 
\item The learner and optimizer then follow this schedule of moves (the learner playing $x_t$ and the optimizer playing $y_t$ at round $t$). If the optimizer never defects, they converge to the point $\csp$.
\item If the optimizer ever defects from their sequence of play, the learner switches to playing the original algorithm $\cA$. In the remainder of the rounds, the time-averaged CSP is guaranteed to converge to some point  $\csp_{\mathrm{suff}} = \cM(\cA) = \cM$. Since the time-averaged CSP of the prefix $\csp_{\mathrm{pre}}$ lies in $\cM'$, the overall time-averaged CSP will still lie in $\cM'$, so the optimizer cannot incentivize any point outside of $\cM'$.
\end{enumerate}
\end{proof}

Combining Lemmas \ref{lem:char_blackwell} and Lemmas \ref{lem:char_upwards_closed}, we can now prove Theorem \ref{thm:characterization}.

\begin{proof}[Proof of Theorem \ref{thm:characterization}]
As mentioned earlier, the necessity condition is straightforward: assume for contradiction that there exists an algorithm $\cA$ with asymptotic menu $\cM$ such that, for some $x \in \Delta_{m}$, there is no point in $\cM$ of the form $x \otimes y$ for any $y$. Then, let the optimizer play $x$ in each round. The resulting CSP induced against $\cA$ must be of the form $x \otimes y$ for some $y \in \Delta_n$, deriving a contradiction. 

Now we will prove that if a set $\cM$ has the property that $\forall x \in \Delta_{m}$, there exists a $y \in \Delta_{n}$ such that $x \otimes y \in \M$, then it is a valid menu. To see this, consider any set $\cM$ with this property. Then by Lemma~\ref{lem:char_blackwell} there exists a valid menu $\cM' \subseteq \cM$. Then, by the upwards-closedness property of Lemma~\ref{lem:char_upwards_closed}, the set $\cM \supseteq \cM'$ is also a menu.

\end{proof}

\subsection{No-regret and no-swap-regret menus}
\label{sec:NSR}

Another nice property of working with asymptotic menus is that no-regret and no-swap-regret properties of algorithms translate directly to similar properties on these algorithms' asymptotic menus (the situation for mean-based algorithms is a little bit more complex, and we discuss it in Section~\ref{sec:mb_algos_and_menus}).

To elaborate, say that the CSP $\csp$ is \emph{no-regret} if it satisfies the no-regret constraint

\begin{equation}\label{eq:no-regret-constraint}
\sum_{i\in[m]}\sum_{j\in[n]} \csp_{ij}u_{L}(i, j) \geq \max_{j^{*} \in [n]}\sum_{i\in[m]}\sum_{j\in[n]} \csp_{ij}u_{L}(i, j^*).
\end{equation}

\noindent
Similarly, say that the CSP $\csp$ is \emph{no-swap-regret} if, for each $j \in [n]$, it satisfies

\begin{equation}\label{eq:no-swap-regret-constraint}
\sum_{i \in [m]} \csp_{ij}u_{L}(i, j) \geq \max_{j^* \in [n]} \sum_{i \in [m]} \csp_{ij}u_{L}(i, j^*).
\end{equation}

For a fixed $u_L$, we will define the \emph{no-regret menu} $\M_{NR}$ to be the convex hull of all no-regret CSPs, and the \emph{no-swap-regret menu} $\M_{NSR}$ to be the convex hull of all no-swap-regret CSPs. In the following theorem we show that the asymptotic menu of any no-(swap-)regret algorithm is contained in the no-(swap-)regret menu.

\begin{theorem}
\label{thm:nr_nsr_containment}
If a learning algorithm $\cA$ is no-regret, then for every $u_L$, $\M(\cA) \subseteq \M_{NR}$. If $\cA$ is no-swap-regret, then for every $u_L$, $\M(\cA) \subseteq \M_{NSR}$.
\end{theorem}

Note that both $\M_{NR}$ and $\M_{NSR}$ themselves are valid asymptotic menus, since for any $x \in \Delta_{m}$, they will contain some point of the form $x \otimes y$ for some $y \in \BR_{L}(x)$. In fact, we can say something much stronger about the no-swap-regret menu: it is exactly the convex hull of all such points.

\begin{lemma}\label{lem:nsr_characterization}
The no-swap-regret menu $\M_{NSR}$ is the convex hull of all CSPs of the form $x \otimes y$, with $x \in \Delta_{m}$ and $y \in \BR_{L}(x)$.
\end{lemma}
\begin{proof}
First, note that every CSP of the form $x \otimes y$, with $x \in \Delta_{m}$ and $y \in \BR_{L}(x)$, is contained in $\M_{NSR}$. This follows directly follows from the fact that this CSP satisfies the no-swap-regret constraint \eqref{eq:no-swap-regret-constraint}, since no action can be a better response than $y$ to $x$.

For the other direction, consider a CSP $\csp \in \M_{NSR}$. We will rewrite $\csp$ as a convex combination of product CSPs of the above form. For each pure strategy $a \in [n]$ for the learner, let $\beta(a) \in \Delta_{m}$ represent the conditional mixed strategy of the optimizer corresponding to $X$ given that the learner plays action $a$, i.e. $\beta_j(a) =\frac{\csp_{ja}}{\sum_{k \in [m]} \csp_{ka}}$ for all $j \in [m]$ (setting $\beta_k(a)$ arbitrarily if all values $\csp_{ka}$ are zero). With this, we can write $\csp = \sum_{a \in [m]} (\sum_{k \in [m]} \csp_{ka}) (\beta(a) \otimes a)$. 

Now, note that if $a \notin \argmax_{b} u_L(\beta(a), b)$, this would violate the no-swap-regret constraint \eqref{eq:no-swap-regret-constraint} for $j = a$. Thus, we have rewritten $\csp$ as a convex combination of CSPs of the desired form, completing the proof.
\end{proof}


One key consequence of this characterization is that it allows us to show that the asymptotic menu of \emph{any} no-regret algorithm must contain the no-swap-regret menu $\M_{NSR}$ as a subset. Intuitively, this is since every no-regret menu should also contain every CSP of the form $x \otimes y$ with $y \in \BR_L(x)$, since if the optimizer only plays $x$, the learner should learn to best-respond with $y$ (although some care needs to be taken with ties).

\begin{lemma}
\label{lem:nsr_within_all_nr}
    For any no-regret algorithm $\A$, $\M_{NSR} \subseteq \M(\A)$.
\end{lemma}





This fact allows us to prove our first main result: that all consistent\footnote{Actually, as a consequence of this result, it is possible to show that any no-swap-regret algorithm must be consistent: see Corollary \ref{cor:all_nsr_algos} for details.} no-swap-regret algorithms have the same asymptotic menu (namely, $\M_{NSR}$). 

\begin{theorem}\label{thm:unique_nsr_menu}
If $\cA$ is a no-swap-regret algorithm, then $\M(\cA) = \M_{NSR}$.
\end{theorem}
\begin{proof}
From Theorem~\ref{thm:nr_nsr_containment}, $\M(A) \subseteq \M_{NSR}$. However, since any no-swap-regret algorithm also has no-regret, Lemma \ref{lem:nsr_within_all_nr} implies $\M_{NSR} \subseteq \M(A)$. The conclusion follows.
\end{proof}

Note that in the proof of Theorem \ref{thm:unique_nsr_menu}, we have appealed to Lemma~\ref{lem:nsr_within_all_nr} which uses the fact that $u_L$ has no weakly dominated actions. This is necessary: consider, for example, a game with two identical actions for the learner, $a$ and $a'$ ($u_L(\cdot, a) = u_L(\cdot, a')$). We can consider two no-swap-regret algorithms for the learner, one which only plays $a$ and never plays $a'$, and the other which only plays $a'$ and never plays $a$. These two algorithms will have different asymptotic menus, both of which contain only no-swap-regret CSPs. But as mentioned earlier, this is in some sense a degeneracy -- the set of learner payoffs $u_L$ with weakly dominated actions has zero measure (any small perturbation to $u_L$ will prevent this from taking place). 

Theorem \ref{thm:unique_nsr_menu} has a number of conceptual implications for thinking about learning algorithms in games:

\begin{enumerate}
    \item First, all no-swap-regret algorithms are \emph{asymptotically equivalent}, in the sense that regardless of which no-swap-regret algorithm you run, any asymptotic strategy profile you converge to under one algorithm you could also converge to under another algorithm (for appropriate play of the other player). This is true even when the no-swap-regret algorithms appear qualitatively quite different in terms of the strategies they choose (compare e.g. the fixed-point based algorithm of \cite{blum2007external} with the more recent algorithms of \cite{dagan2023external} and \cite{peng2023fast}).
    \item In particular, there is no notion of regret that is meaningfully \emph{stronger} than no-swap-regret for learning in (standard, normal-form) games. That is, there is no regret-guarantee you can feasibly insist on that would rule out some points of the no-swap-regret menu while remaining no-regret in the standard sense. In other words, the no-swap-regret menu is \emph{minimal} among all no-regret menus: every no-regret menu contains $\M_{NSR}$, and no asymptotic menu (whether it is no-regret or not) is a subset of $\M_{NSR}$. 
    \item Finally, these claims are \emph{not} generally true for external regret. There are different no-regret algorithms with very different asymptotic menus (as a concrete example, $\M_{NR}$ and $\M_{NSR}$ are often different, and they are both asymptotic menus of some learning algorithm by Theorem \ref{thm:characterization}). 
\end{enumerate}

Of course, this does not tell us whether it is actually \emph{good} for the learner to use a no-swap-regret algorithm, from the point of view of the learner's utility. In the next section we will revisit this question through the lens of understanding which menus are Pareto optimal.

\section{Characterizing Pareto-optimal menus}
\label{sec:pareto-optimal-menus}

In this section we shift our attention to understanding which asymptotic menus are Pareto-optimal and which are Pareto-dominated by other asymptotic menus. The ideal result would be a characterization of all Pareto-optimal asymptotic menus; we will stop a little short of this and instead provide a full characterization of all Pareto-optimal \emph{no-regret} menus that are also \emph{polytopal} -- i.e., can be written as the intersection of a finite number of half-spaces. This characterization will be sufficient for proving our main results that the no-swap-regret menu $\M_{NSR}$ is Pareto-optimal, but that the menu corresponding to multiplicative weights is sometimes Pareto-dominated. 

Before we introduce the characterization, we introduce a little bit of additional notation. For any menu $\M$, let $\Up(\M) = \max_{\csp \in \M} u_{L}(\csp)$ denote the maximum learner payoff of any CSP in $\M$; likewise, define $\Um(\M) = \min_{\csp \in \M} u_{L}(\csp)$. We will also let $\M^{+} = \argmax_{\csp \in \M} u_{L}(\csp)$ and $\M^{-} = \argmin_{\csp \in \M} u_{L}(\csp)$ be the subsets of $\M$ that attain this maximum and minimum (we will call these the \emph{maximum-value} and \emph{minimum-value sets} of $\M$).

Our characterization can now be simply stated as follows.

\begin{theorem}\label{thm:pareto-optimal-characterization}
Let $\M$ be a polytopal no-regret menu. Then $\M$ is Pareto-optimal iff $\M^{-} = \M^{-}_{NSR}$. That is, $\M$ must share the same minimum-value set as the no-swap-regret menu $\M_{NSR}$.
\end{theorem}

Note that while this characterization only allows us to reason about the Pareto-optimality of polytopal no-regret menus, in stating that these menus are Pareto-optimal, we are comparing them to all possible asymptotic menus. That is, we show that they are not Pareto-dominated by any possible asymptotic menu, even one which may have high regret and/or be an arbitrary convex set. We conjecture that this characterization holds for all no-regret menus (even ones that are not polytopal).

The remainder of this section will be dedicated to proving Theorem \ref{thm:pareto-optimal-characterization}. We will begin in Section \ref{sec:learner-utilities} by establishing some basic properties about $\M^{+}$, $\M^{-}$, $\Um(\M)$, and $\Up(\M)$ for no-regret and Pareto-optimal menus. Then in Section \ref{sec:minimal-lemma} we prove our main technical lemma (Lemma \ref{lem:non-dominance}), which shows that a menu cannot be Pareto-dominated by a menu with a larger minimal set. Finally, we complete the proof of Theorem \ref{thm:pareto-optimal-characterization} in Section \ref{sec:completing-the-proof}, and discuss some implications for the Pareto-optimality of the no-regret and no-swap-regret menus in Section \ref{sec:regret-implications}.



\subsection{Constraints on learner utilities}
\label{sec:learner-utilities}

We begin with some simple observations on the possible utilities of the learner under Pareto-optimal menus and no-regret menus. We first consider $\M^{+}$. Recall that (by assumption) there is a unique pure strategy profile $\csp^{+} = (i^{*}) \otimes (j^{*})$ that maximizes the learner's reward. We claim that any Pareto-optimal menu must contain $\csp^{+}$.

\begin{lemma}\label{lem:has_max_point}
If $\M$ is a Pareto-optimal asymptotic menu, then $\M^{+} = \{\csp^{+}\}$.
\end{lemma}
\begin{proof}
Assume $\M$ is a Pareto-optimal asymptotic menu that does not contain $\csp^{+}$. By Lemma \ref{lem:char_upwards_closed}, the set $\M' = \conv(\M, \csp^{+})$ is also a valid asympotic menu. We claim $\M'$ Pareto-dominates $\M$.

To see this, first note that when $u_O = u_L$, $V_L(\M', u_O) = u_L(\csp^{+}) > V_{L}(\M, u_O)$, since $\csp^{+}$ is the unique CSP in $\Delta_{mn}$ maximizing $u_L$. On other hand, for any other $u_O$, the maximizer of $u_O$ over $\M'$ is either equal to the maximizer of $u_O$ over $\M$, or equal to $\csp^{+}$. In either case, the learner's utility is at least as large, so $V_L(\M', u_O) \geq V_L(\M, u_O)$ for all $u_O$. It follows that $\M'$ Pareto-dominates $\M$.  
\end{proof}

Note also that $\csp^{+}$ belongs to $\M_{NSR}$ (since it a best-response CSP of the same form as in Lemma \ref{lem:nsr_characterization}), so $\M^{+}_{NSR} = \csp^{+}$. Since $\M_{NSR}$ is also contained in every no-regret menu, this also means that for any (not necessarily Pareto-optimal) no-regret menu $\M$, $\M^{+} = \M^{+}_{NSR} = \csp^{+}$.

We now consider the minimum-value set $\M^{-}$. Unlike for $\M^{+}$, it is no longer the case that all Pareto-optimal menus share the same set $\M^{-}$. It is not even the case (as we shall see in Section \ref{sec:completing-the-proof}), that all Pareto-optimal menus have the same minimum learner utility $\Um(\M)$. 


However, it is the case that all \emph{no-regret} algorithms share the same value for the minimum learner utility $\Um(\M)$, namely the ``zero-sum'' utility $\Uzs = \min_{x \in \Delta_{m}} \max_{y \in \Delta_{n}} u_L(x, y)$. The reason for this is that $\Uzs$ is the largest utility the learner can guarantee when playing a zero-sum game (i.e., when the optimizer has payoffs $u_{O} = -u_{L}$), and thus it is impossible to obtain a higher value of $\Um(\M)$. This is formalized in the following lemma.

\begin{lemma}\label{lem:menu-min-value}
Every asymptotic menu must have $\Um(\M) \leq \Uzs$. Moreover, if $\M$ is a no-regret asymptotic menu, then $\Um(\M) = \Uzs$, and $\M_{NSR}^{-} \subseteq \M^{-}$.
\end{lemma}
\begin{proof}
Let $(x_{ZS}, y_{ZS})$ be the solution to the minimax problem $\min_{x \in \Delta_{m}} \max_{y \in \Delta_{n}} u_L(x, y)$ (i.e., the Nash equilibrium of the corresponding zero-sum game). By Theorem \ref{thm:characterization}, any asymptotic menu $\M$ must contain a point of the form $x_{ZS} \otimes y$. By construction, $u_L(x_{ZS} \otimes y) \leq U_{ZS}$, so $\Um(\M) \leq U_{ZS}$.

To see that every no-regret asymptotic menu satisfies $\Um(\M) = \Uzs$, assume that $\M$ is a no-regret menu, and $\csp \in \M$ satisfies $u_L(\csp) < \Uzs$. Since $\csp$ has no-regret (satisfies the conditions of \eqref{eq:no-regret-constraint}), we must also have $u_L(\csp) \geq \max_{y \in \Delta_{n}} \min_{x \in \Delta_{m}} u_L(x, y)$, since this holds for whatever marginal distribution $x$ is played by the optimizer under $\csp$. But by the minimax theorem, $\max_{y \in \Delta_{n}} \min_{x \in \Delta_{m}} u_L(x, y) = U_{ZS}$, and so we have a contradiction.

Finally, note that since $\M$ is no-regret, $\M_{NSR} \subseteq \M$ and so $\M_{NSR}^{-} \subseteq \M^{-}$ (since they share the same minimum value). 
\end{proof}

\subsection{Pareto-domination and minimum-value sets}
\label{sec:minimal-lemma}



We now present our two main lemmas necessary for the proof of Theorem \ref{thm:pareto-optimal-characterization}. The first lemma shows that if one menu contains a point not present in the second menu (and both menus share the same maximum-value set), then the first menu cannot possibly Pareto-dominate the second menu.
 
\begin{lemma}\label{lem:non-dominance}
Let $\M_1$ and $\M_2$ be two distinct asymptotic menus where $\M_{1}^{+} = \M_{2}^{+}$. Then if either: 

\begin{itemize}
    \item i. $\M_2 \setminus \M_1 \neq \emptyset$, or
    \item ii. $\M_{1}^{-} = \M_{2}^{-}$, 
\end{itemize}

\noindent
then there exists a $u_O$ for which $V_{L}(\M_1, u_O) > V_{L}(\M_2, u_O)$ (i.e., $\M_2$ does not Pareto-dominate $\M_1$).
\end{lemma}


Note that Lemma \ref{lem:non-dominance} holds also under the secondary assumption that $\M_{1}^{-} = \M_{2}^{-}$. One important consequence of this is that all menus with identical minimum value and maximum value sets $\M^{-}$ and $\M^{+}$ are incomparable to each other under the Pareto-dominance order (even such sets that may contain each other).

The key technical ingredient for proving Lemma ~\ref{lem:non-dominance} is the following lemma, which establishes a ``two-dimensional'' variant of the above claim.

\begin{lemma}\label{lem:geometry}
Let $f, g: [a, b] \rightarrow \mathbb{R}$ be two distinct concave functions satisfying $f(a) \leq g(a)$ and $f(b) = g(b)$. For $\theta \in [0, \pi]$, let $\hat{f}(\theta) = \argmax_{x \in [a, b]} (x\cos\theta + f(x)\sin\theta)$ (if the argmax is not unique, then $\hat{f}(\theta)$ is undefined). Define $\hat{g}(\theta)$ symmetrically. Then there exists a $\theta$ for which $\hat{f}(\theta) > \hat{g}(\theta)$. 
\end{lemma}
\begin{proof}

\begin{figure}[htbp]
    \centering
    \includegraphics[width=0.7\textwidth]{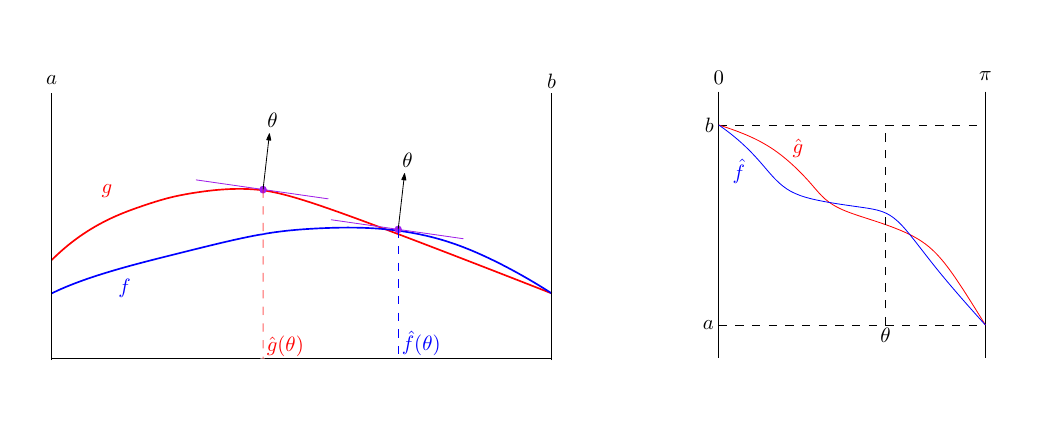}
    \caption{A visual depiction of Lemma~\ref{lem:geometry}. The purple points in the left figure denote the maximizers of $f$ and $g$ in the direction $\theta$; since the purple point on $f$ is to the right of that on $g$, we have $\hat{f}(\theta) > \hat{g}(\theta)$ for this $\theta$.}
    \label{fig:two_sets_and_functions}
\end{figure}

Since $f(x)$ is a concave curve, it has a (weakly) monotonically decreasing derivative $f'(x)$. This derivative is not necessarily defined for all $x \in [a, b]$, but since $f$ is concave it is defined almost everywhere. At points $c$ where it is not defined, $f$ still has a well-defined left derivative $f'_L(x) = \lim_{h \rightarrow 0} (f(x) - f(x - h))/h$ and right derivative $f'_R(x) = \lim_{h \rightarrow 0} (f(x + h) - f(x))/h$. We will abuse notation and let $f'(x)$ denote the interval $[f'_L(x), f'_R(x)]$ (at the boundaries defining $f'(a) = (-\infty, f_{R}(a)]$ and $f'(b) = [f_{L}(b), \infty)$. Similarly, the interval-valued inverse function $(f')^{-1}(y)$ is also well-defined, decreasing in $y$, and uniquely-defined for almost all values of $y$ in $(-\infty, \infty)$. 

Note that since $\hat{f}(\theta)$ is the $x$ coordinate of the point on the curve $f(x)$ that maximizes the inner product with the unit vector $(\cos\theta, \sin\theta)$, $f'(\hat{f}(\theta))$ must contain the value $-\cos\theta / \sin\theta = -\cot \theta$. In particular, if $\hat{f}(\theta)$ is uniquely defined, $\hat{f}(\theta) = (f')^{-1}(-\cot \theta)$. So it suffices to find a $y$ for which $(f')^{-1}(y) > (g')^{-1}(y)$. 

To do this, we make the following observation: since $f(b) - f(a) \geq g(b) - g(a)$, $\int_{a}^{b} f'(x)dx \geq \int_{a}^{b} g'(x)dx$~\footnote{These integrals are well defined because the first derivatives of $f$ and $g$ exist almost everywhere.}. This means there must be a point $c \in (a, b)$ where $f'(c) > g'(c)$. If not, then we must have $f'(x) \leq g'(x)$ for all $x \in (a, b)$; but the only way we can simultaneously have $f'(x) \leq g'(x)$ for all $x \in (a, b)$ and $\int_{a}^{b}f'(x)dx \geq \int_{a}^{b}g'(x)dx$ is if $f'(x) = g'(x)$ for almost all $x \in (a, b)$ -- but this would contradict the fact that $f$ and $g$ are distinct concave functions.

Now, take a point $c \in (a, b)$ where $f'(c) > g'(c)$ and choose a $y$ in $f'(c)$. Since $g'$ is a decreasing function, there must exist a $c' < c$ such that $y \in g'(c')$, and so $(f')^{-1}(y) > (g')^{-1}(y)$. 
\end{proof}

We can now prove Lemmas~\ref{lem:non-dominance} through an application of the above lemma. 

\begin{proof}[Proof of Lemma~\ref{lem:non-dominance}]
We will consider the two preconditions separately, and  begin by considering the case where $\M_1^{-} = \M_2^{-}$. Since $\M_1$ and $\M_2$ are distinct asymptotic menus, there must be an extreme point $\csp$ in one menu that does not belong to the other. In particular, there must exist an optimizer payoff $u_{O}$ where $u_{O}(\csp) > u_{O}(\csp')$ for any $\csp'$ in the other menu. Denote this specific optimizer payoff by $u_{O}^{0}$.

We will show that there exists a $u_{O} \in \mathrm{span}(u_L, u_{O}^{0})$ where $V_{L}(\M_1, u_O) > V_{L}(\M_2, u_O)$. To do this, we will project $\M_1$ and $\M_2$ to two-dimensional sets $S_1$ and $S_{2}$ by letting $S_{1} = \{(u_L(\csp), u_{O}^{0}(\csp)) \mid \csp \in \M_1 \}$ (and defining $S_{2}$ symmetrically). By our construction of $u_{O}^{0}$, these two convex sets $S_{1}$ and $S_{2}$ are distinct. Also, note that if $u_O = \lambda_{1}u_L + \lambda_{2}u_{O}^{0}$, we can interpret $V_{L}(\M_1, u_O)$ as the maximum value of $z_1$ for any point in $\arg\max_{(z_1, z_2) \in S_1} (\lambda_1z_1 + \lambda_2z_2)$. We can interpret $V_{L}(\M_2, u_O)$ similarly.

Let us now consider the geometry of $S_1$ and $S_2$. Let $u^{-}$ denote the common value of $\Um(\M_1)$ and $\Um(\M_2)$, and similarly, let $u^{+}$ denote the common value of $\Up(\M_1)$ and $\Up(\M_2)$. Both $S_1$ and $S_2$ are contained in the ``vertical'' strip $u^{-} \leq z_1 \leq u^{+}$. We can therefore write $S_1$ as the region between the concave curve $f_{\text{up}}:[u^{-}, u^{+}]$ (representing the upper convex hull of $S_1$) and $f_{\text{down}}$ (representing the lower convex hull of $S_1$; define $g_{\text{up}}$ and $g_{\text{down}}$ analogously for $S_2$. Since $S_1$ and $S_2$ are distinct, either $f_{\text{up}} \neq g_{\text{up}}$ or $f_{\text{down}} \neq g_{\text{down}}$; without loss of generality, assume $f_{\text{up}} \neq g_{\text{up}}$ (we can switch the upper and lower curves by changing $u_{O}^{0}$ to $-u_{O}^{0}$).

Note also that since $\M^{+}_1 = \M^{+}_2$ and $\M^{-}_1 \subseteq \M^{-}_2$, we have $f_{\text{up}}(u^{+}) = g_{\text{up}}(u^{+})$ and $f_{\text{up}}(u^{-}) \leq g_{\text{up}}(u^{-})$ (since $[f_{\text{down}}(u^{-}), f_{\text{up}}(u^{-})] \subseteq [g_{\text{down}}(u^{-}), g_{\text{up}}(u^{-})]$). By Lemma \ref{lem:geometry}, there exists a $\theta \in [0, \pi]$ for which $\hat{f}_{\text{up}}(\theta) > \hat{g}_{\text{up}}(\theta)$. But by the definition of $\hat{f}$ and $\hat{g}$, this implies that for $u_{O} = \cos(\theta) u_L + \sin(\theta) u^{0}_{O}$, $V_{L}(\M_1, u_{O}) > V_{L}(\M_2, u_O)$, as desired. This proof is visually depicted in Figure~\ref{fig:two_sets_and_functions}.

The remaining case, where $\M_2 \setminus \M_1 \neq \emptyset$, can be proved very similarly to the above proof. We make the following changes:

\begin{itemize}
    \item First, we choose an extreme point $\csp$ of $\M_2$ that belongs to $\M_2$ but not $\M_1$. Again, we choose a $u_O$ which separates $\csp$ from $\M_1$. We let $u^* = u_O(\csp)$; note that $u^* < u^{+}$ (since $\M_1^{+} = \M_2^{+}$). 
    \item Instead of defining our functions $f_{up}$ and $g_{up}$ on the full interval $[u^{-}, u^{+}]$, we instead restrict them to the interval $[u^{*}, u^{+}]$. Because of our construction of $\csp$, we have that $f_{up}(u^*) < g_{up}(u^*)$, and $f_{up}(u^{+}) = g_{up}(u^{+})$.
    \item We can again apply Lemma \ref{lem:geometry} to these two functions on this sub-interval, and construct a  $u_O$ for which $V_{L}(\M_1, u_O) > V_{L}(\M_2, u_O)$.
\end{itemize}
\end{proof}

One useful immediate corollary of Lemma \ref{lem:non-dominance} is that it is impossible for high-regret menus (menus that are not no-regret) to Pareto-dominate no-regret menus.

\begin{corollary}\label{cor:high-regret-non-dominance}
Let $\M_1$ and $\M_2$ be two asymptotic menus such that $\M_1$ is no-regret and $\M_2$ is \emph{not} no-regret. Then $\M_2$ does not Pareto-dominate $\M_1$.
\end{corollary}
\begin{proof}
If $\M_2$ does not contain $\csp^{+}$, add it to $\M_2$ via Lemma \ref{lem:has_max_point} (this only increases the position of $\M_2$ in the Pareto-dominance partial order). Since $\M_1$ is no-regret, it must already contain $\csp^{+}$, and therefore we can assume $\M^{+}_1 = \M^{+}_2 = \{\csp^{+}\}$.

Since $\M_2$ is not no-regret, it must contain a CSP $\csp$ that does not lie in $\M_{NR}$, and therefore $\M_2 \setminus \M_1 \neq \emptyset$. It then follows from Lemma \ref{lem:non-dominance} that $\M_2$ does not Pareto-dominate $\M_1$.
\end{proof}

\subsection{Completing the proof}
\label{sec:completing-the-proof}

We can now finish the proof of Theorem~\ref{thm:pareto-optimal-characterization}.
\begin{proof}[Proof of Theorem~\ref{thm:pareto-optimal-characterization}]
We will first prove that if a no-regret menu $\M$ satisfies $\M^{-} = \M_{NSR}^{-}$, then it is Pareto-optimal. To do so, we will consider any other menu $\M'$ and show that $\M'$ does not Pareto-dominate $\M$. There are three cases to consider:

\begin{itemize}
\item \textbf{Case 1}: $U^{-}(\M') < U^{-}(\M)$. In this case, $\M'$ cannot dominate $\M$ since $V_{L}(\M, -u_L) > V_{L}(\M', -u_L)$ (note that $U^{-}(\M) = V_{L}(\M, -u_L)$, since if $u_O = -u_L$, the optimizer picks the utility-minimizing point for the learner).

\item \textbf{Case 2}: $U^{-}(\M') > U^{-}(\M)$. By Lemma \ref{lem:menu-min-value}, this is not possible.

\item \textbf{Case 3}: $U^{-}(\M') = U^{-}(\M)$. 
If $\M'$ is not a no-regret menu, then by Corollary \ref{cor:high-regret-non-dominance} it cannot dominate $\M$. We will therefore assume that $\M'$ is a no-regret menu, i.e. $\M' \subseteq \M_{NR}$
Then, by Lemma~\ref{lem:menu-min-value}, $\M^{-} \subseteq (\M')^{-}$. Also, by Lemma \ref{lem:has_max_point}, we can assume without loss of generality that $(\M')^{+} = \{\csp^{+}\} = \M^{+}$ (if $\M'$ does not contain $\csp^{+}$, replace it with the Pareto-dominating menu that contains it). Now, by Lemma~\ref{lem:non-dominance}, $\M'$ does not dominate $\M$. 
\end{itemize}

We now must show that if $\M^{-} \neq \M_{NSR}^{-}$, then it is Pareto-dominated by some other menu. Since $\M$ is (by assumption) a no-regret menu, we must have $U^{-}(\M) = U^{-}(\M_{NSR}) = U_{ZS}$, and $\M^{-} \supset \M^{-}_{NSR}$ (Lemmas \ref{lem:menu-min-value}). Consider an extreme point $\csp_0$ that belongs to $\M^{-}$ but not to $\M_{NSR}^{-}$. Construct the menu $\M'$ as follows: it is the convex hull of $\M_{NSR}$ and all the extreme points in $\M$ \emph{except for} $\csp_0$. By Lemma \ref{lem:char_upwards_closed}, this is a valid menu (it is formed by adding some points to the valid menu $\M_{NSR}$). Note also that $\M'$ has all the same extreme points of $\M$ except for $\csp_0$ (since $\M$ is a polytope, we add a finite number of extreme points to $\M'$, all of which are well-separated from $\csp_0$), and in particular is distinct from $\M$.

We will show that $\M'$ Pareto-dominates $\M$. To see this, note first that, by Lemma~\ref{lem:non-dominance}, there is some $u_{O}$ such that $V_{L}(U_{\M},u_{O}) < V_{L}(U_{\M'},u_{O})$. Furthermore, for all other values of $u_{O}$, $V_{L}(U_{\M},u_{O}) \leq V_{L}(U_{\M'},u_{O})$. This is since the maximizer of $u_O$ over $\M$ is either the minimal-utility point $\csp_0$ (which cannot be strictly better than the maximizer of $u_O$ over $\M'$), or exactly the same point as the maximizer of $u_O$ over $\M'$. It follows that $\M'$ Pareto-dominates $\M$.
\end{proof}

Note that in the Proof of Theorem \ref{thm:pareto-optimal-characterization}, we only rely on the fact that the menu $\cM$ is polytopal in precisely one spot, when we construct a menu $\M'$ that Pareto-dominates $\M$ by ``removing'' an extreme point from $\M^{-}$. As stated, this removal operation requires $\M$ to be a polytope: in general, it is possible that any extreme point $\csp_0$ that belongs to $\M^{-}$ is a limit of other extreme points in $\M$, and so when attempting to construct $\M'$ per the procedure above, we would just perfectly recover the original $\M$ when taking the convex closure of the remaining points. 

That said, it is not clear whether the characterization of non-polytopal Pareto-optimal menus is any different than the characterization in Theorem \ref{thm:pareto-optimal-characterization}. In fact, by the argument in the proof of Theorem \ref{thm:pareto-optimal-characterization}, one direction of the characterization still holds (if a non-polytopal no-regret menu satisfies $\M^{-} = \M^{-}_{NSR}$, then it is Pareto-optimal). We conjecture that this characterization holds for non-polytopal menus (and leave it as an interesting open problem). 

\begin{conjecture}\label{conj:main}
Any no-regret menu $\M$ is Pareto-optimal iff $\M^{-} = \M_{NSR}$.
\end{conjecture}

On the other hand, the restriction to no-regret menus is necessary for the characterization of Theorem \ref{thm:pareto-optimal-characterization} to hold. To see this, note that another interesting corollary of Lemma \ref{lem:non-dominance} is that any minimal asymptotic menu is Pareto-optimal (in fact, we have the slightly stronger result stated below).

\begin{corollary}\label{cor:minimal_is_po}
Let $\M$ be an inclusion-minimal asymptotic menu (i.e., with the property that no other asymptotic menu $\M'$ satisfies $\M' \subset \M$). Then the menu $\M' = \conv(\M, \{\csp^{+}\})$ is Pareto-optimal. 
\end{corollary}



Corollary \ref{cor:minimal_is_po} allows us to construct some high-regret asymptotic menus that are Pareto-optimal. For example, we can show that the algorithm that always plays the learner's component of $\csp^{+}$ is Pareto-optimal.

\begin{theorem}\label{thm:high-regret-po}
Let $\cM$ be the asymptotic menu of the form $\cM = \{x \otimes j^* \mid x \in \Delta_{m}\}$ (where $j^*$ is the learner's component of $\csp^{+} = (i^*) \otimes (j^*)$). Then $\cM$ is Pareto-optimal.
\end{theorem}
\begin{proof}
By Theorem \ref{thm:characterization}, $\cM$ is inclusion-minimal. Since $\cM$ also includes the CSP $\csp^{+}$, it is Pareto-optimal by Corollary \ref{cor:minimal_is_po}.
\end{proof}

Note that in general, the menu $\M$ in Theorem \ref{thm:high-regret-po} is not no-regret, and may have $U^{-}(\M) < \Uzs$. We leave it as an interesting open question to provide a full characterization of all Pareto-optimal asymptotic menus.

\subsection{Implications for no-regret and no-swap-regret menus}
\label{sec:regret-implications}

Already Theorem \ref{thm:pareto-optimal-characterization} has a number of immediate consequences for understanding the no-regret menu $\M_{NR}$ and the no-swap-regret menu $\M_{NSR}$, both of which are polytopal no-regret menus by their definitions in \eqref{eq:no-regret-constraint} and \eqref{eq:no-swap-regret-constraint} respectively. As an immediate consequence of our characterization, we can see that the no-swap-regret menu (and hence any no-swap-regret learning algorithm) is Pareto-optimal.

\begin{corollary} \label{cor:NSR_pareto_optimal}
The no-swap-regret menu $\M_{NSR}$ is a Pareto-optimal asymptotic menu.
\end{corollary}

It would perhaps be ideal if $\M_{NSR}$ was the unique Pareto-optimal no-regret menu, as it would provide a somewhat clear answer as to which learning algorithm one should use in a repeated game. Unfortunately, this is not the case -- although $\M_{NSR}$ is the minimal Pareto-optimal no-regret menu, Theorem \ref{thm:pareto-optimal-characterization} implies there exist infinitely many distinct Pareto-optimal no-regret menus. 

On the more positive side, Theorem \ref{thm:pareto-optimal-characterization} (combined with Lemma \ref{lem:char_upwards_closed}) gives a recipe for how to construct a generic Pareto-optimal no-regret learning algorithm: start with a no-swap-regret learning algorithm (the menu $\M_{NSR}$) and augment it with any set of additional CSPs that the learner and optimizer can agree to reach. This can be any set of CSPs as long as i. each CSP $\csp$ has no regret, and ii. each CSP has learner utility $u_L(\csp)$ strictly larger than the minimax value $U_{ZS}$. 

\begin{corollary}
There exist infinitely many Pareto-optimal asymptotic menus.
\end{corollary}

Finally, perhaps the most interesting consequence of Theorem \ref{thm:pareto-optimal-characterization} is that, despite this apparent wealth of Pareto-optimal menus and learning algorithms, the no-regret menu $\M_{NR}$ is very often \emph{Pareto-dominated}. In particular, it is easy to find learner payoffs $u_L$ for which $\M^{-}_{NR} \neq \M^{-}_{NSR}$, as we show below.

\begin{corollary}
There exists a learner payoff\footnote{In fact, there exists a positive measure of such $u_L$. It is easy to adapt this proof to work for small perturbations of the given $u_L$, but this fact is also implied by Theorem \ref{thm:high-mb-swap-regret} (and its extension in Appendix \ref{app:proof_no_pos_neg_regret}), which proves a similar statement for the mean-based menu.} $u_L$ for which the no-regret $\M_{NR}$ is \emph{not} a Pareto-optimal asymptotic menu.
\end{corollary}
\begin{proof}
Take the learner's payoff from Rock-Paper-Scissors, where the learner and optimizer both have actions $\{a_1, a_2, a_3\}$, and $u_L(a_i, a_j) = 0$ if $j = i$, $1$ if $j = i + 1 \bmod{3}$, and $-1$ if $j = i - 1 \bmod{3}$. For this game, $\Uzs = 0$ (the learner can guarantee payoff $0$ by randomizing uniformly among their actions).

Now, note that the CSP $\csp = (1/3)(a_1 \otimes a_1) + (1/3)(a_2 \otimes a_2) + (1/3)(a_3 \otimes a_3)$ has the property that $u_L(\csp) = 0 = \Uzs$ and that $\csp \in \M_{NR}$, but also that $\csp \not\in \M_{NSR}$ (e.g. it is beneficial for the learner to switch from playing $a_1$ to $a_2$). Since $\M_{NR}$ is a polytopal no-regret menu, it follows from our characterization in Theorem \ref{thm:pareto-optimal-characterization} that $\M_{NR}$ is not Pareto-optimal.
\end{proof}



\section{Mean-based algorithms and menus}
\label{sec:mb_algos_and_menus}

In this section, we return to one of the main motivating questions of this work: are standard online learning algorithms (like multiplicative weights or follow-the-regularized-leader) Pareto-optimal? Specifically, are \emph{mean-based} no-regret learning algorithms, which always approximately best-respond to the historical sequence of observed losses, Pareto-optimal?

We will show that the answer to this question is \emph{no}: in particular, there exist payoffs $u_L$ where the menus of some mean-based algorithms (specifically, menus for multiplicative weights and FTRL) are not Pareto-optimal. Our characterization of Pareto-optimal no-regret menus in the previous section (Theorem \ref{thm:pareto-optimal-characterization}) does most of the heavy lifting here: it means that in order to show that a specific algorithm is not Pareto-optimal, we need only find a sequence of actions by the optimizer that both causes the learner to end up with the zero-sum utility $u_{ZS}$ and high swap-regret (i.e., at a point not belonging to $\M^{-}$). Such games (and corresponding trajectories of play by the optimizer) are relatively easy to find -- we will give one explicit example shortly that works for any mean-based algorithm.

However, there is a catch -- our characterization in Theorem \ref{thm:pareto-optimal-characterization} only applies to \emph{polytopal} menus (although we conjecture that it also applies to non-polytopal menus). So, in order to formally prove that a mean-based algorithm is not Pareto-optimal, we must additionally show that its corresponding menu is a polytope. 
Specifically, we give an example of a family of games where the asymptotic menus of all $\ftrl$ algorithms have a simple description as an explicit polytope which we can show is not Pareto-optimal.


\begin{theorem}
    \label{thm:ftrl_main}
    There exists a learner payoff $u_L$ with $m=2$ actions for the optimizer and $n=3$ actions for the learner for which all $\ftrl$ algorithms are Pareto-dominated~\footnote{This result also holds for any sufficiently small perturbation of $u_L$.}.
\end{theorem}

In Section~\ref{sec:mean-based-menu} we introduce the concept of the ``mean-based menu'', a menu of CSPs that is achievable against any mean-based algorithm, and introduce the family of games we study. We then (in Section \ref{sec:mb-is-polytope}) show that for this game, the mean-based menu is a polytope. Finally, in Section \ref{sec:mb-equals-mwu}, we prove that for sufficiently structured mean-based algorithms (including multiplicative weights update and FTRL) their asymptotic menu must equal this mean-based menu.

\subsection{The mean-based menu}
\label{sec:mean-based-menu}

To understand the behavior of mean-based algorithms, we will introduce an convex set that we call the \emph{mean-based menu}, consisting of exactly the correlated strategy profiles that are attainable against any mean-based learning algorithm. 

\sloppy{To define the mean-based menu, we will make use of a continuous-time reformulation of strategizing against mean-based learners presented in Section 7 of \cite{deng2019strategizing}. In this reformulation, we represent a trajectory of play by a finite non-empty list of triples $\tau = \{(x_1, t_1, b_1), (x_2, t_2, b_2), \dots, (x_k, t_k, b_k)\}$ where $x_{i} \in \Delta^m$, $t_i \in \mathbb{R}_{+}$, and $b_{i} \in [n]$. Each segment of this strategy intuitively represents a segment of time of length $t_i$ where the optimizer plays a mixed strategy $x_i$ and the learner responds with the pure strategy $b_i$.}

Let $\overline{x}_{i} = \sum_{j=1}^{i} (x_jt_j) / \sum_{j=1}^{i} t_j$ (with $\overline{x}_0 = 0$) be the time-weighted average strategy of the optimizer over the first $i$ triples. In order for this to be a valid trajectory of play for an optimizer playing against a mean-based learner, a trajectory must satisfy $b_{i} \in \BR_{L}(\overline{x}_{i-1})$ and $b_{i} \in \BR_{L}(\overline{x}_{i})$. That is, the pure strategy played by the learner during the $i$th segment must be a best-response to the the average strategy of the optimizer at the beginning and ends of this segment. We let $\cT$ represent the set of all valid trajectories (with an arbitrary finite number of segments). 

Given a trajectory $\tau$, we define $\Prof(\tau) \in \Delta_{mn}$ to equal the correlated strategy profile represented by this trajectory, namely:

$$\Prof(\tau) = \frac{\sum_{i=1}^{k} \tau_{i}(x_{i} \otimes b_{i})}{\sum_{i=1}^{k} \tau_{i}}.$$

Finally, we define the \textit{mean-based menu} $\MB$ to equal the convex closure of $\Prof(\tau)$ over all valid trajectories $\tau \in \cT$. Although $\MB$ is defined solely in terms of the continuous formulation above, it has the operational significance of being a subset of the asymptotic menu of any mean-based algorithm (the proof of Lemma \ref{lem:sub-mean-based} is similar to that of Theorem 9 in \cite{deng2019strategizing}, and is deferred to the appendix).

\begin{lemma}\label{lem:sub-mean-based}
    If $\cA$ is a mean-based learning algorithm with asymptotic menu $\M$, then $\MB \subseteq \M$. 
\end{lemma}


Note that unlike $\M_{NR}$ and $\M_{NSR}$, which contain all no-regret/no-swap-regret menus, Lemma \ref{lem:sub-mean-based} implies that $\M_{MB}$ is \emph{contained} in all mean-based menus. Since there exist mean-based algorithms that are no-regret (e.g., multiplicative weights), Lemma \ref{lem:sub-mean-based} immediately also implies that the mean-based menu is a subset of the no-regret menu (that is, $\MB \subseteq \M_{NR}$)\footnote{It is natural to ask whether $\MB$ is in general just equal to $\M_{NR}$. In Appendix \ref{app:mb_nr_gap}, we show this is not the case.}.

To construct our counterexample, we will specifically consider games with $m=2$ actions for the optimizer and $n=3$ actions for the learner, where the learner's payoff $u_L$ is specified by the following $3 \times 2$ matrix\footnote{Although we present all the results in this section for a specific learner's payoff $u_L$, they continue to hold for a positive-measure set of perturbations of this payoff function; see Appendix~\ref{app:mean_based_pareto_dominated} for details.}:

\begin{equation}\label{eq:game}
u_L = \begin{bmatrix}
  0 & 0 \\
 -1/6 & 1/3 \\
 -1/2 & 1/2
\end{bmatrix}.
\end{equation}

For convenience, we will label the learner's three actions (in order) $A$, $B$, and $C$, and the optimizer's two actions as $N$ and $Y$ (so e.g. $u_L(Y, B) = 1/3$, and $u_L(N, 0.5A + 0.5C) = -1/4$). As mentioned in the introduction, this example originates from an application in \cite{guruganesh2024contracting} for designing dynamic strategies for principal-agent problems when an agent runs a mean-based learning algorithm (the details of this application are otherwise unimportant). 

An interesting property of this choice of $u_L$ is that there exist points in $\MB^{-}$ not contained within $\M_{NSR}^{-}$.

\begin{theorem}\label{thm:high-mb-swap-regret}
For the choice of $u_L$ in \eqref{eq:game}, $\MB^{-} \not \subseteq \M_{NSR}^{-}$.
\end{theorem}
\begin{proof}
Consider the continuous-time trajectory $\tau = \{(N/3 + 2Y/3, 1/2, C), (N, 1/2, B)\}$. First, note that it is indeed a valid trajectory: we can check that $\ox_1 = N/3 + 2Y/3$, $\ox_2 = 2N/3 + Y/3$, $\BR_{L}(N/3 + 2Y/3) = \{B, C\}$, and that $\BR_{L}(2N/3 + Y/3) = \{A, B\}$. For this trajectory, $\Prof(\tau) = (1/6) (N \otimes C) + (1/3) (Y \otimes C) + (1/2) (N \otimes B)$.

For this game, $U_{ZS} = 0$ (the learner can guarantee non-negative utility by playing $Y$ and the optimizer can guarantee non-positive utility for the learner by playing $N$). Note that $u_L(\Prof(\tau)) = (1/6)(-1/2) + (1/3)(1/2) + (1/2)(-1/6) = 0$, so $\Prof(\tau)$ lies in $\MB^{-}$. But the swap regret of $\Prof(\tau)$ is positive, since by swapping all weight on $B$ to $A$ increases the learner's average utility by $1/12$, so $\Prof(\tau) \not\in \M_{NSR}$. 
\end{proof}

Combined with our characterization of Pareto-optimal no-regret menus (Theorem \ref{thm:pareto-optimal-characterization}), the above theorem implies that any no-regret mean-based algorithm with a polytopal menu is Pareto-dominated. In the next two sections, we will show that for the set of games we are considering a large class of mean-based algorithms have polytopal menus.

\subsection{The mean-based menu is a polytope}
\label{sec:mb-is-polytope}

Our first goal will be to show that the mean-based menu $\MB$ for this game is a polytope. Our main technique will be to show that trajectories $\tau$ whose profiles correspond to extreme points of $\MB$ must be simple and highly structured (e.g. have a small number of segments). We can then write down explicit systems of finitely many linear inequalities to characterize these profiles (implying that their convex closure $\MB$ must be a polytope). 


\begin{figure}[htbp]
    \centering
    \includegraphics[width=0.5\textwidth]{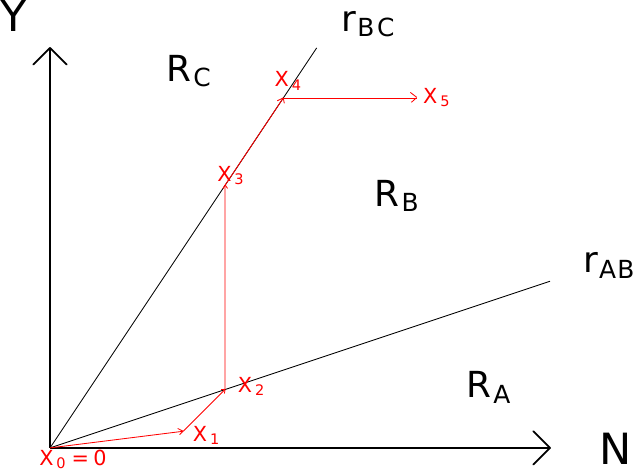}
    \caption{A simple trajectory with five segments.}
    \label{fig:mb_trajectory}
\end{figure}

Throughout the following arguments, it will be helpful to geometrically visualize a trajectory by thinking about how the cumulative strategy taken by the optimizer evolves over time. This starts at zero and traces out a polygonal line with breakpoints $X_{i} = \sum_{j=1}^{i} x_{j}t_{j}$, with segment $i$ of the trajectory taking place between breakpoints $X_{i-1}$ and $X_{i}$ (see Figure~\ref{fig:mb_trajectory}. Since the optimizer's strategy each round belongs to $\Delta_{2}$, this trajectory always stays in the positive quadrant $\Rset_{\geq 0}^2$, and strictly increases in $\ell_{1}$-norm over time. 


The best-response condition on $b_i$ can also be interpreted through this diagram. The quadrant can be divided into three cones $R_{A}, R_{B}, R_{C}$, where for any $j \in \{A, B, C\}$, $R_{j}$ contains the states $X$ where pure strategy $j$ for the learner is a best response to the normalized historical strategy $X / |X|_{1}$. Separating $R_{A}$ and $R_{B}$ is a ray $r_{AB}$ (where actions 1 and 2 are both best-responses to states on this ray); similarly, separating $R_{B}$ and $R_{C}$ is a ray $r_{BC}$. Here (and throughout this section) we use the term \emph{ray} to specifically refer to a set of points that lie on the same path through the origin (i.e., a set of states $X$ with the same normalized historical strategy $X/|X|_1$). 

We begin with several simplifications of trajectories that are not specific to the payoff matrix above, but hold for any game. We first show it suffices to consider trajectories where no two consecutive segments lie in the same cone.

\begin{lemma}\label{lem:switch_regions}
For any $\tau \in \cT$, there exists a trajectory $\tau' \in \cT$ such that $\Prof(\tau') = \Prof(\tau)$, and such that there are no two consecutive segments in $\tau'$ where the learner best-responds with the same action (i.e., $b_{i} \neq b_{i+1}$ for all $i$).
\end{lemma}
\begin{proof}
Consider a trajectory $\tau$ with two consecutive segments of the form $(x_i, t_i, b)$ and $(x_{i+1}, t_{i+1}, b)$ for some $b \in \{A, B, C\}$. Note that replacing these two segments with the single segment $((x_it_{i} + x_{i+1}t_{i+1})/(t_i + t_{i+1}), t_i + t_{i+1}, b)$ does not change the profile of the trajectory (e.g., in Figure \ref{fig:basic_trajectory}, we can average the first two segments into a single segment lying along $r_{AB}$). We can form $\tau'$ by repeating this process until no two consecutive $b_i$ are equal. 
\end{proof}

We similarly show that we cannot have three consecutive segments that lie along the same ray.

\begin{lemma}\label{lem:leave_ray}
We say a segment $(x_i, \tau_i, b_i)$ of a trajectory lies along a ray $r$ if both $X_{i-1}$ and $X_{i}$ lie on $r$. For any $\tau \in \cT$, there exists a trajectory $\tau' \in \cT$ such that $\Prof(\tau') = \Prof(\tau)$ and such that there are no three consecutive segments of $\tau'$ that lie along the same ray.
\end{lemma}
\begin{proof}
Assume $\tau$ contains a sequence of three segments $(x_i, t_i, b_i)$, $(x_{i+1}, t_{i+1}, b_{i+1})$, and $(x_{i+2}, t_{i+2}, b_{i+2})$ that all lie along the same ray $r$. This means that $X_{i}$, $X_{i+1}$, and $X_{i+2}$ are all non-zero points that lie on $r$, and therefore that $\ox_i = \ox_{i+1} = \ox_{i+2}$. This means that $b_{i}$, $b_{i+1}$, and $b_{i+2}$ must all belong to $\BR_L(\ox_i)$. Note also that this means that we can arbitrarily reorder these segments without affecting $\Prof(\tau)$.

However, $\BR_L(\ox_i)$ contains at most two actions for the agent (it is a singleton unless $\ox_i$ lies along $r_{AB}$ or $r_{BC}$, in which case it equals $\{A, B\}$ or $\{B, C\}$ respectively). This means that at least two of $\{b_{i}, b_{i+1}, b_{i+2}\}$ must be equal. By rearranging these segments so they are consecutive and applying the reduction in Lemma \ref{lem:switch_regions}, we can decrease the total number of segments in our trajectory while keeping $\Prof(\tau)$ invariant. We can repeatedly do this until we achieve the desired $\tau'$. 
\end{proof}


Finally we show that we only need to consider trajectories that end on a best-response boundary (or single segment trajectories).

\begin{lemma}\label{lem:end_on_br}
For any $\tau \in \cT$, we can write $\Prof(\tau)$ as a convex combination of profiles $\Prof(\tau')$ of trajectories $\tau'$ that either: i. consist of a single segment, or ii. terminate on a best-response boundary (i.e., $|\BR_L(\ox_k)| \geq 2$). 
\end{lemma}
\begin{proof}
We will proceed by induction on the number of segments in $\tau$. Assume $\tau$ is a trajectory with $k \geq 2$ segments that does not terminate on a best-response boundary (so $\BR_L(\ox_k) = \{b_k\}$). Let $(x_k, t_k, b_k)$ be the last segment of this trajectory, and let $\tau_{pre}$ be the trajectory formed by the first $k-1$ segments.

There are two cases. First, if $b_k \in \BR_L(x_k)$, then the single segment $(x_k, t_k, b_k)$ is itself a valid trajectory; call this trajectory $\tau_{post}$. Then $\Prof(\tau)$ is a convex combination of $\Prof(\tau_{pre})$ and $\Prof(\tau_{post})$, which both satisfy the inductive hypothesis.

Second, if $b_k \not\in \BR_L(x_k)$, then let $\tau(t)$ be the trajectory formed by replacing $t_k$ with $t$ (so the last segment becomes $(x_k, t, b_k)$). Since $b_k \not\in \BR_L(x_k)$, there is a maximal $t$ for which $\tau(t)$ is a valid trajectory (since for large enough $t$, $\ox_k$ will converge to $x_k$ and $b_k$ will no longer be a best response to $\ox_k$). Let $t_{k}^{\max}$ be this maximal $t$ and note that $\tau(t_{k}^{\max})$ has the property that it terminates on a best-response boundary. But now note that $\Prof(\tau)$ can be written as the convex combination of $\Prof(\tau(0))$ and $\Prof(\tau(t_k^{\max}))$ both of which satisfy the constraints of the lemma (the first by the inductive hypothesis, as it has one fewer segment, and the second because it terminates on a best-response boundary). The analysis of this case is shown in Figure~\ref{fig:convex1_trajectory}.
\end{proof}

\begin{figure}[htbp]
    \centering
    \includegraphics[width=0.5\textwidth]{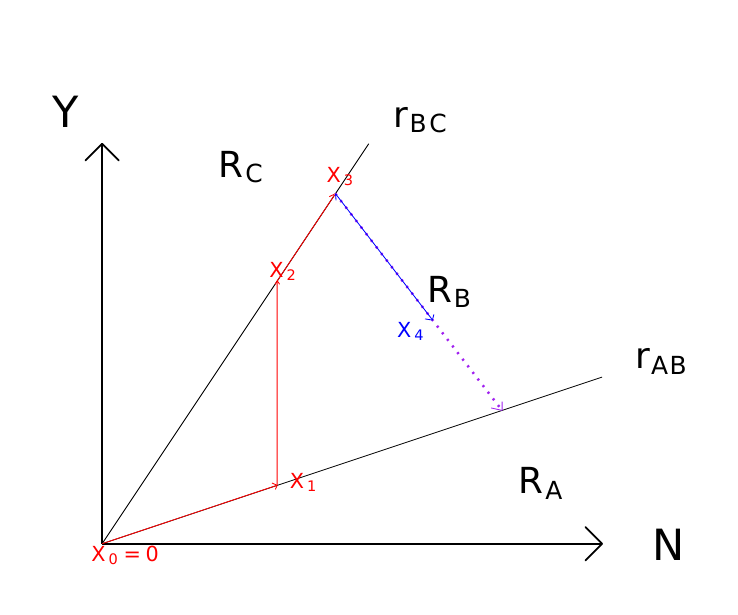}
    \caption{The trajectory ending at $X_4$ is drawn in red with the last segment in blue -- it can be seen as a convex combination of the trajectory upto $X_3$ (all in red) and the latter's extension marked by the purple segment.}
    \label{fig:convex1_trajectory}
\end{figure}

We will say a \textit{proper trajectory} is a trajectory that satisfies the conditions of Lemmas \ref{lem:switch_regions}, \ref{lem:leave_ray}, and \ref{lem:end_on_br}. A consequence of these lemmas is that  $\MB$ is also the convex closure of the profiles of all proper trajectories, so we can restrict our attention to proper trajectories from now on. 

Before we further simplify the class of proper trajectories, we introduce some additional notation to describe suffixes of valid trajectories. We define an \emph{offset trajectory} $(\tau, X_0)$ to be a pair containing a trajectory $\tau$ and a starting point $X_0 \in \Rset^{2}_{\geq 0}$. Similar to ordinary trajectories, in order for an offset trajectory to be valid, we must have that $b_{i} \in \BR(\overline{x}_{i-1}) \cap \BR(\overline{x}_{i})$; however, for offset trajectories, the definition of $\overline{x}_i$ changes to incorporate the starting point $X_{0}$ via

$$\overline{x}_i = \frac{X_{0} + \sum_{j=1}^{i}t_jx_{j}}{||X_{0}||_{1} + \sum_{j=1}^{i}t_{j}}.$$

\noindent
The profile $\Prof(\tau, X_0)$ of an offset trajectory is defined identically to $\Prof(\tau)$ (in particular, it only depends on the segments after $X_0$, and not directly on $X_0$). For offset trajectories, we define $X_i = X_0 + \sum_{j=1}^{i} x_{j}t_{j}$. We will commonly use offset trajectories to describe suffixes of proper trajectories -- for example, we can model the suffix of $\tau = \{(x_1, t_1, b_1), \dots, (x_k, t_k, b_k)\}$ starting at segment $j$ with the offset trajectory $(\tau', X_0')$, where $\tau'$ is the suffix of $\tau$ starting at segment $j$ and $X_0' = X_{j-1}$. 


If an offset trajectory of length $k$ has the property that $X_{k} = \lambda X_{0}$ for some $\lambda > 1$, we call this trajectory a \emph{spiral}. The following lemma shows that the profile of any spiral belongs to $\MB$.

\begin{figure}[htbp]
    \centering
    \includegraphics[width=0.5\textwidth]{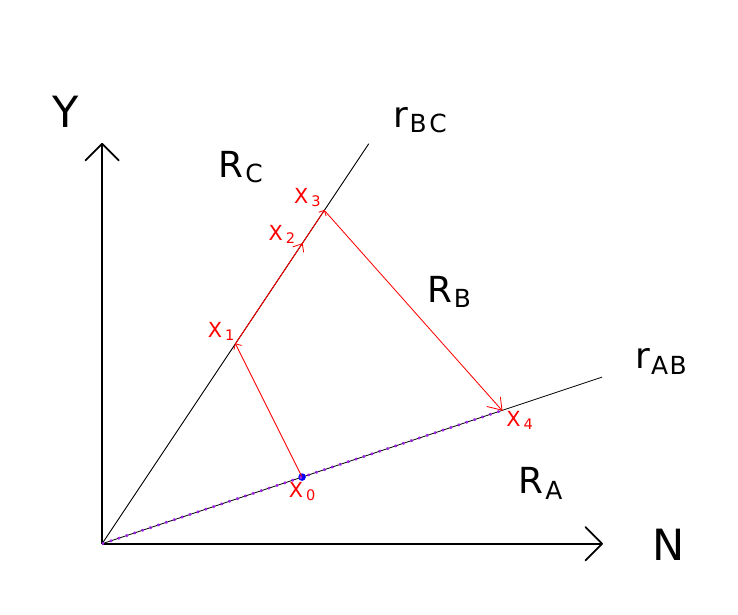}
    \caption{An offset trajectory (in red) that is also a spiral, with $X_4$ and $X_0$ on the same ray (along $r_{AB}$ and marked in purple) }
    \label{fig:offset_trajectory}
\end{figure}

\begin{lemma}\label{lem:spiral}
Let $\psi = (\tau, X_0)$ be a spiral. Then $\Prof(\psi) \in \MB$.
\end{lemma}
\begin{proof}
Let $x_0 = X_0 / ||X_0||_1 \in \Delta_2$, and choose an arbitrary $b_0 \in \BR_L(x_0)$. Also, let $\lambda = X_k/X_0$ be the scaling factor of $\psi$. Fix any positive integer $n \geq 1$, and consider the trajectory $\pi_n$ formed by concatenating the initial segment $(x_0, ||X_0||_1, b_0)$ with the offset trajectories $(\lambda^i\tau, \lambda^iX_0)$ for each $1 \leq i \leq n$ in that order. Note that this forms a valid trajectory since the $i$th offset trajectory ends at $\lambda^{i+1}X_0$, the starting point of the next trajectory.

Let $T_n$ be the total length of trajectory $\pi_n$. Note that $\Prof(\pi_n) = (||X_0||_1 (x_0 \otimes b_0) + (T_n - ||X_0||_1) \Prof(\tau))/T_n$. As $n \rightarrow \infty$, this approaches $\Prof(\tau)$, so $\Prof(\tau) \in \MB$ (recall $\MB$ is the convex closure of all profiles of trajectories).
\end{proof}

We say a trajectory or spiral is \emph{primitive} if it has no proper suffix that is a spiral. The following lemma shows that it suffices to consider primitive trajectories and spirals when computing $\MB$. 


\begin{lemma}\label{lem:primitive}
Let $\tau$ be a proper trajectory. Then $\Prof(\tau)$ can be written as a convex combination of the profiles of primitive trajectories and primitive spirals.
\end{lemma}
\begin{proof}
We induct on the length of the trajectory. All trajectories with one segment are primitive.

Consider a trajectory $\tau$ with length $k$. If $\tau$ is not primitive, some suffix of $\tau$ is a primitive spiral. Divide $\tau$ into its prefix trajectory $\tau_{pre}$ and this suffix spiral  $\psi_{suff}$. Now, $\Prof(\tau)$ is a convex combination (weighted by total duration) of $\Prof(\tau_{pre})$ and $\Prof(\psi_{suff})$. But $\psi_{suff}$ is a primitive cycle, and $\Prof(\tau_{pre})$ is the convex combination of profiles of primitive cycles/trajectories by the inductive hypothesis. This completes the proof.
\end{proof}

One consequence of Lemma \ref{lem:primitive} is that $\MB$ is the convex hull of the profiles of all primitive trajectories and cycles. It is here where we will finally use our specific choice of game (or more accurately, the fact that the learner only has three actions and hence there are only three regions in the state space). 

\begin{lemma}\label{lem:characterize_primitive}
Any primitive trajectory or primitive spiral for a game with $m = 2$ and $n = 3$ has at most three segments.
\end{lemma}
\begin{proof}
Consider a primitive trajectory $\tau$ with $k > 3$ segments. Since it is a proper trajectory, each of the cumulative points $X_1, X_2, \dots X_k$ must lie on a best-response boundary (one of the two rays $r_{AB}$ or $r_{BC}$). Since the trajectory is primitive, all of $X_1, X_2, \dots, X_{k-1}$ must lie on a different ray from $X_k$, and hence must all lie on the same ray. But by Lemma \ref{lem:leave_ray}, we can't have $X_0 = 0$, $X_1$, $X_2$, and $X_3$ all lie on the same ray. Therefore $k-1 \leq 2$, and $k \leq 3$. The argument for primitive spirals works in exactly the same way.
\end{proof}

It only remains to show that there are finitely many extremal profiles of primitive trajectories and spirals. To this end, call the sequence of pure actions $(b_1, b_2, \dots, b_k)$ taken by the learner in a given trajectory (or spiral) the \emph{fingerprint} of that trajectory (or spiral). 

\begin{lemma}\label{lem:fingerprint_lp}
Fix a sequence $\bb \in [n]^{k}$. Then the convex hull of $\Prof(\tau)$ over all trajectories with fingerprint $\bb$ is a polytope (the intersection of finitely many half-spaces). Similarly, the convex hull of $\Prof(\gamma)$ over all spirals with fingerprint $\bb$ is also a polytope.
\end{lemma}
\begin{proof}
First, note that by scaling any spiral or trajectory by a constant, we can normalize the total time $\sum_{i=1}^{k} t_i$ to equal $1$. If this is the case, then we can write $\Prof(\tau) = \sum_{i=1}^{k} (t_ix_i) \otimes b_i$. For a fixed fingerprint $\bb$, this is a linear transformation of the $k$ variables $y_i = t_ix_i$. We will show that the set of possible values for $(y_1, y_2, \dots, y_k)$ forms a polytope.

To do so, we will write an appropriate linear program. Our variables will be the time-weighted optimizer actions $y_1, \dots, y_k \in \Delta_{2}$, times $t_1, \dots, t_k$, and cumulative actions $X_1, X_2, \dots, X_k$ (where $X_i = y_1 + y_2 + \dots + y_i$). In the case of a spiral, $X_0$ will also be one of the variables (for trajectories, we can think of $X_0$ being fixed to $0$). The following finite set of linear constraints precisely characterizes all possible values of these variables for any valid trajectory or spiral with this fingerprint:

\begin{itemize}
    \item Time is normalized to 1: $t_1 + t_2 + \dots + t_k = 1$.
    \item Relation between $X_i$ and $y_i$: $X_i = y_1 + \dots + y_i$.
    \item Relation between $y_i$ and $t_i$: For each $y_i$, $||y_i||_{1} = y_{i, 1} + y_{i, 2} = t_i$.
    \item Best-response constraint: For each $1 \leq i \leq k$, $X_i \in R_{b_i}$ and $X_{i-1} \in R_{b_{i}}$. Here $R_{b_i}$ (in our example, one of $\{R_A, R_B, R_C\}$) is the cone of states where $b_i$ is a best-response, and is specified by finitely many linear constraints.
\end{itemize}

The set of values of $(y_1, y_2, \dots, y_k)$ is a projection of the polytope defined by the above constraints and hence is also a polytope. 
\end{proof}

We can now prove the main result of this section:

\begin{theorem}
\label{thm:mb_is_a_polytope}
The mean-based menu $\MB$ of the learner payoff defined by \eqref{eq:game} is a polytope.
\end{theorem}
\begin{proof}
By Lemma \ref{lem:characterize_primitive}, there are finitely many possible fingerprints for a primitive spiral or trajectory. By Lemma \ref{lem:fingerprint_lp}, the set of profiles for each of these fingerprints is a polytope, so $\MB$ is equal to the convex hull of finitely many polytopes (and hence itself is a polytope).
\end{proof}

As a side note, note that the proof of Lemma \ref{lem:characterize_primitive} allows the optimizer to efficiently optimize against a mean-based learner in this class of games by solving a small finite number of LPs (since the output of the LP described in the proof provides not just the extremal CSP in $\M_{MB}$, but also valid trajectory attaining that CSP). It is an interesting question whether there exists a bound on the total number of possible fingerprints in more general games (a small bound would allow for efficient optimization, but it is not clear any finite bound exists).


\subsection{The mean-based menu equals the FTRL menu}
\label{sec:mb-equals-mwu}

We would now like to show that for the games in question, the menu $\MB$ is not merely a subset of the asymptotic menu of a mean-based no-regret algorithm (as indicated by Lemma \ref{lem:sub-mean-based}), but actually an exact characterization of such menus. While we do not show this in full generality, we show this for the large class of algorithms in $\ftrl$; recall these are instantiations of Follow-The-Regularized-Leader with a strongly convex regularizer. These algorithms exhibit the following useful properties (in fact, any learning algorithm satisfying these properties can be shown to be Pareto-dominated via the machinery we introduce).

\begin{lemma}
\label{lem:ftrl_properties}
    Algorithms $\cA$ in $\ftrl$ exhibit the following properties:
    \begin{itemize}
        \item $\A$ is $\gamma$-mean-based where $\gamma(t) = o(1)$ and $\gamma(t) \cdot t$ is a weakly monotone function of $t$.
        \item On every sequence of play, the regret of $\A$ lies between $0$ and $\regbound(T)$, for a sublinear function $\regbound(T) = o(T)$. (The lower bound on regret is due to~\cite{gofer2016lower}.)
        \item For any $T$, there exists a function $\wtA_T:\mathbb{R}^{n} \rightarrow \Delta_n$ such that the action $y_t$ taken by algorithm $A^T$ at round $t$ can be written in the form:
        
        $$y_t = \wtA_{T}(U_{t,1}, U_{t,2}, \dots, U_{t, n}),$$

        \noindent
        where $U_{t, j} = \sum_{s=1}^{t-1} u_L(x_s, j)$ is the cumulative payoff of action $j$ up to round $t$. Moreover, for any $\Delta \in \mathbb{R}$, 
        
        $$\wtA_T(U_1 + \Delta, U_2 + \Delta, \dots, U_{n} + \Delta) = \wtA_T(U_1, U_2, \dots, U_n).$$
        
        In words, the moves played by $\A$ each round are entirely determined by the cumulative payoffs of the actions and are invariant to additively shifting the cumulative payoffs of all actions by the same quantity.
        
        \item 

        There exists a sublinear function $\sshift(T) = o(T)$ (depending only on $\gamma$ and $R$) with the following property: if  $U \in \mathbb{R}^n$ satisfies $U_{j} \leq \max_{j' \in [n]} U_{j'} - \gamma(T)\cdot T$, then there exists an $\wtA'_{T}$ corresponding to a $\ftrl_{T}(\eta', R')$ algorithm over $n-1$ actions where if we define $y \in \Delta_n$ and $y' \in \Delta_{n-1}$ via

        \begin{align*}
          y &= \wtA_{T}(U_{1}, \dots, U_{j-1}, U_j, U_{j+1}, \dots, U_{n})  \\
          y' &= \wtA'_{T}(U_{1}, \dots, U_{j-1}, U_{j+1}, \dots, U_{n}), 
        \end{align*}
        
        then $||y - \pi(y')||_1 \leq \sshift(T)/T$. Here $\pi$ is the embedding of $\Delta_{n-1}$ into $\Delta_{n}$ by setting the $\pi(x)_{j'} = x_{j'}$ for $j' < j$, $\pi(x)_{j} = 0$, and $\pi(x)_{j'} = x_{j'-1}$ for $j' > j$. Intuitively, this states that whenever certain actions are historically-dominated, it is as if we are running a different FTRL instance on the non-dominated actions. 
    \end{itemize}
\end{lemma}

With these properties, we can prove the main theorem of this section -- that any FTRL algorithm has asymptotic menu equal to $\M_{MB}$ for the game defined in \eqref{eq:game}.

\begin{theorem}
\label{thm:sbm=mb}
For the $u_L$ defined in \eqref{eq:game}, if $\cA$ is an algorithm in $\ftrl$, then $\cM(\cA) = \MB$. 
\end{theorem}
\begin{proof}[Proof sketch]
We provide a sketch of the proof here, and defer the details to Appendix~\ref{app:sbm=mb}. The non-trivial part of the proof to show is to show that $\cM(\cA) \subseteq \MB$ since the other direction follows directly from $\cA$ being mean-based and Lemma~\ref{lem:sub-mean-based}. To this end, we focus on a discrete-time version of the trajectories used to build $\MB$. Intuitively, these trajectories guarantee that a unique action is significantly the best action in hindsight, or the historical leader, on all but a sublinear number of time steps. In other words, all mean-based algorithms behave identically, i.e induce identical CSPs, up to subconstant factors against each such optimizer subsequence. 

The main idea in the proof is that we can transform any optimizer sequence inducing a CSP against $\cA$ into a sequence with the above properties, against which all mean-based algorithms (including $\cA$) induce almost identical CSPs. Clearly, the parts of the trajectory that lie in a region with a ``clear'' historical leader are not an issue, and the challenge to this approach comes from parts of a trajectory that lie in a region where more than one action is close to being the leader. In such regions, different mean-based algorithms could potentially exhibit markedly different behavior. The key to addressing this issue is that in the game~\eqref{eq:game}, there can be at most two actions in contention to be the leader at any point of time. Utilizing the last property of  algorithms in $\ftrl$ shown in Lemma~\ref{lem:ftrl_properties}, it is as if $\cA$ is being played exclusively with two actions in such segments. The no-regret upper and lower bound guarantees, when suitably applied in such segments, automatically strengthen into a stronger guarantee, which we use in concert with the regret lower bounds to construct a local transformation of the trajectory in such a region (found in Lemma~\ref{lem:nc_replacement}). Each such transformation add a sublinear error to the generated CSP.  Thus, to bound the total drift from the original CSP, we show that we do not need a large number of transformations (the specific number balances with the asymptotics of the sublinear error that we make in each local transformation) using some properties of~\eqref{eq:game}.
    
\end{proof}

As a corollary, we can finally prove the main result of this section -- that there exist learner payoffs for which algorithms in $\ftrl$ are Pareto dominated. This result follows directly from an application of Theorem~\ref{thm:pareto-optimal-characterization} to the menu of such an algorithm; since we have shown it is a polytopal no-regret menu and its minimum value points are a superset of the minimum value points of $\M_{NSR}$.  

\begin{corollary}
\label{cor:coup_de_grace_mw}
For the $u_L$ defined in \eqref{eq:game}, all algorithms in $\ftrl$ are Pareto-dominated.
\end{corollary}

We have proved a sequence of results about algorithms in $\ftrl$ for a single learner payoff, i.e.~\ref{eq:game}, but our results in fact hold for a positive measure set of $2 \times 3$ learner payoffs. All the results in this section can be generalized to perturbed versions of this learner payoff \eqref{eq:perturbed_game}, and we show how to do so in Appendix~\ref{app:mean_based_pareto_dominated}.

\bibliographystyle{alpha}
\bibliography{references}
 
\appendix
\section{Examples of asymptotic menus}\label{sec:examples}

In the section, we give some explicit examples of asymptotic menus for simple learning algorithms. All of these algorithms will operate in a $2 \times 2$ game where the optimizer has actions $A$ and $B$ and the learner has actions $P$ and $Q$.

\subsubsection*{Algorithm 1: always play $Q$}

We will begin by considering the algorithm $\cA_1$ that always plays action $Q$ (i.e., always plays $y_t = (0, 1)$). The optimizer can choose to play any sequence of actions against this, and therefore, for any finite time horizon $T$, the resulting menu $\M(\cA^{T})$ is the convex hull of $A \otimes Q$ and $B \otimes Q$. It follows that the asymptotic menu $\M(\cA)$ is also the convex hull of $A \otimes Q$ and $B \otimes Q$ (i.e., all CSPs of the form $x \otimes Q$ for any strategy $x$ of the optimizer).

\subsubsection*{Algorithm 2: always play $P/2 + Q/2$}

Similar to $\cA_2$, $\cA_1$ always plays the fixed mixed strategy $P/2 + Q/2$. Again, only the optimizer's marginal distribution matters and the asymptotic menu $\cM(\cA_2)$ is the convex hull of $A \otimes (P/2 + Q/2)$ and $B \otimes (P/2 + Q/2)$. 


\subsubsection*{Algorithm 3: Alternate between $P$ and $Q$}

Algorithm $\cA_3$ plays action $P$ on odd rounds and action $Q$ on even rounds. Although this may seem similar to $\cA_2$, this change actually introduces many new CSPs into the asymptotic menu.

In particular, while the learner's marginal distribution is the same (in the limit) as under $\A_{2}$, the optimizer can now correlate their actions with the actions of the learner. Playing against $\A_3$ is akin to playing against two disjoint algorithms similar to $\cA_1$ (for odd and even rounds) and averaging the results, since the algorithm is completely non-adaptive. Since in odd rounds the adversary can implement any CSP of the form $x_1 \otimes P$ for $x_1 \in \Delta_2$, and in even rounds the adversary can implement any CSP of the form $x_2 \otimes Q$ for $x_2 \in \Delta_2$, the asymptotic menu $\M(\cA_3)$ is given by the set:

$$\M(\cA_3) = \left\{\frac{1}{2} (x_1 \otimes P) + \frac{1}{2}(x_2 \otimes Q) \mid x_1, x_2 \in \Delta_2\right\}.$$

Note this is exactly the set of CSPs whose marginal distribution for the learner is $P/2 + Q/2$.

\subsubsection*{Algorithm 4: Play $P$ until $A$, then only $Q$}

Finally, we consider an algorithm $\cA_4$ which plays the action $P$ until the first time the optimizer plays action $A$, after which it switches forever to playing action $Q$ (a ``grim trigger'' strategy). Note that unlike the previous learning algorithms, $\cA_{4}$ is responsive to what actions the optimizer has played in the past. 

To analyze this algorithm, consider the first time step $t$ (or more generally the fraction of $T$) where the optimizer does not play $B$. Until this happens, the only possible outcome in the transcript of the game is $(B,P)$, so the average CSP is $B \otimes P$. One the single time step $t$ where the optimizer first plays an action other than $B$, the learner will still play $P$. But for all time steps beyond this, the algorithm behaves exactly as $\cA_1$ (with menu given by the convex hull of $A\otimes Q$ and $B \otimes Q$). As the contribution of the single time step in which the deviation occurs goes to zero in the limit of $T$, the resulting asymptotic menu $\M(\cA_4)$ is therefore given by the convex hull $\conv(\{B \otimes P, A \otimes Q, B \otimes Q\})$.
\section{Generalized optimizer distributions}\label{sec:generalized_optimizer_distributions}
Our results can be generalized to a setting with a finitely supported distribution over optimizers. The learner plays the same game (same number of actions) independently with each optimizer and employing the same algorithm, all of whom may have different payoffs, and obtains as their payoff the weighted (by the distribution) sum of payoffs. The optimizers are $O_1, O_2 \cdots O_k$, each with their own payoff function $u_{O_i}$; and with probability $\alpha_i$ of occurring and the distribution is referred to as $\cO$. Formally, we define the value of the learning algorithm $\cA$ against optimizer distribution $\cO$ as $V_{L}(\cA, \cO) := \sum_i \alpha_i V_{L}(\cA, O_i)$. The notions of Pareto-domination and Pareto-optimality are similarly defined, over the set of optimizer distributions, with the new learner payoff definition.

We show that the Pareto-domination relation between algorithms is invariant across these two definitions, and therefore the algorithms that are Pareto-dominated (Pareto-optimal) are the same set in both settings. 

\begin{lemma}
$\mathcal{A}_{1}$ Pareto-dominates $\mathcal{A}_{2}$ over the space of all generalized optimizers iff $\mathcal{A}_{1}$ Pareto-dominates $\mathcal{A}_{2}$ over the space of all linear optimizers.
\end{lemma}
\begin{proof}
First, we will prove the forward direction: assume that $\cA_2$ is Pareto-dominated by $\cA_1$ over the space of all generalized optimizers. Then, there is at least one generalized optimizer $\cO$ for which $V_{L}(\cA_1, \cO) > V_{L}(\cA_2, \cO)$. We therefore have:

\begin{align*}
V_{L}(\cA_1, \cO) > V_{L}(\cA_2, \cO) \\
\Rightarrow \sum_{i=1}^{k}V_{L}(\cA_1, \mu_{\cO_i})\alpha_{i} > \sum_{i=1}^{k}V_{L}(\cA_2, \mu_{\cO_i})\alpha_i\\
\Rightarrow \exists\, i\text{ s.t. }V_{L}(\cA_1, \mu_{\cO_i}) \alpha_{i} > V_{L}(\cA_2, \mu_{\cO_i})\alpha_i \\
\Rightarrow \exists\, i \text{ s.t. }V_{L}(\cA_1, \mu_{\cO_i}) > V_{L}(\cA_2, \mu_{\cO_i}).
\end{align*}

Thus, there is at least one linear optimizer for which $V_{L}(\cA_1, \mu_{\cO_i}) > V_{L}(\cA_2, \mu_{\cO_i})$. Next, we need that, for all linear optimizers, $V_{L}(\cA_1, \mu_{\cO_i}) \geq V_{L}(\cA_2, \mu_{\cO_i})$. Assume for contradiction that this is not the case. Then, there exists some $i$ such that $V_{L}(\cA_1, \mu_{\cO_i}) < V_{L}(\cA_2, \mu_{\cO_i})$. But this is a contradiction, as $\cO_i$ is also a generalized optimizer. Therefore, there is at least one linear optimizer $\cO$ for which $V_{L}(\cA_1, \mu_i) > V_{L}(\cA_2, \mu_i)$, and for all for all linear optimizers, $V_{L}(\cA_1, \mu_{i}) \geq V_{L}(\cA_2, \mu_{i})$. Thus, $\cA_1$ Pareto-dominates $\cA_2$ over the space of all linear optimizers. 

For the reverse direction, assume $\cA_2$ is Pareto-dominated in the original single optimizer setting by $\cA_1$. Thus, there exists an optimizer $O$ for which $V_{L}(\cA_1, O) > V_{L}(\cA_2, O)$ - thinking of a distribution $\cO$ with all the mass on this single optimizer, we see that $V_{L}(\cA_1, \cO) > V_{L}(\cA_2,\cO)$. Assume for contradiction there there is some optimizer distribution $\cO = \{O_i, \alpha_i \}_{i=1}^k$ such that $V_{L}(\cA_1, \cO) < V_{L}(\cA_2,\cO)$. We have already seen that this implies an optimizer $O_i$ such that $V_{L}(\cA_1, O_i) < V_{L}(\cA_2,O_i)$, which results in a contradiction. 
\end{proof}

\section{Pareto-dominance at a point implies Pareto-dominance over a positive measure set}
\label{app:pareto-dominance-eq}

In this section, we prove that in our definition of Pareto-domination, it does not matter whether we require one menu to strictly outperform the other menu for only a single choice of optimizer payoffs $u_O$ or a positive measure set of optimizer payoffs -- both notions result in the same partial ordering across menus.

Intuitively, we can expect this to be the case since if $V_{L}(\cM_1, u_O) > V_{L}(\cM_2, u_O)$ and both functions are continuous in $u_O$, then there must be an open neighborhood of $u_O$ where this continues to hold. Unfortunately, $V_L$ is in general not continuous (e.g., in the case that $\cM$ is a polytope, it is a piece-wise constant function). However, it is still nicely enough behaved that this same intuition continues to hold.

We begin by proving some convex analysis preliminaries. In the lemmas that follow, we say $x$ is an \emph{extreme point} of a (bounded, closed) convex set $S$ if $x$ cannot be written as the convex combination of other points in $S$. We say $x$ is an extreme point of $S$ in the direction $u$ if $x$ is a maximizer of the linear function $f(s): S \rightarrow \mathbb{R} = \langle u, s\rangle$; similarly, we say $x$ is the \emph{unique} extreme point of $S$ in the direction $u$ if it is the unique such maximizer. 

\begin{lemma}\label{lem:pareto-eq-3}
Let $x$ be a (not necessarily unique) extreme point of the bounded convex set $S$ in the direction $u$. Then, for any $\eps > 0$, there exists a $u'$ with $|u - u'| < \eps$ where $x$ is the unique extreme point of $S$ in the direction $u'$.
\end{lemma}
\begin{proof}
Every extreme point of a convex subset of Euclidean space is the unique maximizer of at least one linear functional; assume $x$ is the unique extreme point in the direction $u_0$. Then if we take $u' = u + \eps u_0$, $x$ is also the unique extreme point in the direction $u'$. 
\end{proof}

\begin{lemma}\label{lem:pareto-eq-1}
Assume $x$ is the unique extreme point of the bounded convex set $S$ in the direction $u$. Then for all $\eps > 0$, there exists a $\delta > 0$ s.t.

$$\diam\left(\{s \in S \mid \langle u, s \rangle \geq \langle u, x \rangle - \delta \}\right) \leq \eps.$$
\end{lemma}
\begin{proof}
Assume to the contrary that the statement is false for some choice of $S$, $x$, $u$, and $\eps$. This means that for any $\delta > 0$, the set $S_{\delta} = \{s \in S \mid \langle u, s \rangle \geq \langle u, x \rangle - \delta \}$ has diameter at least $\eps$; in particular, there must exist a point $s_{\delta} \in S_{\delta}$ where $|s_{\delta} - x| \geq \eps/2$. Now pick any sequence $\delta_1, \delta_2, \dots$ converging to $0$, and look at the sequence of points $s_{\delta_i}$. Since they are a sequence of points inside the compact set $S$, some subsequence of this sequence converges to a limit $\overline{s}$ with the properties that $\langle u, \overline{s}\rangle = \langle u, x\rangle$ and $|\overline{s} - x| \geq \eps/2$. But this contradicts the fact that $x$ is the unique extreme point of $S$ in direction $u$. 
\end{proof}

\begin{lemma}\label{lem:pareto-eq-2}
Assume $x$ is the unique extreme point of the bounded convex set $S$ in the direction $u$. Then for any open neighborhood $X$ of $x$ there exists an open neighborhood $U$ of $u$ s.t. for any $u' \in U$, all extreme points of $S$ in the direction $u'$ lie in $X$.
\end{lemma}
\begin{proof}
Since $X$ is an open neighborhood of $x$, $X$ contains a ball of radius $\eps$ centered on $x$ (for some $\eps > 0$). By Lemma \ref{lem:pareto-eq-1}, there exists a $\delta > 0$ such that $\diam\left(\{s \in S \mid \langle u, s \rangle \geq \langle u, x \rangle - \delta \}\right) \leq \eps$. We will pick $U$ to be the open ball of radius $r = \delta/3B$ around $u$, where $B = \max_{s \in S} |s|$ is the maximum norm of any element of $S$.

To see why this works, consider any $u' \in U$; we can write $u' = u + \nu$ where $|\nu| \leq r$. Assume that $x'$ is an extreme point of $S$ in the direction $u'$ and $|x - x'| > \eps$. But now, note that:

$$\langle u', x'\rangle = \langle u, x'\rangle + \langle \nu, x'\rangle < \langle u, x \rangle - \delta + rB,$$

\noindent
where this inequality follows since $|x' - x| > \eps$, so $x'$ cannot lie within $S_{\delta} = \{s \in S \mid \langle u, s \rangle \geq \langle u, x \rangle - \delta \}$, and therefore $\langle u, x'\rangle < \langle u, x\rangle - \delta$ (also $|\nu| \leq r$, $|x'| \leq B$). Similarly, note that:

$$\langle u', x\rangle = \langle u, x \rangle + \langle \nu, x \rangle \geq \langle u, x \rangle - rB.$$

Since $2rB < \delta$, we have that $\langle u', x \rangle > \langle u', x'\rangle$. But this means that $x'$ could not have been an extreme point of $S$ in direction $u'$, as desired.
\end{proof}

The above lemmas connect with our analysis of asymptotic menus since if $\cM$ has a unique extreme point $\csp$ in the direction $u_O$, then $V_{L}(\cM, u_O) = u_L(\csp)$. We use this to prove the following two sublemmas. In the first sublemma, we show that we can find an open neighborhood around $u_O$ of optimizer payoffs $u$ such that $V_{L}(\cM, u)$ is never too much larger than $V_L(\cM, u_O)$ (this we can actually do directly, and will not require any of the lemmas above). 

\begin{lemma}\label{lem:pareto-equiv-menu-1}
Let $\cM$ be an asymptotic menu and $u_O$ an optimizer payoff where $V_{L}(\cM, u_O) = v$. Then, for any $\eps > 0$, there exists an open neighborhood of optimizer payoffs $U$ around $u_O$ such that $V_{L}(\cM, u'_O) \leq v + \eps$ for any $u'_O \in U$.
\end{lemma}
\begin{proof}
Assume to the contrary that any open neighborhood $U$ around $u_O$ contains a payoff $u'_O$ with $V_{L}(\cM, u'_O) > v + \eps$. Let $U_1, U_2, \dots$ be a sequence of open balls around $u_O$ with decreasing radius converging to $0$, and for each $U_i$ let $u_i$ be the payoff such that $V_{L}(\cM, u_i) > v + \eps$, and let $\csp_i$ be the corresponding maximizing point in the menu $\cM$ (i.e., $u_L(\csp_i) > v + \eps$). 

Since the $\csp_i$ all belong to the compact convex set $\cM$, there exists a subsequence converging to some $\overline{\csp} \in \cM$. This limit point must have the property that $u_L(\overline{\csp}) \geq v + \eps$. But since $\overline{\csp}$ is also the limit of extreme points of $\cM$ in directions converging to $u_O$, it must be an extreme point of $u_O$. But this contradicts the fact that $V_{L}(\cM, u_O) = v$ (since $V_{L}(\cM, u_O)$ is the maximum of $u_L(\csp)$ over all maximizers $\csp$ of $u_O$, of which $\overline{\csp}$ is one). The conclusion follows.
\end{proof}

In the second lemma, we show that we can find an open set of optimizer payoffs $u$ where $V_L(\cM, u)$ is never too much larger than $V_L(\cM, u_O)$. Note that since we define $V_L(\cM, u_O)$ to be the \emph{largest} value of $u_L(\csp)$ over all $\csp \in \cM$ maximizing $u_O(\csp)$, this is not quite symmetric with Lemma \ref{lem:pareto-equiv-menu-1}. In fact, we will prove a weaker condition (instead of an open neighborhood of $u_O$, a sufficiently nearby open set of payoffs), and it is here that we will use Lemmas \ref{lem:pareto-eq-3} and \ref{lem:pareto-eq-2}.

\begin{lemma}\label{lem:pareto-equiv-menu-2}
Let $\cM$ be an asymptotic menu and $u_O$ an optimizer payoff where $V_{L}(\cM, u_O) = v$. Then, for any $\eps > 0$, there exists an open set of optimizer payoffs $U$ such that $V_{L}(\cM, u'_O) \geq v - \eps$ for any $u' \in U$, and such that there exists a $u' \in U$ with $|u' - u_O| < \eps$.
\end{lemma}
\begin{proof}
Since $V_{L}(\cM, u_O) = v$, there must be a $\csp \in \cM$ which is an extreme point in the direction $u_O$, and for which $u_L(\csp) = v$. By Lemma \ref{lem:pareto-eq-3}, this means there exists a payoff $u'_O$ for which $|u'_O - u_O| \leq \eps$ and for which $\csp$ is the unique optimizer of $u'_O$. Then, by Lemma \ref{lem:pareto-eq-2}, there exists an open neighborhood $U$ of $u'_O$ such that for any payoff $u \in U$, the optimizer in $\cM$ of $u$ has distance at most $\eps$ from $\csp$ (we pick $X$ to be the open ball of radius $\eps$ around $\csp$ in this application of Lemma \ref{lem:pareto-eq-2}). This set $U$ has our desired properties. 
\end{proof}

With Lemmas \ref{lem:pareto-equiv-menu-1} and \ref{lem:pareto-equiv-menu-2}, we can now prove the main result of this section.

\begin{theorem}\label{thm:pareto-equiv}
Let $\cM_1$ and $\cM_2$ be two menus such that there exists a $u_O$ where $V_{L}(\cM_1, u_O) > V_{L}(\cM_2, u_O)$. Then there exists a positive measure set $U$ of optimizer payoffs where $V_{L}(\cM_1, u) > V_{L}(\cM_2, u)$ for any $u \in U$.
\end{theorem}
\begin{proof}
Pick an $\eps < (V_{L}(\cM_1, u_O) - V_{L}(\cM_2, u_O))/10$. By Lemma \ref{lem:pareto-equiv-menu-1}, there exists an open neighborhood $U_2$ of optimizer payoffs around $u_O$ such that $V_{L}(\cM_2, u) \leq V_{L}(\cM_2, u_O) + \eps$ for all $u \in U_2$. This open set $U_2$ must contain a ball of some radius $r$; set $\eps' = \min(r, \eps)$. 

Then, by Lemma \ref{lem:pareto-equiv-menu-2}, there exists an open set $U_1$ of optimizer payoffs such that $V_{L}(\cM_1, u) \geq V_{L}(\cM_1, u_O) - \eps'$ for all $u \in U_1$ and such that there exists a $u \in U_1$ such that $|u - u_0| < \eps'$.

This means that the sets $U_1$ and $U_2$ have non-empty intersection. Let $U = U_1 \cap U_2$; since $U$ is a non-empty open set, it has positive measure. But now note that for any $u \in U$, $V_{L}(\cM_1, u) \geq V_{L}(\cM_1, u_O) - \eps' > V_{L}(\cM_2, u_O) + \eps \geq v_{L}(\cM_2, u)$, as desired.
\end{proof}

\section{Inconsistent learning algorithms}
\label{sec:inconsistent}
In the majority of this paper, we restrict our attentions to consistent learning algorithms -- learning algorithms $\cA$ with a unique asymptotic menu $\cM(\cA)$ such that if you play against $\cA$ for a sufficiently large time horizon, you can implement any CSP $\csp \in \cM(\cA)$ arbitrarily closely.

Not all learning algorithms are consistent. A simple example is an algorithm $\cA$ that always plays action $1$ on odd rounds and always plays action $2$ on even rounds. Nonetheless, it is still possible to extend much of our analysis to the setting of inconsistent learning algorithms. The main caveat is that instead of each learning algorithm being associated with a single asymptotic menu, it will be associated with a collection of asymptotic menus, each corresponding to a limit point of the sequence of menus for each time horizon $T$.

\begin{theorem}\label{thm:nonconsistent}
For any (not necessarily consistent) learning algorithm $\cA$, the sequence $\M(A^1), \M(A^2), \dots$ has at least one limit point under the Hausdorff metric. Moreover, all such limit points are valid asymptotic menus.
\end{theorem}
\begin{proof}
Since the space of convex subsets of a bounded set is sequentially compact under the Hausdorff metric (see \cite{henrikson1999completeness}), it follows that the sequence $\M(A^1), \M(A^2), \dots$ has at least one limit point.

To prove the second part: consider a limit point $\M$ of the above sequence. This means that there is a subsequence $\M(A^{T_1}), \M(A^{T_2}), \dots$ of the above sequence that converges to $\M$ under the Hausdorff metric. We will construct a learning algorithm $\cA^*$ whose asymptotic menu is $\cM$. 

The algorithm $\cA^*$ will work as follows. To run $\cA^*$ for a time horizon of length $T$, choose the largest $i$ such that $T_i \leq T$. Run $A^{T_i}$ for the first $T_i$ rounds of the game. Then recursively run $\cA^*$ for the remaining $T - T_i$ rounds (i.e., finding the next largest $i'$ s.t. $T_{i'} \leq T - T_i$, and so on). 

We will now prove that sequence of menus corresponding to $\cA^*$ converges to $\M$ under the Hausdorff metric. Fix an $\eps > 0$. Since the subsequence $M(A^{T_i})$ converges to $\cM$, there must be a threshold $n_\eps$ such that for any $n \geq n_{\eps}$, the Hausdorff distance $d_H(\M(A^{T_{n}}), \M) < \eps/100$. On the other hand, in any execution of $A^*$, at most $T_{n+1}$ rounds are spent playing one of the first $n$ sub-algorithms $A^{T_1}, A^{T_2}, \dots, A^{T_n}$. Therefore, if we pick $T \geq 100 T_{n_{\eps}}/\eps$, we can write any CSP $\csp$ reached when playing against $(A^*)^T$ in the form:

$$\csp = \frac{\eps}{100} \csp_{\mathrm{pre}} + \left(1 - \frac{\eps}{100}\right) \csp_{\mathrm{post}},$$

\noindent
where $\csp_{\mathrm{pre}}$ is the CSP induced by the first $\eps/100$ fraction of rounds, and $\csp_{\mathrm{post}}$ is the CSP induced by the remainder of the rounds. Since $\csp_{\mathrm{post}}$ is a convex combination of CSPs belonging to the sets $\M(A^{T_{n}})$ for $n \geq n_{\eps}$, we must have $d(\csp_{\mathrm{post}}, \M) \leq \eps/100$. Since the maximum norm of an element in $\Delta_{mn}$ is bounded by $1$, this implies (via the triangle inequality) that $d(\csp, \M) < \frac{\eps}{100} + \frac{\eps}{100} < \eps$. It follows that every point in $\M((A^*)^{T})$ is within $\eps$ of an element of $\cM$. But also, for any target $\csp^* \in \M$, we can pick $\csp_{\mathrm{post}}$ to be the convex combination of the closest element to $\csp^*$ in each of the relevant menus $\M(A^{T_n})$ -- this also guarantees $d(\csp_{\mathrm{post}}, \csp^*) \leq \eps/100$ and thus that $d(\csp, \csp^*) < \eps$ and so every point in $\cM$ is within $\eps$ of an element of $\cM((A^*)^T)$. It follows that $\cM(A^*) = \cM$. 
\end{proof}

The fact that no-regret and no-swap-regret algorithms induce no-regret and no-swap-regret menus (Theorem \ref{thm:nr_nsr_containment}) extends easily to the case of inconsistent algorithms.

\begin{lemma}\label{lem:nr_nsr_containment_nonconsistent}
If a $\M$ is an asymptotic menu of the (inconsistent) no-regret learning algorithm $\cA$, then $\M(\cA) \subseteq \M_{NR}$. If a $\M$ is an asymptotic menu of the (inconsistent) no-swap-regret learning algorithm $\cA$, then $\M(\cA) \subseteq \M_{NSR}$.
\end{lemma}
\begin{proof}
Follows identically to the proof of Theorem \ref{thm:nr_nsr_containment}, but applied to the subsequence of menus converging to $\M$. 
\end{proof}

As a result of Theorem \ref{thm:nonconsistent} and Lemma \ref{lem:nr_nsr_containment_nonconsistent}, we immediately obtain a stronger version of Theorem \ref{thm:unique_nsr_menu}: not only do all consistent no-swap-regret algorithms share the same asymptotic menu, but in fact every no-swap-regret algorithm shares the sasme asymptotic menu (and is in fact, consistent).

\begin{corollary}
\label{cor:all_nsr_algos}
Any no-swap-regret algorithm is consistent (and therefore by Theorem \ref{thm:unique_nsr_menu} has a unique asymptotic menu of $\M_{NSR}$).
\end{corollary}
\begin{proof}
Let $\cA$ be a no-swap-regret algorithm. By Theorem \ref{thm:nonconsistent} and Lemma \ref{lem:nr_nsr_containment_nonconsistent}, any asymptotic menu $\cM$ of $\cA$ must be the asymptotic menu of a consistent no-swap-regret algorithm. By Theorem \ref{thm:unique_nsr_menu}, this implies $\cM = \cM_{NSR}$. This means the sequence $\cM(A^1), \cM(A^2), \dots$ has a unique limit point $\M_{NSR}$ and therefore converges, so the original algorithm $\cA$ was consistent.
\end{proof}

It is an interesting question whether standard (mean-based) no-regret learning algorithms such as multiplicative weights or FTRL are consistent learning algorithms. Although intuitively it seems like this should be the case (due to the simple form such algorithms take), it is surprisingly difficult to show this formally. In particular, given a CSP $\csp$ reached by e.g. multiplicative weights for some large finite time horizon $T$, it is not obvious how to reach a nearby CSP $\csp'$ for larger time horizons $T'$ (it is not sufficient to just ``scale up'' the sequence of adversary actions). We leave this as an interesting open question. Nonetheless, the results of this section imply that we can still say meaningful things about asymptotic menus of such algorithms (and we prove in Section \ref{sec:mb-equals-mwu} that these algorithms are in fact consistent for some payoffs $u_L$).
\section{Additional Results and Proofs from Section~\ref{sec:mb_algos_and_menus}}
\label{app:mean_based_pareto_dominated}


This appendix section continues the discussion in Section~\ref{sec:mb_algos_and_menus}, discusses additional results about $\MB$, extends the set of learner payoffs~\eqref{eq:learner_value} for which we prove the Pareto-domination of $\ftrl$ algorithms and presents the complete proof of Theorem~\ref{thm:sbm=mb}.


\subsection{Additional Results about $\MB$}
\label{app:proof_no_pos_neg_regret}


Lemma~\ref{lem:sub-mean-based} showed that $\MB$ is contained within the menu of every mean-based algorithm. Since mean-based no-regret algorithms are known to exist; we immediately obtain the following corollary (additionally using Theorem~\ref{thm:nr_nsr_containment}, which tells us that the asymptotic menu of every no-regret algorithm is contained within $\M_{NR}$).

\begin{corollary}
     $\MB \subseteq \M_{NR}$.
 \end{corollary}

One way to interpret this is that in the CSP $\Prof(\tau)$ associated with any valid trajectory $\tau$, the learner can never benefit by unilaterally deviating: they always do at least as well as playing a best response to the optimizer's marginal distribution. Interestingly, we can show that mean-based algorithms also have the property that they can never do much better than this -- no profile can end up with negative regret. This property is useful in generalizing a result in Section~\ref{sec:mb_algos_and_menus}.

\begin{lemma}
\label{lem:no_pos_neg_regret}
    For any valid trajectory $\tau$, the value of the learner in the CSP $\Prof(\tau)$ exactly equals the value obtained by best-responding to the optimizer's marginal distribution.
\end{lemma}

This result implies that all extreme points of $\MB$ have exactly zero regret, however, interior points might still have negative regret, since a regret lower bound is not preserved by convex combinations of CSPs.

\begin{proof}
Consider any continuous time-formulation trajectory $\tau = \{(x_1, t_1, b_1), (x_2, t_2, b_2), \dots, (x_k, t_k, b_k)\}$; let $\alpha_O$ be the marginal distribution of the optimizer in the CSP $\Prof(\tau)$. We show that $\mu_L(\Prof(\tau)) = \max_{a \in [n]} \mu_L(\alpha_O \otimes a)$. The proof proceeds via induction on $k$, the length of the transcript. Let $\tau_i$ represent the prefix of the first $i$-triples of $\tau$ (this is also a valid trajectory). Recall that $\overline{x}_{i} = \sum_{j=1}^{i} (x_jt_j) / \sum_{j=1}^{i} t_j$ (with $\overline{x}_0 = 0$) is the time-weighted average strategy of the optimizer over the first $i$ triples. The base case is $k=1$, for which the induction hypothesis is satisfied, since the validity condition guarantees that $b_1 \in BR_L(x_1)$. Assuming the induction hypothesis for $i$, that is equivalent to stating that $\mu_L(\tau_i) = \max_{a \in [n]} \mu_L(a \otimes \overline{x}_i)$.

However, the validity condition guarantees that $b_{i+1} \in \argmax_{a \in [n]} \mu_L(a \otimes \overline{x}_i)$ implying that $\mu_L(\tau_i) =  \mu_L(b_{i+1} \otimes \overline{x}_i)$. However, the CSP $\tau_{i+1}$ is generated from $\tau_i$ by playing $b_{i+1}$ against $x_{i+1}$ and reweighting accordingly, this leads to $\mu_L(\tau_{i+1}) =  \mu_L(b_{i+1} \otimes \overline{x}_{i+1})$. The validity condition also tells us that $b_{i+1} \in \argmax_{a \in [n]} \mu_L(a \otimes \overline{x}_{i+1})$, which implies that $\mu_L(\tau_{i+1}) = \max_{a \in [n]} \mu_L(a \otimes \overline{x}_{i+1})$, extending the induction hypothesis and completing the proof.  
\end{proof}

\subsection{Extending the Learner Payoffs}

In Section~\ref{sec:mb_algos_and_menus}, we showed a learner payoff~\eqref{eq:game} under which $\ftrl$ algorithms are Pareto-dominated, through a sequence of results about the menus of such algorithms showing that they fit our characterization of Pareto-dominated no-regret algorithms. The same sequence of results can be shown to be true for a positive measure set of learner payoffs obtained by perturbing~\eqref{eq:game}.

The learner payoffs we consider are of the following form, where all $\epsilon_{1}...\epsilon_{6}$ are between $0$ and $\varepsilon <= \frac{1}{100}$, and $\epsilon_{1} > \epsilon_{2}$:

\begin{equation}
\label{eq:perturbed_game}
\mu_{\mathcal{L}} = 
    \begin{bmatrix}
        & \text{N} & \text{Y} \\
        \text{A} &  \epsilon_{1} & \epsilon_{2}  \\
        \text{B} &  -1/6+ \epsilon_{3} & 1/3+ \epsilon_{4} \\
        \text{C} & -1/2+ \epsilon_{5} & 1/2+ \epsilon_{6} \\
    \end{bmatrix}
\end{equation}




\begin{theorem}\label{thm:high-mb-swap-regret_generalized}
For any choice of $u_L$ in \eqref{eq:perturbed_game}, $\MB^{-} \not \subseteq \M_{NSR}^{-}$.
\end{theorem}

The proof of this result is along the same lines as that of Theorem~\ref{thm:high-mb-swap-regret}, with some additional accounting to adjust for the perturbation, and can be found in Appendix~\ref{sec:proof_srmeanbased}.

\begin{theorem}
\label{thm:mb_is_a_polytope_generalized}
The mean-based menu $\MB$ of any learner payoff of the form~\eqref{eq:perturbed_game} is a polytope.
\end{theorem}

The proof of Theorem~\ref{thm:mb_is_a_polytope} works as is for~\eqref{eq:perturbed_game}.

\begin{theorem}
\label{thm:sbm=mb_generalized}
For any $u_L$ defined in \eqref{eq:perturbed_game}, if $\cA$ is a $\ftrl$ algorithm, then $\cM(\cA) = \MB$. 
\end{theorem}

The proof of this result is exactly the same as that of Theorem~\ref{thm:sbm=mb}, for which we provided a proof sketch in Section~\ref{sec:mb_algos_and_menus}. A full version of this proof is given below in Appendix~\ref{app:sbm=mb}.

\begin{corollary}
For every $u_L$ defined in \eqref{eq:perturbed_game}, all $\ftrl$ algorithms are Pareto-dominated.
\end{corollary}

The proof of this corollary is exactly the same as that of Corollary~\ref{cor:coup_de_grace_mw}, with the referenced results in Section~\ref{sec:mb_algos_and_menus} swapped out for the corresponding ones in this Appendix \ref{app:mean_based_pareto_dominated}.


\subsection{Proof of Theorem~\ref{thm:high-mb-swap-regret_generalized}}
\label{sec:proof_srmeanbased}



Fix any fixed learner payoff $u_L \in \eqref{eq:perturbed_game}$. We prove this result by showing a point in $\MB$ for which the learner obtains their maximin value and has positive swap regret (and is hence not in $\M_{NSR}$). More atomically, we come up with a valid trajectory $\tau$ in the continuous formulation mentioned in Section~\ref{sec:mb_algos_and_menus}. Based on Lemma~\ref{lem:no_pos_neg_regret}, we know that for $u_L(\Prof(\tau))$ to be the maximin value of the learner, the optimizer's marginal distribution in $\Prof(\tau)$ must be equal to the minimax distribution -- this informs our construction of $\tau$. The first step in our proof is therefore (approximately) identifying the minimax distribution for the learner payoff.

\begin{claim}
    There is a unique, minimax strategy $M$ against the learner payoff $u_L$ where $M$ picks action $X$ with probability $p^* = 2/3 \pm O(\varepsilon)$.
\end{claim}
\begin{proof}
    First, we claim that the minmax strategy must cause the learner to be indifferent between at least two of their actions. To see this, assume for contradiction the following cases:

    \begin{itemize}
    \item The minimax strategy $s$ strictly incentivizes $A$ over $B$ and $C$ by $\delta$. Then the optimizer can define a new strategy $(1 - \frac{\delta'}{2})s + \frac{\delta'}{2}Y$ for any $\delta' \leq \delta$, which still incentivizes $A$ but gives the learner strictly lower payoff. The only way this would not decrease payoff is if the strategy was already the pure strategy $Y$, but this cannot incentivize $A$. So this derives a contradiction.
    \item  The minimax strategy $s$ strictly incentivizes $B$ over $A$ and $C$ by $\delta$. Then the optimizer can define a new strategy $(1 - \frac{\delta'}{2})s + \frac{\delta'}{2}N$ for any $\delta' \leq \delta$, which still incentivizes $B$ but gives the learner strictly lower payoff. The only way this would not decrease payoff is if the strategy was already the pure strategy $N$, but this cannot incentivize $B$. So this derives a contradiction.
    \item  The minimax strategy $s$ strictly incentivizes $C$ over $A$ and $B$ by $\delta$. Then the optimizer can define a new strategy $(1 - \frac{\delta'}{2})s + \frac{\delta'}{2}N$ for any $\delta' \leq \delta$, which still incentivizes $B$ but gives the learner strictly lower payoff. The only way this would not decrease payoff is if the strategy was already the pure strategy $N$, but this cannot incentivize $C$. So this derives a contradiction.
    \end{itemize}

    Next, note that $A$ and $C$ are equal when $P(N) = p$ for some $p = \frac{1}{2} \pm O(\epsilon)$. Furthermore, when $P(N) = \frac{1}{2} \pm O(\epsilon)$, $B$ is better than both $A$ and $C$. So there cannot be a minimax strategy when $A$ and $C$ are tied best repsonses. Furthermore, note that $B$ and $C$ are equal when $P(N) = q$ for some $q =  \frac{1}{3} \pm O(\epsilon)$. Against this distribution, the learner can play $C$ and get payoff $\frac{1}{2}.\frac{2}{3} - \frac{1}{2}.\left(\frac{1}{3} \pm O(\epsilon)\right) = \frac{1}{6} \pm O(\epsilon)$. This cannot be the minimax strategy, as if the optimizer just plays $N$, the learner will get payoff at most $\epsilon_{1} < \frac{1}{6} \pm O(\epsilon)$.

    Therefore, the only remaining possibility is that the minimax distribution is when $A$ and $B$ are equal. This is a unique distribution, and will be of the form $P(N) = p^*$ for some  $p^* =  \frac{2}{3} \pm O(\epsilon)$

\end{proof}

Using a similar analysis, there is a unique distribution for the optimizer for which the learner's best response is tied between actions $B$ and $C$, this distribution plays action $N$ with probability $q^* = \frac{1}{3} \pm O(\varepsilon)$. Consider the continuous time trajectory $\tau = \{(q^*\cdot N + (1-q^*) \cdot Y, t_1, C), (N, t_2, B)\}$ where $t_1 + t_2 = 1$, $t_1,t_2 \ge 0$; and $t_1$ and $t_2$ are chosen so that $t_1 q + t_2  = p^*$. A feasible solution exists with $t_1, t_2 = \frac{1}{2} \pm O(\varepsilon)$ because $q$ is very close to $1/2$, and $p^*$ is very close to $2/3$. By construction, this is indeed a valid trajectory: we can check that $\ox_1 = q^*\cdot N + (1-q^*) \cdot Y$, $\ox_2 = p^* \cdot N + (1-p^*)Y$, $\BR_{L}(q^*\cdot N + (1-q^*) \cdot Y) = \{B, C\}$, and that $\BR_{L}(p^* \cdot N + (1-p^*)Y) = \{A, B\}$. For this trajectory, $\Prof(\tau) = t_1 q^* (N \otimes C) + t_1(1-q^*) (Y \otimes C) + t_2 (N \otimes B)$. The swap regret of $\Prof(\tau)$ is positive, since by swapping all weight on $B$ to $A$ increases the learner's average utility by $\frac{t_2}{6} \ge \frac{1}{24}$ (since $t_2 \ge \frac{1}{4}$), so $\Prof(\tau) \not\in \M_{NSR}$. However, by construction the marginal distribution of the optimizer in $\Prof(\tau)$ is exactly the minimax distribution of the optimizer, implying that the learner gets the minimax value, completing the proof.

\subsection{Proof of Theorem~\ref{thm:sbm=mb}/Theorem~\ref{thm:sbm=mb_generalized}}
\label{app:sbm=mb}







As a consequence of Lemma~\ref{lem:ftrl_properties}, the $\ftrl$ algorithm $\A$ is $\gamma$-mean based with $\gamma(t)\cdot t 
 $ being a sublinear and monotone function of $t$.  
Let $\intersectionmb$ denote the intersection of the menu of all mean-based algorithms. A restatement of Lemma~\ref{lem:sub-mean-based} is that $\MB \subseteq \intersectionmb$. However, by definition -- $\intersectionmb \subseteq \M(\A)$ implying that $\MB \subseteq \M(\A)$. To complete the proof of the lemma only, we only need to show that the menu $\M = \M(\A)$ of a given $\ftrl$ algorithm is contained in $\MB$.


To this end, we pick up each extreme point of $\M$ and show its containment in $\MB$. For each such CSP, we demonstrate a sequence of points in $\MB$ converging to this CSP; which, when combined with the closed and bounded property of $\MB$ (a polytope, see Theorem~\ref{thm:mb_is_a_polytope}), implies that the limit point is in $\MB$.

As our proof sketch indicated, this sequence of points is generated by observing the finite time menus of the algorithm $\cA$. For each point $P$ in $\M$, there is a a sequence of CSPs $P_1, P_2, \dots $ in the finite time menus that converge to $P$. Our proof shows a sequence of points $Q_1, Q_2, \dots, Q_T, \ldots \in \MB$  satisfying the property that $d(P_T,Q_T) = o(1)
~\footnote{Here, as everywhere else in this section, $o(1)$ refers to a sub-constant dependence on $T$}$. Directly using the triangle inequality then guarantees that the sequence $\{Q_T\}$ converges to $P$.


For ease of writing, we introduce $f(T) = \max \{ \gamma(T) \cdot T, \regbound(T), \sshift(T)\}$. 
$\rho(T)$ is the regret bound guaranteed by Lemma~\ref{lem:ftrl_properties}; while $\sshift(T)$ is the sublinear function associated with $\A$ from the last clause of Lemma~\ref{lem:ftrl_properties}, and is roughly the upper bound in the difference in the strategy chosen by the original algorithm and that chosen by the same algorithm restricted to two actions when the third action is sufficiently far from being a historical best action.~\esh{Sanity check this usage, please} This function $f$, being a max of sublinear functions, is itself sublinear. There exists a positive constant $K$ such that the magnitude of the regret of the algorithm is upper bounded by $K f(T)$. For ease of writing, we consider $K=1$, although our proof will work for any constant $K>0$. Even though, Lemma~\ref{lem:ftrl_properties} guarantees that the regret is contained in $[0,\regbound(T)]$; we work with a weaker bound of regret being contained in $[-\regbound(T),\regbound(T)]$

The rest of the proof focuses mostly on proving the existence of $Q_T$ given $P_T$. Without loss of generality, we restrict our attention to the simpler case where $P_T$ is an extreme point of the menu $\M(\A)_T$; and find a corresponding $Q_T$. The same idea can be extended to arbitrary menu points by rewriting $P_T$ as a convex combination of $mn+1$ extreme points (invoking Caratheodory's theorem).


We introduce some notation to aid with our proofs, and formalize the notions of boundaries and finite time trajectories. We track the (partial) state of learning algorithms by measuring the cumulative payoffs of the three actions $A,B,C$ after $t-1$ moves, denoted by $\sigma^A_t, \sigma^B_t, \sigma^C_t$. 
Since $\ftrl$ algorithms only care about the relative differences in the payoffs of the actions,  we track these differences $u_t := \sigma^A_t - \sigma^C_t$ and $v_t := \sigma^B_t - \sigma^C_t$, and refer to $(u_t,v_t) \in \mathbb{R}^2$ as the state of the learner. The move $x_t$ of the optimizer takes the learner's state from $(u_t,v_t)$ to $(u_{t+1},v_{t+1})$.
Recall that a CSP is an extreme point in the menu of an algorithm only if there exists an ordered sequence of optimizer moves combined with the algorithm's responses, that induces this point. $P_T$ can thus be thought of as being generated by a sequence of optimizer actions $x_1,x_2, \cdots, x_T$, in conjunction with the responses of the $\ftrl$ algorithm $\A$.

\sloppy{This sequence of optimizer actions induces a corresponding sequence of learner states $(u_1, v_1), (u_2, v_2), \dots, (u_T, v_T)$.
We refer to the curve generated by joining consecutive states with line segments as the finite-time trajectory associated with the point $P_T$ (for convenience, we will drop the ``finite-time'' identifier and simply call these trajectories, until otherwise specified). 
The part of the trajectory associated with a contiguous subsequence of time steps is referred to a sub-trajectory. It is useful to keep in mind that the finite-time trajectory is entirely a function of the optimizer's actions and the learner's payoff function, and that different learner algorithms can potentially induce significantly different CSPs for the same trajectory.}

We divide our state space $\mathbb{R}^2$ into multiple zones, based on the historical best response and its corresponding advantage in these zones. 
These zones are demarcated by the following lines:

\begin{itemize}
    \item $H_{AB}$ is the set of all points $(u,v)$ where $u=v$ and $u,v \ge 0$.
    \item $H_{BC}$ is the set of all points $(u,v)$ where $v=0$ and $u<0$.
    \item $H_{AC}$ is the set of all points $(u,v)$ where $u=0$ and $v<0$.
    \item $S_{AB}$ is the set of all points $(u,v)$ where $v-u =f(T)$ and $u,v \ge 0$. 
    \item $\overline{S}_{AB}$ is the set of all points $(u,v)$ where $u-v =f(T)$ and $u,v \ge 0$.
    \item $S_{BC}$ is the set of all points $(u,v)$ where $v = f(T)$ and $u<0$.
    \item $\overline{S}_{BC}$ is the set of all points $(u,v)$ where $v = -f(T)$ and $u<0$.
    \item $S_{AC}$ is the set of all points $(u,v)$ where $u = f(T)$ and $v<0$.
    \item $\overline{S}_{AC}$ is the set of all points $(u,v)$ where $u = -f(T)$ and $v<0$.
    
\end{itemize}

We refer to $H_{AB},H_{BC}, H_{AC}$ as the ``hard" boundaries separating the three actions, and refer to $S_{AB}, S_{BC}, S_{AC},\overline{S}_{AB}, \overline{S}_{BC}, \overline{S}_{AC}$ as the ``soft" boundaries.
The separation of $f(T)$ between the corresponding hard and soft boundaries is interpreted an upper bound on the regret of algorithm $\A$ and the separation in cumulative payoff corresponding to the mean-based condition of $\cA$. While we have defined all possible hard and soft boundaries, some of them are not explicitly used in our subsequent analysis. This is due to an asymmetry in the learner payoffs that prevents certain states from being reached, such as points on $H_{AC}$ that are far from the origin~\footnote{This asymmetry is exploited later to prove Lemma~\ref{lem:constant_drift}, a crucial result showing that the state of the algorithm is always moving away from the origin}.

The learner is said to have a ``clear" historical best response action at the start of a given round if the action with the largest cumulative payoff has an advantage of at least $f(T)$ over the other actions. In other words, at state $(u,v)$; $A$ is a clear historical best response when $u > f(T)$ and $u-v > f(T)$, $B$ is a clear historical best response when $v > f(T)$ and $v-u > f(T)$ and $C$ is a clear historical best response when $u <- f(T)$ and $v < - f(T)$. The significance of this definition can be seen by recalling that a $\gamma$-mean based learner plays a clear historical best response at round $t$ with probability at least $1- 2 \gamma(t)$ and plays the other actions with probability at most $\gamma(t)$. In words, a clear historical best response at round $t$ is practically guaranteed to be played by $\cA$, with some $o(1)$ probability mass for other actions.

Figure~\ref{fig:basic_state_space} shows a picture of the state space, with separations for the soft and hard boundaries.

\begin{figure}[htbp]
    \centering
    \includegraphics[width=0.5\textwidth]{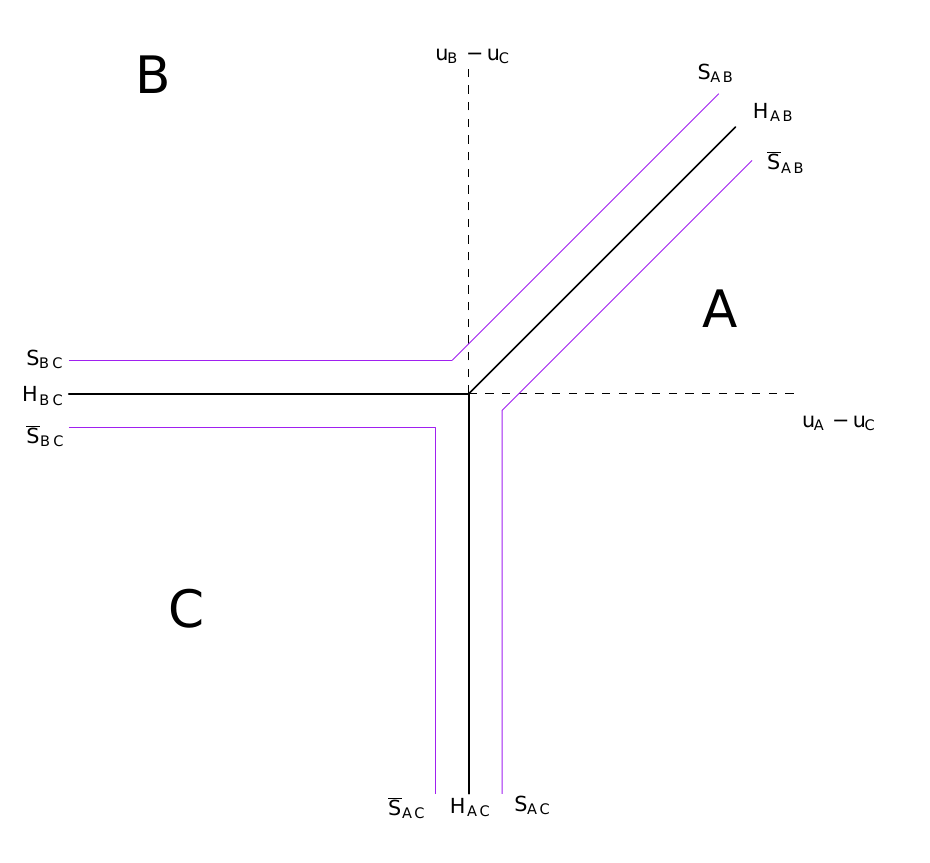}
    \caption{The state space with hard and soft boundaries, in black and purple respectively.}
    \label{fig:basic_state_space}
\end{figure}

 Consider the original point $P$ that we begin with. For $t_1 < t_2$, we call the part of a trajectory from $(u_{t_1}, v_{t_1})$ to $(u_{t_2}, v_{t_2})$ a segment. We break up its trajectory into the following segments:

 \begin{itemize}
     \item (I) An initial segment starting at $(0,0)$ and continuing until the first point where either $S_{AB}$ or $S_{BC}$ is crossed. 
     \item (NC1) A maximal segment that starts within a move of crossing $S_{AB}$ and returns to a point within a move of $S_{AB}$ without ever crossing $S_{BC}$.
     \item (NC2)  A maximal segment that starts within a move of crossing $S_{BC}$ and returns to a point within a move of $S_{BC}$ without ever crossing $S_{AB}$.
     \item (C1)
     A minimal segment that starts within a move of crossing $S_{AB}$ and crosses $S_{BC}$ with its last line segment.
     \item (C2)
     A segment that starts within a move of crossing $S_{BC}$ and crosses $S_{AB}$ with its last line segment.
     \item (F) The final segment is broken up into a maximal C1 or C2 segment followed by a segment starting at one of the soft boundaries and staying in the zone with a unique clear historical best response. 
 \end{itemize}

The different types of segments for a trajectory are depicted in Figure~\ref{fig:basic_trajectory}.

\begin{figure}[htbp]
    \centering
    \includegraphics[width=0.5\textwidth]{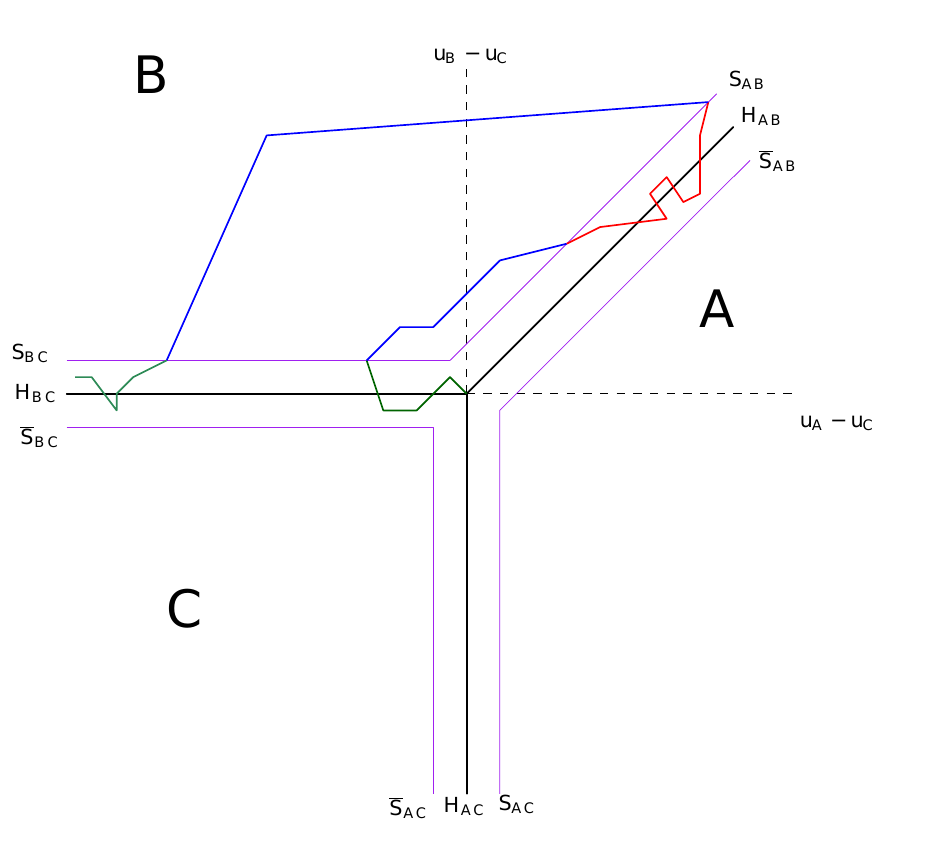}
    \caption{A trajectory is shown with different types of segments having different colors. In order, it contains: an I segment (in green), a C2 segment (in blue), an NC1 segment (in red), a C1 segment (in blue), and a final F segment (in green). 
    }
    \label{fig:basic_trajectory}
\end{figure}

The careful reader might ponder about the problem created by these segments not ending exactly at the boundaries, but overshooting by one time step. This is indeed connected to a deeper problem, connected to trajectories that oscillate back and forth across these boundaries a significant number of times. We address this issue in detail later by proving (using Lemma~\ref{lemma:bound_crossing}) that the number of segments is sublinear; and thus for sake of convenience, we can assume these segments exactly end on the corresponding boundaries. 


 \begin{observation}
     The trajectory of states cannot intersect with $H_{AC}$ after $t = 0$.
 \end{observation}
 \begin{proof}
     For this event to happen, the historical average optimizer action (corresponding to the point of intersection, possibly for a fractional time step) must induce a tie in best-responses between actions $A$ and $C$ - no such optimizer distribution exists for these learner payoffs, ruling out the possibility.
 \end{proof}
 Armed with this observation, it is straightforward to see that this decomposition of trajectories is uniquely defined and can be constructed directly by splitting a trajectory into an initial segment (I) followed by alternating non-crossing (NC) and crossing (C) segments and finishing with a final (F) segment. Since the trajectory never crosses $H_{AC}$; the only way crossing segments move across from $S_{AB}$ to $S_{BC}$ is by going through the region where $B$ is a clear historical best response. The alternating nature of non-crossing and crossing segments implies that the number of non-crossing segments is at most the number of crossing segments plus 1. 


The main constructive step involved in the proof is in showing how to transform the trajectory corresponding to point $P_T$
into another trajectory which spends only sublinear time within the soft boundaries and induces, together with algorithm $\cA$, a CSP $P'_T$ that is a sub-constant distance to $P_T$.

\begin{lemma}
\label{lem:main_trajectory_transformation}
    Given any trajectory inducing a CSP $P_T$ when played against $\cA$ with time horizon $T$, there exists a trajectory that spends all but $o(T)$ time outside the soft boundary and induces a CSP $P'_T$ such that $d(P_T, P'_T) = o(1)$ .
\end{lemma}

The proof is presented later in this section. Armed with such a trajectory, we then show how to construct a point $Q_T \in \MB$ that is $o(1)$ close to $P'_T$, and hence close to $P_T$. Based on the finite-time trajectory inducing $P'_T$, we can construct a valid trajectory $\tau_T$ in the continuous-time formulation referred to in Section~\ref{sec:mb_algos_and_menus}. We do this by taking the trajectory inducing $P'_T$, breaking it up whenever it crosses a hard boundary separating two learner actions, and replacing each segment by the time average action repeated for the length of the trajectory. We also attach the corresponding learner best response action to each segment of this trajectory to complete the triple. Since $P'_T$ spends only sublinear time within the soft boundaries, it is easy to see that the CSP $Q_T := \Prof(\tau_T)$ is sublinearly close to $P'_T$ and hence to $P_T$. Thus, we have a sequence of valid trajectories $\{\tau_T\}$ whose images under the continuous function $\Prof$ are converging to $P$. Since $\MB$ is a closed and bounded set (by virtue of being a polytope, see Theorem~\ref{thm:mb_is_a_polytope}), the point $P$ is contained in $\MB$, completing the proof.

\subsubsection{Proof of Lemma~\ref{lem:main_trajectory_transformation}}

We outline this transformation of the trajectory below, through a sequence of lemmas that prove some critical properties about this class of games and give procedures to transform the different types of segments in the target trajectory.

Based on the learner utilities, we observe that any single move of the optimizer shifts the state of the learner by some constant distance (in say, the $\ell_1$ norm).
\begin{observation}
\label{obs:shifts}
    Starting at state $(u,v)$ and playing optimizer distribution $\alpha(N)$ for one step changes the state to $(u+\xi_1, v +\xi_2)$ with
    \begin{itemize}
        \item  $\xi_1 = 1/2 \pm 2\varepsilon$ and $\xi_2 = 1/3 \pm 2\varepsilon$ when $\alpha(N) = 1$.
        \item  $\xi_1 = -1/2 \pm2 \varepsilon$ and $\xi_2 = -1/6 \pm 2\varepsilon$ when $\alpha(N) = 0$ .
    \end{itemize}
    where $\xi = a \pm b$ is equivalent to stating $\xi \in [a-b,a+b]$.
\end{observation}

Next, we observe that the optimizer's action space in the stage game can be partitioned into contiguous intervals corresponding to which action of the learner is a best response to it. We note that the breakpoints are close to $1/3$ and $2/3$, where the optimizer action is parametrized by the probability with which action $N$ is played.
\begin{observation}
\label{obs:br_intervals}
    For an optimizer distribution $\alpha(N)$; there are numbers $q = 1/3 \pm O(\varepsilon)$ and $p = 2/3 \pm O(\varepsilon)$ such that :
    \begin{itemize}
        \item If $\alpha(N) \in [0,p)$, $C$ is a unique best response.
        \item If $\alpha(N) \in (p,q)$, $B$ is a unique best response.
        \item If $\alpha(N) \in (q,1]$, $A$ is a unique best response.
    \end{itemize}
\end{observation}

These observations lead to an important result -- that the state of the algorithm drifts away from the origin, in the $\ell_1$ distance, by $\theta(t)$ after $t$ time steps. 

\begin{lemma}
There exists constants $C_1, C_2 > 0$ such that every sub-trajectory of length $L$ corresponds to somewhere between $C_1 \cdot L$ and $C_2 \cdot L$ moves, and a constant $B > 0$ such that after $t$ time steps, the trajectory is at least $B \cdot t$-far from the origin in the $\ell_1$ distance. \label{lem:constant_drift}
\end{lemma}

\begin{proof}
Each time step moves the state only by a constant amount in the $\ell_1$ metric (follows directly from Observation~\ref{obs:shifts}). Thus, a segment of some length $L$ must correspond to a contiguous subsequence of at least $C_1 \cdot L$ time steps for some constant $C_1>0$. Showing the second half of the lemma statement will imply a similar upper bound on the number of time steps. To do so, we look at the offsets to the state generated by the two actions $N$ and $Y$, which are $V_1 =  (1/2 \pm 2\varepsilon, 1/3 \pm 2\varepsilon)$ and $V_2 = (-1/2 \pm 2\varepsilon, -1/6 \pm 2\varepsilon)$. The average offset caused by the action of the optimizer over any number of time steps must therefore be some convex combination of $V_1$ and $V_2$ (we use linearity to rewrite the average offset of the actions as the offset generated by the historical average optimizer action). Ignoring the $\varepsilon$ terms for the moment, the line between the points $V_1$ and $V_2$ is $y = x/2 - 1/12$. The $\ell_2$ distance from the origin $(0,0)$ to this line is $\frac{1}{6 \sqrt{5}}$, implying that every point on the line is at least  $\frac{1}{6 \sqrt{10}}$ far from the origin. With the $\varepsilon$ terms included, it is easy to see that this statement is still true for the distance $\frac{1}{6 \sqrt{10}} \pm O(\varepsilon)$. Thus, after $T$ time steps, the state has $\ell_1$ distance at least $\frac{T}{24}$ from the origin, completing the proof.
\end{proof}

The above lemma implies a bound on the number of non-crossing segments in any trajectory.

\begin{lemma}
\label{lemma:bound_crossing}
The number of different segments between rounds $t = \phi$ and $t = T$ is at most $O\left(\frac{T}{\phi}\right)$. 
\end{lemma}

\begin{proof}
    An upper bound the number of crossing segments after $\phi$ time steps implies the same asymptotic upper bound for the number of non-crossing segments. 
    We assume, without loss of generality, that $\phi > f(T)$ since the first crossing segment can only appear after $\Omega(f(T))$ time steps (which is the minimum length of the initial segment). Using Lemma~\ref{lem:constant_drift}, we see that any crossing segment is the largest side of an obtuse angled triangle whose two smallest sides have length $\Omega(\phi)$. Thus, each non-crossing segment has length $\Omega(\phi)$, and must consequently correspond to $\Omega(\phi)$ time steps (since each time step's line segment has constant length). This directly implies the desired result.
\end{proof}

The critical part of our proof is a method to transform non-crossing segments into a segment with identical start and end points, but one that only spends $O(f(T))$ (i.e sublinear) number of time steps within the soft boundaries; and induces almost the same CSP. For a clean exposition that accounts for different non-crossing segments having different length, we unnormalize the CSP stick to this convention for the proof. By this, we mean that the CSP is redefined to be $\sum_{t=1}^T x_t \otimes y_t$ rather than $\frac{1}{T}\sum_{t=1}^T x_t \otimes y_t$.  
A visual version of this transformation is shown in Figure~\ref{fig:altered_nc}.

\begin{figure}[htbp]
    \centering
    \includegraphics[width=0.5\textwidth]{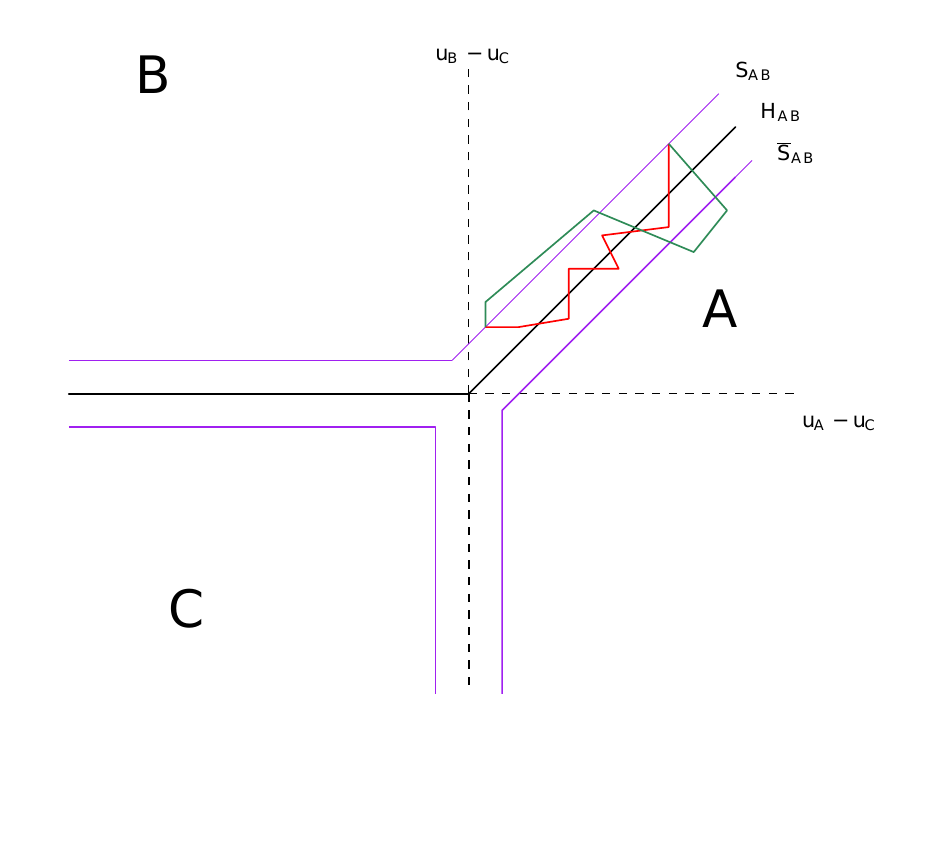}
    \caption{The NC segment (in red) is replaced by a segment (in green) that spends minimal time inside the soft boundary.}
    \label{fig:altered_nc}
\end{figure}

\begin{lemma}
\label{lem:nc_replacement}
Every non-crossing segment can be replaced with a segment beginning and ending in the same place, running for the same number of time steps and spending $O(f(T))$ time within the soft boundaries, with the property that the resulting trajectory induces a CSP which differs by $O(f(T))$ from the starting CSP.
\end{lemma}

\begin{proof}

We will prove this for the soft boundary $S_{AB}$. A similar proof works for $S_{BC}$.

Let $s = x_{t_1:t_2}$ be the segment of play by the optimizer starting at $S_{AB}$ and ending somewhere within the soft boundary of $BC$. It has a resulting CSP of $c$ (the average joint distribution over action pairs measured over this sequence). 

Instead of analyzing algorithm $\A$'s induced CSP $c$ on the segment $s$; we analyze the (unnormalized) CSP $c'$ generated by the action of $\A_{2}$, a learning algorithm which is equivalent to $\A$ except that it only considers actions $A$ and $B$, on the same segment.
The last clause of Lemma~\ref{lem:ftrl_properties} shows that the algorithms $\A$ and $\A_2$ pick almost identical distributions in each round (separated by at most $\frac{\sshift(T)}{T}$) -- thus the CSPs $c$ and $c'$ are sublinearly close, i.e., $||c-c'||_1  = O(f(T))$. Further, this alternative algorithm is also an instance of $\ftrl$, guaranteeing the same properties, at least with respect to the actions $A$ and $B$.

\jon{flagging for myself that this is where we use the last property of FTRL.}

Now, we will describe a new trajectory that approximates $c'$. Recall that $x_t$ is the distribution played by the optimizer at round $t$. Let $(\alpha_A^t,\alpha_B^t)$ be the distribution played by the learning algorithm $\A_2$ at round $t$. Define $D_A := \frac{1}{\alpha_A} \sum_{t=t_1}^{t_2} \alpha_A^t x_t$  and $D_B := \frac{1}{\alpha_B} \sum_{t=t_1}^{t_2} \alpha_B^t x_t$ where $\alpha_A := \sum_{t=t_1}^{t_2} \alpha_A^t$ and $\alpha_B := \sum_{t=t_1}^{t_2} \alpha_B^t$. In other words, $D_{A}$ and $D_{B}$ are the average distributions of the optimizer conditioned upon the learner playing $A$ and $B$ respectively, while $\alpha_{A}$ and $\alpha_{B}$ are the sum of the probabilities (over all the rounds in this segment) with which the optimizer plays $A$ and $B$ respectively.

Since the segment $s$ starts and ends at the soft boundary $S_{AB}$, we know that $A$ and $B$ are both within a cumulative value of $O(f(T))$ to a best response to the distribution $d' = \frac{\alpha_{A}D_{A}+\alpha_{B}D_{B}}{\alpha_{A} +\alpha_{B}}$, i.e the average distribution of the optimizer over the segment of interest, played $\alpha_A+ \alpha_B$ times.

We argue that algorithm $\A_2$ has an upper bound of  $O(f(T))$ on the magnitude of regret in this segment. To see this, we imagine a prefix to this segment that starts at the hard boundary $H_{AB}$ and moves to the soft boundary $S_{AB}$ by playing $K'f(T)$ steps of $Y$ for some constant $K'$~\footnote{We employ this argument since regret upper\ lower bounds are typically only guarantees for sequences starting at time $t=0$.}. Now, since this extended sequence begins at the hard boundary, it is as if we started running $\A_2$ at $t=0$ at the start of this sequence (by adding a constant offset to the cumulative payoffs of $A$ and $B$, a legal action, by Lemma~\ref{lem:ftrl_properties}). Thus, we get the regret guaranteed to be in the interval $[-O(f(T)), O(f(T))]$ for the extended sequence (from the properties of $\ftrl$ algorithms \jon{flagging for myself that this is where we use the bound on negative-regret}), and observe that the regret for the ``true" segment is therefore guaranteed to be in $[- H f(T), H f(T)]$ for some constant $H> 0$.

First, we observe that the no-regret property with two actions implies no-swap regret. Quantitatively, we have:

\[\alpha_B (\mu_{\mathcal{L}}(D_B,A) - \mu_{\mathcal{L}}(D_B,B)) \le H f(T) \qquad \text{(based on regret bounds against $A$)}\]
\[\alpha_A (\mu_{\mathcal{L}}(D_B,B) - \mu_{\mathcal{L}}(D_B,A)) \le -H f(T)  \qquad \text{(based on regret bounds against $B$)}\]

We show the stronger property that $A$ and $B$ are both close to a best response to both $D_{A}$ and $D_{B}$ (the observation above only shows that $A$ is a near best response to $D_A$ and $B$ is a near best response to $D_B$). Assume for contradiction that $\alpha_A(\mu_{\mathcal{L}}(D_B,A) - \mu_{\mathcal{L}}(D_B,B)) > 2 H f(T)$.
\begin{align*}
\mu_{\mathcal{L}}(c') & = \alpha_{A}\mu_{\mathcal{L}}(D_{A},A) + \alpha_{B}\mu_{\mathcal{L}}(D_{B},B) \\
& > \alpha_{A}\mu_{\mathcal{L}}(D_{A},B) + \alpha_{B}\mu_{\mathcal{L}}(D_{B},B) + 2H f(T) \qquad \text{(by assumption)}\\
& = (\alpha_A+\alpha_B)\mu_{\mathcal{L}}(d',B) + 2H f(T) \qquad \text{(using linearity of the payoffs)} \\
& \geq \max_{S \in \{A,B\}}(\alpha_A+\alpha_B)\mu_{\mathcal{L}}(d',s) - H f(T) +  2Hf(T) \text{(by the fact that $B$ is close to a best response to $d'$)}\\  
\end{align*}

This is a contradiction for sufficiently large $T$, as $\A_{2}$ has an average regret lower bound of $-H f(T)$. A similar argument can be made with $A$ and $B$ exchanged. Thus, we have :

\[\alpha_A(\mu_{\mathcal{L}}(D_B,A) - \mu_{\mathcal{L}}(D_B,B)) \le 2 H f(T)\]
\[\alpha_B(\mu_{\mathcal{L}}(D_A,B) - \mu_{\mathcal{L}}(D_A,A)) \le 2 H f(T)\]

These inequalities are key to constructing our transformation. We split our analysis into four cases, depending on whether or not actions $A$ and $B$ are both played for a significant amount of time in this segment.


We start by assuming that $\alpha_{A} \geq 8 K_1 f(T)$, $\alpha_{B} \geq 8 K_2f(T)$ for some positive constants $K_1$ and $K_2$ which will be defined shortly. This is the most important case in the analysis, and the other cases (where this assumption is not true) simply follow by exploiting sublinear differences in the construction.

The new sequence is as follows: from the hard boundary, play $N$ until $u_{t} - v_t = 3Hf(T)$. By Observation~\ref{obs:br_intervals}, this will take $K f(T)$ rounds for some constant $K>0$. Then, play $D_{A}$ for $\alpha_{A} - 8K f(T)$ rounds. Then, play $Y$ until $u_{t} - v_t = -3Hf(T)$. We claim that takes only $K_{2}f(T)$ rounds for some constant $K_{2} > 0$ - this is so because the cumulative effect of playing $D_A$ for $\alpha_A$ rounds causes a difference of at most $2Hf(T)$ in the payoffs of actions $A$ and $B$ over these rounds. Therefore, at the time step $t$ at the end of playing $D_A$, $u_t - v_t \le 5Hf(T)$. Then, play $D_{B}$ for $\alpha_{B} - 8K_{2} f(T)$ rounds.

Let $D_{A}(Y)$ be the fraction of the time that $Y$ is played in $D_{A}$, and define $D_{A}(N)$, $D_{B}(Y)$, and $D_{B}(N)$ accordingly. Then, let us consider the number obtained by adding up the probability of playing $N$ in all the rounds described so far. This number $n_N$ is upper bounded by $D_{A}(N)(\alpha_{A} - 8K f(T)) + D_{B}(N)(\alpha_{B} - 8K_{2} f(T)) + K f(T)$.  Likewise, the number of times $n_Y$ we have played $y$ so far is upper bounded by $D_{A}(Y)(\alpha_{A} - 8K f(T)) + D_{B}(Y)(\alpha_{B} - 8K_{2} f(T)) + K_{2} f(T)$.

To finish our sequence, we will play $N$ for $n'_N := d'(N) (\alpha_A + \alpha_B) - n_N$ times, which is lower bounded by  $d'(N) (\alpha_A + \alpha_B)- (D_{A}(N)(\alpha_{A} - 8K f(T)) + D_{B}(N)(\alpha_{B} - 8K_{2} f(T)) + K f(T))$ and then play $Y$ for $n'_Y:= d'(Y)(\alpha_A + \alpha_B) - n_Y$ rounds, which is lower bounded by  $d'(Y)(\alpha_A + \alpha_B) - (D_{A}(Y)(\alpha_{A} - 8K f(T)) + D_{B}(Y)(\alpha_{B} - 8K_{2} f(T)) + K_{2} f(T))$ rounds. We show that this is a legal operation by proving that $n'_N, n'_Y \ge 0$. It suffices to show their lower bounds are non-negative.

Note that by the fact that $A$ and $B$ are near best responses for both $D_{A}$ and $D_B$ and using Observation~\ref{obs:br_intervals}, $D_{A}(N), D_B(N), D_A(Y), D_B(Y) > \frac{1}{4}$.
\begin{align*}
n'_N &= D_{A}(N)\alpha_{A} + D_{B}(N)\alpha_{B} - D_{A}(N)(\alpha_{A} - 8K f(T)) - D_{B}(N)(\alpha_{B} - 8K_{2} f(T)) - K f(T) \\
& =  D_{A}(N)8K f(T) + D_{B}(N)8K_{2} f(T) - K f(T)  \\
& \geq 2K f(T) +2K_{2} f(T) - K f(T) \\
& = K f(T) +2K_2 f(T) > 0 \\
\end{align*}
A similar analysis for $Y$ shows that
\begin{align*}
n'_Y & = D_{A}(Y)\alpha_{A} + D_{B}(Y)\alpha_{B} - D_{A}(Y)(\alpha_{A} - 8K f(T)) - D_{B}(Y)(\alpha_{B} - 8K_{2} f(T)) - K_{2} f(T)  \\
& =  D_{A}(Y)8K f(T) + D_{B}(Y)8K_{2} f(T) + 8K_{2} - K_{2} f(T)  \\
& \geq 2K f(T) +2K_{2} f(T) - K_{2} f(T) \\
& = 2K f(T) + K_{2} f(T) > 0\\
\end{align*}

A similar analysis can also be used to upper bound $n'_X$ and $n'_Y$ by $8(K_1+K_2)f(T)$. Since $n'_X, n'_Y \ge 0$, our construction is valid, and we have played $N$ and $Y$ the same total frequency across the rounds as compared to the original segment. Therefore the endpoint of the new segment will be exactly the same~\footnote{The state at the end of $t$ rounds only depends on the average marginal distribution of the optimizer and is invariant to reordering of the optimizer's play.} and the number of rounds will be as well. Further, we argue that the CSP of this segment is very close to the CSP of the original segment. In the new segment, because we have moved $3Hf(T)$ away from the hard boundary and because we can only move back by $2Hf(T)$ as we play $D_{A}$, the learner will play $A$ the whole time against this (except for a sublinear fraction of the rounds), since we are in the clear historical best response zone for action $A$~\footnote{Here, we use the fact that $f$ is monotone increasing.}. For the same reason, the learner plays $B$ against $D_{B}$ when we play it (except for a sublinear fraction of the rounds). The other drivers of change in the CSP are the $Kf(T)$ rounds of $\alpha_A$ that we do not play $D_A$ (similar for $D_B$), the $K_2 f(T)$ rounds of playing $Y$ that are required to go from the region where $A$ is the clear historical best response to where $B$ is the clear historical best response and the $8(K_1+K_2)(f(T))$ rounds of play at the end to reach the same end point as the original segment. Even if all these changes go in the same direction, the net effect on the CSP is at most $O(f(T))$ deviation in the 
$\ell_1$ distance since each round can only effect a change of $O(1)$ (the constants do not depend on $T$ and are entirely determined by the learner values in the game).

Although we have assumed that $\alpha_{A} \geq 8 K_1 f(T)$, $\alpha_{B} \geq 8 K_2f(T)$, the same analysis can still be done if this assumption is not true. We simply do not play $D_A$/$D_B$ in the construction if the corresponding frequency $\alpha_A$/ $\alpha_B$ is very small. With minor modifications, the same guarantees can be achieved without this assumption.

\end{proof}

Armed with these tools, we focus now on generating $P'_T$ from $P_T$ for a finite time horizon $T$. The procedure starts by decomposing the trajectory of the point $P_T$ in the state space defined above into I,C1,C2,NC1,NC2 and F segments. Each segment is then replaced with another segment (by a suitable modification of the optimizer actions)  with the same start and end point so that the new trajectory has the property that only sublinear time is spent outside the regions with a clear historical best response. We maintain the invariant that the unnormalized CSP induced by this new sequence/ trajectory against algorithm $\A$ is only sublinearly different (i.e $o(T)$ distant in the $\ell_1$ metric) from $P_T$. 


Dealing with the Crossing segments of the type C1 or C2 is the easiest part of the transformation - these segments already spend their entire duration  in a region with a single clear historical best response, so we can leave them untouched. To deal with non-crossing segments (C1/C2); we make use of Lemma~\ref{lem:nc_replacement} to replace them one by one with a segment that has the same start and end point, operates for exactly the same time duration and results in a trajectory that spends time at most $O(f(T))$ outside the clear historical best response region and generates a CSP shifted by only $O(f(T))$ from the starting trajectory's CSP. A concern with this sequence of transformations is that sufficiently large number of crossing segments can result in linear drift from the original CSP/ linear time spent outside the clear historical best response zone. To address this issue, we split the analysis into two parts. Set $g(T) = \sqrt{f(T) \cdot T}$, since $f(T)$ is asymptotically sublinear, so is $g(T)$. 
We split the time horizon $T$ into the first $g(T)$ time steps and the remaining time steps. The play in the first part can only affect the CSP by $O(g(T))$, further the total time spent outside the clear historical best response zone in the first part is also upper bounded by $g(T)$. Using Lemma~\ref{lemma:bound_crossing}, there are at most $O\left(\frac{T}{g(T)} \right) = O \left( \sqrt{ \frac{T}{f(T)}} \right)$ crossing segments in the second half. Therefore, the total drift (in $\ell_1$ distance) from the original CSP is upper bounded by $O \left( f(T) \cdot  \sqrt{ \frac{T}{f(T)}} \right) = O(g(T)) =   o(T)$. Similarly, in the time steps corresponding to the crossing segments, the final trajectory spend time at most $o(T)$ away from the boundary. 

A similar analysis can be done for the initial (I) and final (F) segment, by only replacing the non-crossing part of their trajectories using Lemma~\ref{lem:nc_replacement} and contributing an additional $o(T)$ drift away from the original CSP, and $o(T)$ additional time spent away from the clear historical best response zone. 

\section{Omitted proofs}
\label{app:omitted_proofs}

\subsection{Proof of Lemma~\ref{lem:menu_alg_equiv}}

\begin{proof}[Proof of Lemma~\ref{lem:menu_alg_equiv}]
We will first show that $V_{L}(\M(\cA), u_O) \leq V_{L}(\cA, u_O)$. To see this, let $\csp^*$ be the CSP in $\M(\cA)$ achieving the value $V_L(M(\cA), u_O)$. There exists a sequence $\csp^1, \csp^2, \dots$ with $\csp^{T} \in \M(\cA^T)$ which converges to $\csp^*$. For any menu $\M$ and CSP $\csp \in \M$, let $\eps_{O}(\M, \csp) = (\max_{\csp' \in \M} u_{O}(\csp')) - u_{O}(\csp)$ denote how far $\csp$ is from optimal (for the optimizer) in $\M$. Since $u_O$ is continuous in $\csp$, the function $\eps_{O}$ is continuous, and therefore the sequence $\eps_{O}(\M(\cA^T), \csp^{T})$ converges to $\eps_{O}(\M(\cA), \csp^{*}) = 0$. This means that for any $\eps$ there exists a $T(\eps)$ such that $\eps_{O}(\M(\cA^{T}), \csp^{T}) < \eps$ for any $T \geq T(\eps)$. But this implies there is a transcript in $\cX(\cA, u_O, T, \eps)$ that induces $\csp^{T}$, and in particular, that $V_{L}(\cA, u_O, T, \eps) \geq u_L(\csp^{T})$. In particular, for any $\eps > 0$, $\liminf_{T\rightarrow \infty} V_{L}(\cA, u_O, T, \eps) \geq \liminf_{T\rightarrow\infty} u_L(\csp^{T}) = u_L(\csp^{*}) = V_{L}(\M(\cA), u_O)$. By taking $\eps$ to $0$, it follows that $V_{L}(\M(\cA), u_O) \leq V_{L}(\cA, u_O)$.

We now must show $V_{L}(\cA, u_O) \leq V_{L}(\M(\cA), u_O)$. Let $V = V_{L}(\cA, u_O)$; by the definition \eqref{eq:learner_value} of $V_L(\cA, u_O)$, there must exist a sequence $(T_1, \eps_1), (T_2, \eps_2), \dots$ of strictly increasing $T_i$ and strictly decreasing $\eps_i$ where $V_L(\cA, u_O, T_i, \eps_i)$ converges to $V$. For each $(T_i, \eps_i)$, let $\csp_i \in \M(\cA^{T_i})$ be the CSP induced by a $\eps_i$-optimal point in $\cX(\cA, u_O, T_i, \eps_i)$, i.e., one where $u_L(\csp_i) \geq V_L(\cA, u_O, T_i, \eps_i) - \eps_i$. Since the $\csp_i$ all belong to the bounded closed set $\Delta_{mn}$, there must be a convergent subsequence of these CSPs that converge to a $\csp' \in \M(\cA)$ with the property that $u_L(\csp') = V$. It follows that $V = V_{L}(\cA, u_O) \leq V_{L}(\M(\cA), u_O)$.
\end{proof}

\subsection{Proof of Lemma~\ref{lem:char_upwards_closed}}

\begin{proof}[Proof of Lemma~\ref{lem:char_upwards_closed}]
To prove that some convex set $\cM'$ is an asymptotic menu, it is sufficient to prove there exists a consistent learning algorithm $\cA'$ that has the asymptotic menu $\cM(\cA') = \cM'$. The learning algorithm $\cA'$ is as follows: let $C$ be the smallest integer that is larger than $\sqrt{T}$. Note that this grows asymptotically with $\sqrt{T}$. The learner now construct the set $S$ of at most $O(C^{mn})$ CSPs in $\Delta_{mn}$ containing all CSPs whose coordinates are all integer multiples of $\frac{1}{C}$, and that are within $\ell_{\infty}$-distance $\frac{1}{C}$ of $\cM$. Note that every CSP $\csp$ in $\cM$ is within $\ell_{\infty}$-distance $\frac{1}{C}$ of some CSP of this form.


We can now assign each $\csp' \in S$ a unique sequence of moves which the optimizer can play in order to communicate this CSP to the learner. As there are $C^{mn}$ points and the optimizer can play any of $m$ actions, the optimizer can communicate this in at most $\log(C^{mn}) = mn\log(C)$ rounds, using constant number of bits each round. For the first $mn\log(C)$ rounds of the game, the learner will play their first move, and the optimizer will communicate their preferred $\csp'$. 

After this, the learner will identify the $\csp' \in S$ associated with this sequence, and ascribe a joint cycle of moves they can play with the optimizer $(P(x),P(y))$ that will result in $\csp'$ being played in the limit as the time horizon $T$ goes to infinity. In particular, for some CSP $\csp'$, the sequence will play out in cycles of length $C$, and in each cycle, the learner and optimizer will jointly play the move pair $(i,j)$ for $C \cdot \csp(i,j)$ rounds (recall that this is an integer); iterating over move pairs in lexicographical order. If the opponent ever fails to follow $P(j)$, the optimizer will immediately begin playing the algorithm that converges to $\cM$, ignoring the history of play that happened before this failure occurred. 

First we will show that $\cM_{A'} \subseteq \cM'$. \jon{note: rewrote a little bit of this proof / added more detail in places, make sure I didn't mess up.} Consider, for a game of length $T$, any sequence that the optimizer could play against $A$. During the first $mn\log(C)$ rounds the learner and optimizer will agree on some $\csp'_{\cM'} \in \Delta_{mn}$ and associated sequence $(P(x),P(y))$ s.t. $\lim_{T \rightarrow \infty}(\sum_{t= 1}^{T}P(x)_{t} \otimes P(y)_{t}) = \csp'_{\cM'}$. Then let $t' \geq mn\log(C)$ be the first time that the follower does not play according this sequence. Then, the resulting CSP is of the form 

\begin{align*}
& \frac{1}{T}\sum_{t=1}^{T}x_{t}\otimes y_{t} \\
& =\frac{1}{T}\sum_{t=1}^{mn\log(C)}x_{t}\otimes y_{t} + \frac{1}{T}\sum_{i=1}^{t' - mn\log(C) -1}P(x)_{i} \otimes P(y)_{i} + \frac{1}{T}x_{t'}\otimes y_{t'} + \frac{1}{T}\sum_{t=t'+1}^{T}x_{t}\otimes y_{t} \tag{By the fact that the players coordinate from $mn\log(C)+1$ to $t'-1$}\\
& =\frac{1}{T}\sum_{t=1}^{mn\log(C)}x_{t}\otimes y_{t} + \frac{1}{T}\sum_{i=1}^{t' - mn\log(C) -1}P(x)_{i} \otimes P(y)_{i} + \frac{1}{T}x_{t'}\otimes y_{t'} + \frac{1}{T}\sum_{t=t'+1}^{T}A^{T-t'}(x_{t},...x_{t'+1})\otimes x_{t} \tag{By the fact that $A'$ defects to playing $A$ after round $t'$}\\
\end{align*}




We will now show that, for any $\eps > 0$, for sufficiently large $T$ the above CSP is within $\ell_{\infty}$-distance $\eps$ of a point in $\cM'$. It then follows that $\cM_{A'} \subseteq \M'$.

To do so, first note that if we take $T$ large enough s.t. $\eps T > 10mn\log(C)$, then (simply by modifying $2mn\log(C)$ terms of magnitude at most $1/T$) the above CSP is within distance $\eps/5$ of the CSP

$$\tilde{\csp} = \frac{1}{T}\sum_{i=1}^{t'}P(x)_{i} \otimes P(y)_{i} + \frac{1}{T}\sum_{t=t'+1}^{T}A^{T-t'}(x_{t},...x_{t'+1})\otimes x_{t}.$$

\noindent
We will also define

\begin{equation}\label{eq:Mpdef}
\tilde{\csp}_{\cM'} = \frac{1}{t'}\sum_{i=1}^{t'}P(x)_{i}
\end{equation}

\noindent
and

\begin{equation}\label{eq:Mdef}
\tilde{\csp}_{\cM} = \frac{1}{T-t'}\sum_{t=t'+1}^{T}A^{T-t'}(x_{t},...x_{t'+1})\otimes x_{t}.
\end{equation}

By doing so we can write $\tilde{\csp} = \frac{t'}{T}\csp_{\cM} + (1 - \frac{t'}{T})\csp_{\cM'}$. We will now show that there exist $\csp_{\cM} \in \cM$ and $\csp_{\cM'} \in \cM'$ such that, for sufficient large $T$, 

\begin{align}
\frac{t'}{T}||\tilde{\csp}_{\cM'} - \csp_{\cM}||_{\infty} &\leq \eps/5,\label{eq:Mptilde}\\
\frac{T - t'}{T}||\tilde{\csp}_{\cM} - \csp_{\cM'}||_{\infty} &\leq \eps/5\label{eq:Mtilde}.
\end{align}

Together, these two statements imply that $\tilde{\csp}$ is within distance $2\eps/5$ from a convex combination of two CSPs in $\M'$, and therefore itself within distance $2\eps/5$ from $\M'$.

To prove \eqref{eq:Mptilde}, note that whenever $t'$ is a multiple of $C$, $\tilde{\csp}_{\cM}$ is equal to our selected point $\csp' \in S$. If $t' > 10C/\eps$, then by splitting the sum in \eqref{eq:Mpdef} at the largest multiple of $C$, we have that $||\tilde{\csp}_{\cM'} - \csp'|| \leq \eps/10$. Since $\csp'$ is in turn within distance $1/C$ of a CSP in $\cM'$, it follows that $||\tilde{\csp}_{\cM'} - \csp_{\cM'}||_{\infty} \leq \eps/10 + 1/C$ for this CSP $\csp_{\cM'}$; for sufficiently large $T$, $C \leq 10/\eps$ and we have established \eqref{eq:Mpdef}. On the other hand, if $t' \leq 10C / \eps$, then $t'/T \leq 10C/(\eps T)$, which again is at most $\eps/5$ for sufficiently large $T$ (in particular, some $T = \Omega(\eps^{-4})$ suffices).

Similarly, to prove \eqref{eq:Mtilde}, note that $\tilde{\csp}_{\cM}$ belongs to the finite-time menu $\cM(A^{T-t'})$. Let $d(t)$ equal the Hausdorff distance between $\cM(A^{t})$ and $\cM$; since the sequence $\cM(A^{t})$ converges in the Hausdorff to $\cM$, $d(t) = o(1)$ is sub-constant in $t$. In particular, for sufficiently large $T$, we have that $d(\sqrt{T}) \leq \eps/10$. This directly implies that if $T - t' \geq \sqrt{T}$, then \eqref{eq:Mtilde} holds. But if $T - t' < \sqrt{T}$, then the LHS of \eqref{eq:Mtilde} is at most $1/\sqrt{T}$, which is also at most $\eps/10$ for sufficiently large $T$.  This concludes the proof that $\cM(\cA') \subseteq \cM'$.

Now we will prove that $\cM' \subseteq \cM(\cA')$. To see this, consider any point $\csp \in \cM'$. We will show that it is possible to achieve this point in the limit against $A'$. The way that the Optimizer will achieve this point is by picking the CSP $\csp'_{net(T)}$ in the $\frac{1}{C}$-net $S$ which is closest in distance to their desired point in $\cM'$ and communicating with the Learner that this is what they want to achieve. This defines a limit cycle sequence $(P(x),P(y))$ which the two players will play together. Then, as long as the Optimizer doesn't defect from this sequence, the CSP of the final point is

\begin{align*}
& \frac{1}{T}\sum_{t=1}^{T}x_{t} \otimes y_{t} = \frac{1}{T}\sum_{t=1}^{mn\log(C)}x_{t} \otimes y_{t} + \frac{1}{T}\sum_{i=1}^{T-mn\log(C)} P(x)_{i} \otimes P(y)_{i}
\end{align*}

Taking the limit:

\begin{align*}
& \lim_{T \rightarrow \infty} \left(\frac{1}{T}\sum_{t=1}^{mn\log(C)}x_{t} \otimes y_{t} + \frac{1}{T}\sum_{i=1}^{T-mn\log(C)} P(x)_{i} \otimes P(y)_{i}\right) \\
& =  \lim_{T \rightarrow \infty} \left(\frac{1}{T}\sum_{i=1}^{T-mn\log(C)} P(x)_{i} \otimes P(y)_{i}\right) \\
& =  \lim_{T \rightarrow \infty} \left(\frac{1}{T}\sum_{i=1}^{T} P(x)_{i} \otimes P(y)_{i}\right)  =  \lim_{T \rightarrow \infty} \csp'_{net(T)} \\
& = \csp^{*}
\end{align*}

Recall that $\csp^{*}$ was our desired point in $\cM'$. Therefore $\cM' \subseteq \cM_{A'}$, completing our proof.
\end{proof}

\subsection{Proof of Theorem~\ref{thm:nr_nsr_containment}}

\begin{proof}[Proof of Theorem~\ref{thm:nr_nsr_containment}]
Fix some $u_L$. Let $\M_{NR}(\varepsilon)$ denote the relaxation of the set $\M_{NR}$ where the no-regret constraint is relaxed to:

$$\sum_{i\in[m]}\sum_{j\in[n]} \csp_{ij}u_{L}(i, j) \geq \max_{j^{*} \in [n]}\sum_{i\in[m]}\sum_{j\in[n]} \csp_{ij}u_{L}(i, j^*) - \varepsilon.$$

Let $\A$ be a no-regret algorithm. Then, for all $(u_L, u_O)$, its average total regret over $T$ rounds is upper bounded by some function $f(T)$ such that $\lim_{T \rightarrow \infty} \frac{f(T)}{T} = 0$. 
For each time horizon $T$,  define the set $S_T := M_{NR}\left( \varepsilon_T \right)$ where $\eps_T = \min\left(\eps_{T-1}, \frac{f(T,n)}{T}\right)$ (so that the sequence $\eps_T$ is non-increasing). It follows that $S_1 \supseteq S_2 \supseteq \cdots$ with $\lim_{T \rightarrow \infty} S_T = \M_{NR}$ (since $\lim_{T \rightarrow \infty}  \varepsilon_T = 0$). But it also follows from the no-regret constraint that $\M(A^{T}) \subseteq S_T$ for each $T$, and thus, $\M(\A) \subseteq \M_{NR}$.

A similar argument works for no-swap-regret through the use of the corresponding relaxation.
\end{proof}

\subsection{Proof of Lemma \ref{lem:nsr_within_all_nr}}

\begin{proof}[Proof of Lemma \ref{lem:nsr_within_all_nr}]
We will show that every CSP of the form $x \otimes y$ with $x \in \Delta_m$ and $y \in \BR_L(x)$ belongs to $\M(\cA)$. By Lemma \ref{lem:nsr_characterization}, this implies that $\M_{NSR} \subseteq \M(\A)$.

First, consider the case where $y$ is the unique best response to $x$, i.e., $\BR_L(x) = \{y\}$. By Theorem \ref{thm:characterization}, the asymptotic menu $\M(\A)$ must contain a point of the form $x \otimes y'$ for some $y' \in \Delta_n$. If $y' \not\in \BR_L(x)$, then this CSP has high regret (it violates \eqref{eq:no-regret-constraint}) and contradicts the fact that $\A$ is no-regret. It follows that we must have $y' = y$.

If $y$ is not a unique best-response to $x$, then we will find a point $x'$ close to $x$ such that $y$ is a unique best response to $x'$. To prove we can do this, note that since $y$ is neither weakly dominated (by assumption) nor strictly dominated (since it is a best-response), there must exist a strategy $x^* \in \Delta^m$ for the optimizer which uniquely incentivizes $y$. Thus, for sufficiently small $\delta > 0$, the optimizer strategy $x(\delta) =(1-\delta) x  + \delta x^*$ uniquely incentivizes $y$ as a best response, and thus $x(\delta) \otimes y \in \M(\A)$. Since $\M(\A)$ is closed, taking the limit as $\delta \rightarrow 0$ implies that $x \otimes y \in \M(\A)$, as desired.
\end{proof}

\subsection{Characterization of $\M_{NR}$}

In Lemma \ref{lem:nsr_characterization}, we show that the no-swap-regret menu can be alternatively characterized as the convex hull of all ``best response'' CSPs, of the form $x \otimes y$ with $x \in \Delta_m$ and $y \in \BR_L(x)$. In this subsection, we prove a similar characterization of the extreme points of $\M_{NR}$: namely, $\M_{NR}$ is the convex hull of all CSPs with exactly zero regret.

Specifically, for a CSP $\csp$ define


$$\Reg(\csp) = \max_{j^{*} \in [n]}\sum_{i\in[m]}\sum_{j\in[n]} \csp_{ij}u_{L}(i, j^*) - \sum_{i\in[m]}\sum_{j\in[n]} \csp_{ij}u_{L}(i, j).$$

We have the following lemma.

\begin{lemma}\label{lem:mnr-characterization}
$\M_{NR}$ is the convex hull of all CSPs $\csp$ with $\Reg(\csp) = 0$.
\end{lemma}
\begin{proof}
We can define $\M_{NR}$ as the intersection of the simplex $\Delta_{mn}$ and the set $S$ of points $x \in \Rset^{mn}$ where $\Reg(x) \leq 0$ (note that we can extend the definition of $\Reg(\csp)$ above to a convex function for all $x \in \Rset^{mn}$). By definition, all CSPs $\csp$ with $\Reg(\csp) = 0$ are contained in $\M_{NR}$. All extreme points $\csp$ of $\M_{NR}$ must lie either on the boundary of $S$ (where it is guaranteed that $\Reg(\csp) = 0$), or must themselves be an extreme point of $\Delta_{mn}$. But the extreme points of $\Delta_{mn}$ correspond to pure strategies of the form $\csp = i \otimes j$. If such a strategy has $\Reg(\csp) \leq 0$, it must also have $\Reg(\csp) = 0$ (since replacing $j^* = j$ leaves the learner's utility unchanged).

It follows that all extreme points of $\M_{NR}$ have zero regret.
\end{proof}

Lemma \ref{lem:mnr-characterization} should not be interpreted as implying that \emph{every} point in $\M_{NR}$ has exactly zero regret; indeed, many points in the interior of $\M_{NR}$ may have negative regret $\Reg(\csp) < 0$. In particular, $\Reg$ is a convex function of $\csp$, and it may be the case that the convex combination of two CSPs may have less regret than either individual CSP. In contrast, the swap regret function is always non-negative, so $\M_{NSR}$ does exactly contain the points with zero swap regret.

\subsection{Proof of Lemma~\ref{lem:sub-mean-based}}

\begin{proof}[Proof of Lemma~\ref{lem:sub-mean-based}]

We will do this in two steps. First, call a segment $(x_i, t_i, b_i)$ of a trajectory \emph{non-degenerate} if $b_i$ is a strict best response of the learner throughout the entire segment, i.e., $\BR_{L}(\ox_{i-1}) = \BR_{L}(\ox_i) = \{b_i\}$. We say a trajectory itself is $\eps$-non-degenerate if at least a $1-\eps$ fraction of the segments of the trajectory (weighted by duration $t_i$) are non-degenerate. We begin by claiming that for any valid trajectory $\tau$ and $\eps > 0$, there exists an $\eps$-non-degenerate trajectory $\tau'$ such that $||\Prof(\tau) - \Prof(\tau')||_{1} \leq \eps$. In other words, for any CSP induced by $\tau$, we can always approximately induce it by a trajectory which is almost always non-degenerate.

To prove this, if $\tau = \{(x_1, t_1, b_1), \dots, (x_k, t_k, b_k)\}$, we will construct a $\tau'$ of the form 

$$\tau' = \{ (y_1, \delta_1), (x_1, t_1, b_1), (y_2, \delta_2), (x_2, t_2, b_2), \dots, (y_{k}, \delta_{k}), (x_k, t_k, b_k)\},$$ 

\noindent
where $\tau'$ is formed by interleaving $\tau$ with a sequence of short ``perturbation'' segments $(y_i, \delta_i)$. The goal will be to choose these segments such that each primary segment of the form $(x_i, t_i, b_i)$ is non-degenerate, while keeping the duration $\delta_i$ of each perturbation segment small. Note that we haven't specified the learner's response for these perturbation segments -- the intent is that since these segments are small, they will not significantly contribute to $\Prof(\tau)$ and we can ignore them (in fact, it is possible that they will cross best-response boundaries, in which case we can split them up into smaller valid segments). 

Now, let $\overline{x'}_{i-}$ and $\overline{x'}_{i+}$ be the time-weighted average strategies of the optimizer right before and right after segment $(x_i, t_i, b_i)$ in $\tau'$, respectively. Let $\overline{y}_i = (\sum_{j=1}^{i} y_j\delta_j) / (\sum_{j=1}^{i} \delta_{j})$ be the time-weighted average strategy of the first $i$ perturbations. Note that $\overline{x'}_{i-}$ is a non-trivial convex combination of $\overline{x}_{i-1}$ and $\overline{y}_i$. Since $b_i \in \BR_L(\ox_{i-1})$, if $b_i$ is a strict best-response to $\overline{y}_i$ (i.e., $\BR_{L}(\overline{y}_i) = \{b_i\}$), then this would imply that $b_i$ is a strict best-response to $\overline{x'}_{i-}$. Likewise, $\overline{x'}_{i+}$ is a non-trivial convex combination of $\ox_{i}$ and $\overline{y}_i$, so this would also imply that $b_i$ is a strict best-response to $\overline{x'}_{i+}$. Altogether, if $\overline{y}_i$ strictly incentivizes action $b_i$, then segment $(x_i, t_i, b_i)$ is non-degenerate.

But note that since no action for the learner is weakly dominated, we can construct a sequence of perturbations where each $\overline{y}_i$ strictly incentivizes action $b_i$ (simply set $y_i$ equal to an action which strictly incentivizes $b_i$, and make $\delta_i$ big enough compared to preceding $\delta$). But also, if a given sequence $\{(y_1, \delta_1), \dots, (y_k, \delta_k)\}$ of perturbations has this property, so does any scaling $\{(y_1, \lambda\delta_1), \dots, (y_k, \lambda\delta_k)\}$ of this sequence. So we can scale this sequence so that $\sum_{i=1}^{k} \delta_{i} \leq 0.5 \cdot \eps \cdot \sum_{i=1}^{k} t_i$. 

Doing this guarantees that the trajectory is $\eps$-non-degenerate. It also guarantees that $\Prof(\tau') = (1 - \eps/2)\Prof(\tau) + (\eps/2)\csp_{\delta}$ for some other CSP $\csp_{\delta} \in \Delta_{mn}$; from this, we have that $||\Prof(\tau) - \Prof(\tau')||_{1} \leq (\eps/2) (||\Prof(\tau)||_1 + ||\csp_{\delta}||_1) = \eps$. Our desired claim follows.

In the second step, we will show that given any $\eps$-non-degenerate trajectory $\tau$, there exists a CSP $\csp$ in $\M$ with the property that $||\csp - \Prof(\tau)||_1 \leq 2\eps$. Combining this with the previous claim, this implies that for any valid trajectory $\tau$, there exists a sequence of CSPs in $\M$ converging to $\Prof(\tau)$, and thus $\Prof(\tau) \in \M$ (and $\M_{MB} \subseteq \M$). 

To prove this, we simply play out the trajectory. Without loss of generality, assume that for $\tau$, the total duration $\sum_{i=1}^{k} t_i = 1$ (if not, scale it so that this is true). The adversary will play each action $x_i$ for $\lfloor t_iT \rfloor$ rounds. 

Now, assume $\cA$ is a mean-based algorithm with a specific rate $\gamma(T)$. Assume segment $i$ of $\tau$ is non-degenerate and consider the action of the learner at round $t = (t_1 + t_2 + \dots + t_{i-1} + \lambda t_i)T$ (for some $0 \leq \lambda \leq 1$). At this point in time, for any action $j \in [n]$ for the learner, we can write

$$\frac{1}{t}\sum_{s=1}^{t}u_L(x_t, j) = (1-\lambda) u_L(\ox_{i-1}, j) + \lambda u_L(\ox_i, j).$$

Since segment $i$ is non-degenerate, $\BR_L(\ox_{i-1}) = \BR_L(\ox_i) = \{b_i\}$, so there exists some constant $g_i$ (depending only on $\tau$ and not on $T$) such that $u_L(\ox_{i-1}, b_i) > u_L(\ox_{i-1}, b') + g_i$ and $u_L(\ox_{i}, b_i) > u_L(\ox_{i}, b') + g_i$ for any $b' \neq b_i$. It follows that for any $b' \neq b_i$,

$$\frac{1}{t}\sum_{s=1}^{t}u_L(x_t, b_i) > g_i + \frac{1}{t}\sum_{s=1}^{t}u_L(x_t, b').$$

Since $t$ scales with $T$, for large enough $T$, we will have $g_i > \gamma(t)$. The mean-based property then implies that the learner must put at least $(1 - n\gamma(t)) = 1 - o(1)$ weight on action $b_i$. It follows that each non-degenerate segment $(x_i, t_i, b_i)$ contributes a $t_i(x_i \otimes b_i)$ term to the limiting CSP as $T \rightarrow \infty$. But this is exactly the term this segment contributes to $\Prof(\tau)$. Since all but an $\eps$-fraction of $\tau$ is comprised of non-degenerate segments, this means the limiting CSP $\csp$ in $\M$ has the property that $||\csp - \Prof(\tau)||_{1} \leq 2\eps$. 
\end{proof}

\subsection{Proof of Lemma \ref{lem:ftrl_properties}}

\begin{proof}[Proof of Lemma \ref{lem:ftrl_properties}]
    We prove these properties in order.
    
    First, for any time horizon $T$, the definition of $\ftrl$ ensures that $\left|\frac{R_T(y)}{\eta_T}\right| \le g(T)$ for some monotone sublinear function $g(T) = o(T)$. At round $t$, consider any action $j$ that is at least $\sqrt{\frac{2g(T)}{T}}$ worse on average than the historical best action $j'$. Assume, for contradiction, that the distribution $y_t$ picked by the algorithm puts mass at least $\sqrt{\frac{2g(T)}{T}}$ on action $j$, i.e. $y_{t,j} \ge \sqrt{\frac{2g(T)}{T}}$. An alternative distribution $y'_t \in \Delta^n$ can thus be constructed from $y_t$ by setting $y'_{t,i} = y_{t,i}$ for all $i \notin \{j,j'\}$, $y'_{t,j} = y_{t,j} - \sqrt{\frac{2g(T)}{T}}$ and  $y'_{t,j'} = y_{t,j'} + \sqrt{\frac{2g(T)}{T}}$.  It is easy to see that $y'_t$ obtains lower value on the objective minimized by $y_t$, leading to a contradiction. Thus, the finite-time algorithm is mean-based with $\gamma(T) =\sqrt{\frac{2g(T)}{T}} $.

    The regret upper bound follows from standard regret properties of FTRL algorithms with strongly convex regularizers (see \cite{hazan201210}), while the lower bound is proved for all FTRL algorithms by~\citet{gofer2016lower}.

    The third property follows directly from the definition of $\ftrl_T(\eta, R)$, for the function $\wtA_T(U)$ given by

    \begin{equation}\label{eq:wta-def}
    \wtA_T(U) = \arg\max_{y \in \Delta_{n}} \left(\langle U, y\rangle - \frac{R(y)}{\eta_T}\right).
    \end{equation}

\noindent
    Note that since $y \in \Delta_{n}$, adding $\Delta$ to each component of $U$ increases $\langle U, y \rangle$ by a $y$-independent $\Delta$, and therefore does not affect the argmax.

    Finally, to show the last property, we will take

    \begin{equation}\label{eq:wtap-def}
    \wtA'_T(U') = \arg\max_{y' \in \Delta_{n-1}} \left(\langle U', y'\rangle - \frac{R'(y')}{\eta_T}\right),
    \end{equation}

    \noindent
    where $U' \in \mathbb{R}^{n-1}$ and $R'(y') = R(\pi(y))$ (recall $\pi$ is the embedding from $\Delta_{n-1}$ to $\Delta_{n}$ that ignores the $j$th coordinate). For a fixed $U$ (and corresponding $U'$ obtained by removing coordinate $j$), let $\Phi(y) = \langle U, y\rangle - R(y)/\eta_T$ and $\Phi'(y') = \langle U', y'\rangle - R'(y')/\eta_T$. Let $y_{\opt}$ and $y'_{\opt}$ be the maximizers of $\Phi$ and $\Phi'$ respectively (i.e., the values of $\wtA_T(U)$ and $\wtA'_{T}(U')$). Finally, let $r(\eps)$ be the maximum value of $R(y_1) - R(y_2)$ over all pairs $(y_1, y_2) \in \Delta_n \times \Delta_n$ with $||y_1 - y_2||_1 \leq \eps$. Since $R(y)$ is bounded and convex, $r(\eps)$ must go to zero as $\eps$ approaches $0$.

    By construction, $R$ is $\alpha$-strongly-convex for some fixed $\alpha > 0$. Then $\Phi(y)$ is $(\alpha/\eta_T)$-convex, so for any $y \in \Delta_n$, if $||y - y_{\opt}||_2 \geq \delta$, then $\Phi(y_{\opt}) - \Phi(y) \geq \alpha\delta^2/\eta_T$. 
    
    On the other hand, consider the $z \in \Delta_{n}$ formed by taking $y_{\opt}$ and moving all the mass on coordinate $j$ to the coordinate $j^{*} = \argmax_{j^{*} \in [n]} U_{j^{*}}$. Since we have moved mass to a better coordinate, we have that $\langle U, z \rangle \geq \langle U, y_{\opt} \rangle$. But also, since $y_{\opt}$ has at most $\gamma(T)$ weight on action $j$ (by the mean-based property we established earlier), $||y_{\opt} - z||_{1} \leq 2\gamma(T)$,  and so $|R(y_{\opt}) - R(z)| \leq r(2\gamma(T))$. Together, this shows that $\Phi(y_{\opt}) - \Phi(z) \leq r(2\gamma(T))/\eta$.

    But since $z$ has no mass on coordinate $j$, $z = \pi(z')$ for some $z' \in \Delta_{n-1}$. From the definition of $\Phi'$, we further have that $\Phi'(z') = \Phi(z)$. Since $y'_{\opt}$ maximizes $\Phi'$, we have that $\Phi'(y'_{\opt}) \geq \Phi'(z')$.

    It follows that $\Phi(y_{\opt}) - \Phi(\pi(y'_{\opt})) = \Phi(y_{\opt}) - \Phi'(y'_{\opt}) \geq r(2\gamma(T))/\eta$. From the earlier statement about the strong convexity of $\Phi$, this implies that $||y_{\opt} - \pi(y'_{\opt})||_2 \leq \sqrt{r(2\gamma(T))/\alpha}$. Since $\gamma(T) = o(1)$ and $r(\eps)$ converges to zero as $\eps$ goes to zero, it follows that if we let $\sshift(T) = T\sqrt{r(2\gamma(T))/\alpha}$, then $||y_{\opt} - \pi(y'_{\opt})|| \leq \sshift(T)/T$, and $\sshift(T) = o(T)$, as desired.

\end{proof}

\section{The mean-based menu and no-regret menu may differ}\label{app:mb_nr_gap}

The mean-based menu contains all CSPs asymptotically reachable against some mean-based algorithm. Since there are many no-regret mean-based algorithms, it is natural to wonder if perhaps the mean-based menu $\M_{MB}$ is simply equal to the no-regret menu $\M_{NR}$: that is, perhaps any CSP $\csp$ with no-regret is achievable as the profile $\Prof(\tau)$ of a valid trajectory $\tau$.

In this Appendix, we show that this is not the case. In particular, we prove the stronger statement that for the choice of $u_L$ in \eqref{eq:game}, no mean-based algorithm has $\M_{NR}$ as its asymptotic menu.

\begin{lemma}
\label{lem:nr_ne_mw}
    For the choice of $u_L$ in \eqref{eq:game}, the asymptotic menu of any mean-based no-regret algorithm $\A$ is a strict subset of $\M_{NR}$. In particular, $\M_{MB} \subset \M_{NR}$. 
\end{lemma}

\begin{proof}
Consider any mean-based no-regret algorithm $\A$. We must show there exists a CSP $\csp$ in $\M_{NR}$ that cannot be in $\M(\A)$. Before doing this, we will prove a useful intermediary result: any extreme point of $\M(\A)$ which contains both the actions $A$ and $C$ must also contain the action $B$. Assume for contradiction that this is not the case. Then, there is an extreme point $p$ which contains $A$ and $C$ but not $B$. Consider some trajectory $\tau$ which achieves $p$ against $\A$ (this must exist because $p$ is an extreme point). Because both $A$ and $C$ are played a constant fraction of the time, there is some constant that can be picked such that, after that time, $A$ and $C$ both must still be played, and therefore there is a switch from $A$ to $C$ or from $C$ to $A$. But during that time, $B$ needs to be played.

Now we can turn back to our main proof. Let the CSP which we claim is not in $\M(\A)$ be $\gamma = (\frac{1}{3}(A,X), \frac{1}{3}(A,Y), \frac{1}{3}(C,Y))$. Assume for contradiction that in fact $\gamma \in \M(\A)$. Then $\gamma$ can be decomposed into extreme points of $\M(\A)$. But none of these extreme points can contain both $A$ and $C$, as proven above. Therefore, one of the extreme points is $CY$, and the remaining extreme points have $(\frac{1}{2}AX, \frac{1}{2}AY)$ as a convex combination. However, $(\frac{1}{2}AX, \frac{1}{2}AY)$ has positive regret (the learner would rather plan $B$). Thus $(\frac{1}{2}AX, \frac{1}{2}AY)$ cannot be within $\M(\A)$, and therefore is not the convex combination of extreme points of $\M(\A)$. This is a contradiction. Therefore, $\gamma \notin \M(\A)$ for any mean-based no-regret algorithm $\A$, completing our proof.
\end{proof}
\section{Beyond utility-specific learners}
\label{sec:game_agnostic}

\subsection{Global Learning Algorithms}
In all previous sections of this paper, we have assumed that the learner has access to full information about their own payoffs at the start of the game. However, we can also think about them as committing to a mapping from learner utilities $u_{L}$ to learning algorithms. We will refer to this mapping as a \emph{global} learning algorithm, and extend our definitions accordingly. When necessary for clarity, we will refer to the standard learning algorithms discussed previously in this paper as \emph{utility-specific} learning algorithms.

\begin{definition}[Global Learning Algorithm]
A global learning algorithm $G$ is a mapping from learner utilities $u_L$ to utility-specific learning algorithms $\A$.
\end{definition}

\begin{definition}[No-Regret Global Algorithm]
A global learning algorithm $G$ is no regret if, for all $u_L$, $G(u_L)$ is a no-regret algorithm.
\end{definition}

\begin{definition}[No-Swap Regret Global Algorithm]
A global learning algorithm $G$ is no swap regret if, for all $u_L$, $G(u_L)$ is a no swap regret algorithm.
\end{definition}

\begin{definition}[Strictly Mean-Based Global Algorithm]
A global learning algorithm $G$ is strictly mean-based if, for all $u_L$, $G(u_L)$ is strictly mean-based.
\end{definition}

While we have previously represented $u_L$ implicitly in our notation for learner payoffs, $u_L$ is unknown to our global learning algorithm $G$, and thus we will instead write $V_{L}(\cA, u_L, u_O)$. We can now define Pareto-optimality and Pareto-domination over this larger space:

\begin{definition}[Asymptotic Pareto-Dominance for Global Learning Algorithms]
A global learning algorithm $G_{1}$ \emph{asymptotically Pareto-dominates} another global learning algorithm $G_{2}$ if, for a positive measure set of $u_{L}$, there is a positive measure set of $u_O$ such that $V_{L}(G_{1}(u_{L}),u_L,u_O) > V_{L}(G_{2}(u_{L}), u_L, u_O)$. Furthermore, for all remaining $u_L$ and $u_O$, $V_{L}(G_{1}(u_{L}),u_L,u_O) \geq V_{L}(G_{2}(u_{L}), u_L, u_O)$. A global learning algorithm is Pareto-optimal if no global learning algorithm Pareto-dominates it.
\end{definition}

There is a tight relationship between Pareto-domination as we have discussed it previously and Pareto-domination over the class of global learning algorithms:

\begin{lemma} \label{lem:global-equivalency}
A global learning algorithm $G$ is asymptotically Pareto-dominated if and only if there exists a positive measure set of $u_L$ such that for each $u_L$, $G(u_L)$ is Pareto dominated by some learning algorithm $\A$.
\end{lemma}
\begin{proof}
First, we will prove the forwards direction. Consider some global learning algorithm $G_{2}$ that is asymptotically Pareto-dominated by some global learning algorithm $G_{1}$. Then, by definition, $G_{1}$ Pareto-dominates $G_{2}$ on a positive meausure set of games. Thus for a positive measue set of $u_L$, $G_{2}(u_L)$ is Pareto dominated by the learning algorithm $G_{1}(u_L)$.   

Now, we can prove the backwards direction. Consider some global learning algorithm $G_{2}$ which has some positive measure set of of $u_L$ such that $G_{2}(u_L)$ is Pareto dominated by some learning algorithm $\A$. Let $U^{*}$ be the set of learner utilities such that this Pareto domination occurs, and let $f: U^{*} \rightarrow \A$ be a mapping from these utilities to the learning algorithm that Pareto-dominates $G_{2}(u_L)$. 

Then, we can define a global algorithm $G_{1}$ that takes as input $u_L$ as follows: 
\begin{itemize}
\item If $u_L \in U^*$, behave according to $f(u_L)$.
\item Else, behave according to $G_{2}(u_L)$.
\end{itemize}

We claim that $G_{1}$ Pareto-dominates $G_{2}$. To see this, note that $U^*$ is a positive measure set, so for a positive measure set of $u_L$, $G_{1}(u_L)$ Pareto-dominates $G_{2}(u_L)$. Furthermore, for all $(u_L,u_O)$ pairs, if $u_L \notin U^*$, the algorithms will behave exactly the same (and thus their performance will be equivalent). Finally, if $u_L \in U^*$, by the promise of Pareto-domination of the utility-specific learning algorithm $\A$, there is no $(u_L,u_O)$ pair such that $G_{2}(u_L)$ outperforms $G_{1}(u_L)$.  
\end{proof}

Combining Corollary~\ref{eq:perturbed_game} and Lemma~\ref{lem:global-equivalency}, we attain the following corollary:

\begin{corollary}
    All strictly mean-based global algorithms are Pareto-dominated.
\end{corollary}

And combining Corollary~\ref{eq:perturbed_game}, Theorem~\ref{thm:unique_nsr_menu} and Corollary~\ref{cor:NSR_pareto_optimal}, we attain the following corollary:

\begin{corollary}
    All no swap regret global algorithms are Pareto-optimal.
\end{corollary}

Thus, our main results extend easily to the global setting. However, restating the problem in this way allows us to explore an important subset of global learning algorithms: those which do not directly utilize information about $u_L$.

\subsection{Game-Agnostic Algorithms} 
Both the class of strictly mean-based algorithms and the class of no-swap regret algorithms include many well-known algorithms that do not explicitly map from $u_L$ to a utility-specific learning algorithm. Multiplicative weights, for example, receives information about $u_L$ only in the form of counterfactual payoffs after each round. 

We call algorithms with this level of information about $u_L$ \emph{game-agnostic} algorithms. We will exploit the freedom of dependence that global algorithms can have on $u_L$ in order to explicitly define these game-agnostic algorithms as a subclass of global algorithms.

\begin{definition}[Game-Agnostic Algorithm]
A game-agnostic algorithm is a global learning algorithm which cannot directly utilize information about the initial input $u_L$, or see the number of actions of the optimizer. Instead, this algorithm can only change its state according to the counterfactual payoffs for their actions after each round.
\end{definition}

Game-agnostic algorithms utilize the standard informational model of the online learning setting, in that they only learn about their own utilities counterfactually on a round-per-round basis. We can now define a weaker condition for Pareto-optimality, restricted to the comparison class of game-agnostic algorithms: 

\begin{definition}[Asymptotic Pareto-dominance over Game-Agnostic Algorithms]
A game-agnostic learning algorithm $G$ is \emph{asymptotically Pareto-optimal over Game-Agnostic Algorithms} if it is not asymptotically Pareto-dominated by any game-agnostic learning algorithm. 
\end{definition}

\begin{lemma}
If a game-agnostic algorithm $G$ is Pareto-optimal, it is also Pareto-optimal over game-agnostic algorithms.
\end{lemma}
\begin{proof}
Consider some game-agnostic algorithm $G$ which is Pareto-optimal over all $u_L$. Then, it is not asymptotically Pareto-dominated by any global learning algorithm. This includes the set of all game-agnostic learning algorithms. Therefore, $G$ is not Pareto-dominated by any game-agnostic learning algorithm, completing our proof.
\end{proof}

Notably, within the class of no-swap-regret algorithms, which we have proven to be Pareto-optimal, there are many game-agnostic learning algorithms, such as Blum-Mansour and TreeSwap \cite{blum2007external, dagan2023external, peng2023fast}.

\begin{corollary}
All game-agnostic learning algorithms that guarantee sublinear swap regret are Pareto-optimal over the space of all game-agnostic learning algorithms. 
\end{corollary}

While we are able to extend our Pareto-optimality results to the game-agnostic setting, we leave as a open question whether $\ftrl$ algorithms are still Pareto-dominated in this general sense. Notably, while many $\ftrl$ algroithms are game-agnostic, the algorithm that we prove Pareto-dominates these algorithms in Corollary \ref{cor:coup_de_grace_mw} is not game-agnostic itself. In particular, the algorithm crucially uses $u_L$ in order to determine which points to ``delete'' from the mean-based menu.


\end{document}